\documentclass[final]{article}
\pdfoutput=1
\usepackage[utf8]{inputenc}
\usepackage{jheppub}
\usepackage[export]{adjustbox}
\usepackage{amsmath,amsthm,amsfonts,amssymb,amscdx,mathrsfs}
\usepackage{scalerel}
\usepackage{multirow}
\usepackage{longtable}
\usepackage{mathtools}
\usepackage{upgreek}
\usepackage{subcaption}
\usepackage{tikz}
\usepackage{physics}
\usetikzlibrary{calc}		
\usetikzlibrary{math} 
\usetikzlibrary{intersections} 

\DeclareMathAlphabet{\mathpzc}{OT1}{pzc}{m}{it}
\DeclareMathAlphabet\mathsf{OT1}{lcmss}{m}{n}
\usepackage[mathscr]{eucal}
\newcommand{\sets}[1]{\mathscr{#1}}  

\newcommand{\be}{\begin{equation}}
\newcommand{\ee}{\end{equation}}

\newenvironment{sproof}{
\proof}{\endproof}

\newcommand{\I}{\mathcal{I}}  
\newcommand{\M}{\mathcal{M}}  
\newcommand{\N}{\mathcal{N}}  

\DeclareMathOperator{\area}{area}
\DeclareMathOperator*{\infp}{inf\vphantom{p}}

\newcommand{\pwset}{\mathscr{P}{([\mathsf{N}])}\setminus \{\emptyset\}}

\newcommand{\bkslice}{\sigma}  
\newcommand{\bdyslice}{\Sigma}  
\newcommand{\bksliceset}{\sets{S}}  

\newcommand{\ts}{\tau} 
\newcommand{\tsset}{\sets{T}} 
\newcommand{\pts}[2]{\ts_{#1}^{(#2)}} 

\newcommand{\surf}{\gamma}  
\newcommand{\surfset}{\Gamma} 
\newcommand{\newsurfset}{\Omega_{\tau}} 
\newcommand{\minimax}{\gamma^{+}} 
\newcommand{\newminimax}{\tilde{\gamma}^{+}} 
\newcommand{\HRT}{{\surf_{_{\text{HRT}}}}}  
\newcommand{\psurf}[2]{\surf_{#1}^{(#2)}} 
\newcommand{\pminimax}[2]{\gamma^{+(#2)}_{#1}}
\newcommand{\rmms}{\tilde{\gamma}^{+}} 
\newcommand{\rmmt}{\tilde{\ts}} 
\newcommand{\nmms}{\hat{\gamma}} 
\newcommand{\nmmt}{\hat{\ts}} 
\newcommand{\nmmr}{\hat{\mathcal{R}}} 

\newcommand{\conf}{\mathfrak{t}} 
\newcommand{\confset}{\mathfrak{T}} 
\newcommand{\coopset}{\mathfrak{C}} 
\newcommand{\mintset}{\mathfrak{M}} 

\newcommand{\hor}{\mathcal{H}}
\newcommand{\ew}{\mathcal{W}}

\newtheorem{theorem}{Theorem}[section]
\newtheorem{lemma}[theorem]{Lemma}
\newtheorem{corollary}[theorem]{Corollary}
\newtheorem{definition}{Definition}[section]
\newtheorem{conjecture}{Conjecture}[section]

\definecolor{brg}{RGB}{70, 255,70}

\preprint{BRX-TH-6723}

\title{\boldmath Minimax surfaces and the holographic entropy cone}

\author[a]{Brianna Grado-White,}
\author[a]{Guglielmo Grimaldi,}
\author[a]{Matthew Headrick,} 
\author[b]{and Veronika E.\ Hubeny}

\affiliation[a]{Martin Fisher School of Physics, Brandeis University, Waltham MA 02453, USA}
\affiliation[b]{Center for Quantum Mathematics and Physics (QMAP)\\ Department of Physics \& Astronomy, University of California, Davis CA, USA}
\emailAdd{bgradowhite@brandeis.edu}
\emailAdd{ggrimaldi@brandeis.edu}
\emailAdd{headrick@brandeis.edu}
\emailAdd{veronika@physics.ucdavis.edu}

\abstract{ We study and prove properties of the minimax formulation of the HRT holographic entanglement entropy formula, which involves finding the maximal-area surface on a timelike hypersurface, or time-sheet, and then minimizing over the choice of time-sheet. In this formulation, the homology condition is imposed at the level of the spacetime: the homology regions are spacetime volumes rather spatial regions. We show in particular that the smallest minimax homology region is the entanglement wedge. The minimax prescription suggests a way to construct a graph model for time-dependent states, a weighted graph on which min cuts compute HRT entropies. The existence of a graph model would imply that HRT entropies obey the same inequalities as RT entropies, in other words that the RT and HRT entropy cones coincide. Our construction of a graph model relies on the time-sheets obeying a certain ``cooperating'' property, which we show holds in some examples and for which we give a partial proof; however, we also find scenarios where it may fail.
\newline
\newline
A video abstract is available at \url{https://youtu.be/Ja-TPYNBujM}.
}

\begin{document}
\maketitle
\flushbottom
\section{Introduction}\label{sec:intro}

Holography has revealed a deep connection between spacetime geometry and quantum information. Entanglement, in particular, has played a central role, and it has been suggested that it is the fundamental ingredient from which spacetime itself emerges \cite{VanRaamsdonk:2010pw}. Many of these features are best understood in the classical limit of the AdS/CFT correspondence, where states of the boundary CFT are mapped to classical, smooth geometries of the bulk gravitational theory. However, many questions about this mapping between states and geometries remain. In particular, we would like to understand how to characterize which quantum states admit classical, geometric dual spacetimes.

An important step towards formulating such criteria was made by Ryu-Takayanagi (RT), who proposed that for static states in holographic systems, the entanglement entropy $S(A)$ of a boundary region $A$ was geometrically realized in the bulk as the area of the minimal surface homologous to $A$ \cite{Ryu:06b20v, Ryu:2006ef}. (It was later realized that it is enough for the state to be \emph{instantaneously static}, in other words the invariant Cauchy slice of a time-reflection symmetric spacetime; this is the sense in which we will use the term \emph{static} in this paper.) Early checks of the RT formula relied on showing that these areas satisfied known properties of entanglement entropy. For example, the strong subadditivity (SSA) inequality,
\be
 S(AB) + S(BC) \geq S(B) + S(ABC)\,,
\label{eq:ssa}
\ee
was proven for the static case in \cite{Headrick:2007km} (see also \cite{Headrick:2013zda} for a systematic study of properties satisfied by the RT formula). In addition to inequalities like SSA that are satisfied by general quantum states, static holographic states were shown to satisfy further inequalities, implying that not all quantum states admit classical geometric duals. In particular, it was proven that static states obeyed an inequality termed \textit{monogamy of mutual information} (MMI), 
\be
S(AB) + S(BC) + S(AC) \geq S(A) + S(B) + S(C) + S(ABC)\,,
\label{eq:mmi}
\ee
which can be violated by simple quantum states such as the 4-party GHZ state \cite{Hayden:2011ag}.

MMI, however, is not the only holography-specific inequality; refining the boundary subdivision increases the number and complexity of these inequalities; altogether there are in fact infinitely many of them \cite{Bao:2015bfa}. These inequalities are linear, and as a result they define a convex cone in the space of entropies. More precisely, given a quantum state with $\mathsf{N}$ parties $\{A_1,A_2,\dots,A_{\mathsf{N}}\}$, the entropies of all single and joint parties, $S(A_1),S(A_2),S(A_1A_2),\ldots$, form a real vector with $2^{\mathsf{N}}-1$ components. The $\mathsf{N}$-party \emph{quantum entropy cone} is the set of all such vectors realizable by some quantum state, while the $\mathsf{N}$-party static holographic entropy cone, or \emph{RT cone}, is the subset of such vectors realizable by the entropies of spatial regions of static holographic states. While much effort has been devoted to understanding the structure of the RT cone (see \cite{Grado-White:2024gtx} and references therein, in particular footnote 2, along with the more recent work \cite{Bao:2024vmy,Bao:2024azn,Czech:2025tds}), 
 its full characterization remains elusive, along with a qualitative understanding of its physical interpretation and implications.

A central tool in the study of the RT cone is the \emph{graph model}. This is a weighted graph with $\mathsf{N}+1$ external vertices, constructed out of a static holographic spacetime with $\mathsf{N}$ boundary regions specified, such that the entropy of any subset of the regions is given by the corresponding min cut on the graph. The graph model is obtained by decomposing the bulk slice along the RT surface of every single and joint boundary region. Each vertex corresponds to a connected bulk region (or cell) resulting from this partition.
In particular, $\mathsf{N}$ external vertices correspond to the homology regions of the boundary regions and the $(\mathsf{N}+1)$st represents that of the complementary boundary region, or purifier. Two vertices are connected by an edge if the respective bulk regions share part of an RT surface, with weight given by its area. Since every static spacetime has a graph model, and every weighted graph is the graph model of some spacetime, the RT cone is equal to the min cut cone for weighted graphs. The study of the RT cone is thereby distilled into a problem in graph theory.

Moving beyond the static regime, the RT proposal was soon made covariant by Hubeny-Rangamani-Takayanagi (HRT) \cite{Hubeny:2007xt}, who suggested that the entanglement entropy $S(A)$ in time-dependent settings is given by the area of the minimal extremal surface $\HRT(A)$ that is spacelike-homologous to $A$.\footnote{The spacelike homology condition can be phrased as the existence of a spacelike region $r(A) \subset \bkslice$ on a bulk Cauchy slice $\bkslice$ such that $\partial r(A) = \HRT(A) \cup A$.} An immediate question is whether these entropies obey the same set of inequalities as static holographic ones -- in other words, whether the $\mathsf{N}$-region HRT cone equals the $\mathsf{N}$-region RT cone. Note that all $\mathsf{N}$ given boundary regions must lie on a common Cauchy slice $\bdyslice$, so that their unions are also spatial regions. Their HRT surfaces, as well as those of their unions, then all lie within the Wheeler-DeWitt (WDW) patch of $\bdyslice$, but not necessarily on a common bulk Cauchy slice (unlike in the static case, where all surfaces necessarily lie on the static bulk Cauchy slice). Specifically, when boundary regions cross,\footnote{Regions $A,B$ in the boundary Cauchy slice $\bdyslice$ are said to cross if $A \cap B$, $A^c \cap B^c$, $A \cap B^c$, $A^c \cap B$ are all non-empty, where $A^c:=\bdyslice\setminus A$ and $B^c:=\bdyslice\setminus B$.} their HRT surfaces are not generically mutually achronal and therefore cannot lie on a common Cauchy slice; this fact will be important for us shortly.

Initial progress on the question of whether the HRT cone equals the RT cone came with the \emph{maximin} reformulation of the HRT formula, which made it possible to prove SSA and MMI
in the time-dependent context. The maximin formula is defined by first finding the globally minimal surface homologous to $A$ on each bulk Cauchy slice $\sigma$ containing the entangling surface $\partial A$, and then picking the largest area surface among these minimal surfaces:
\be 
\label{eq:maximin}
S(A) = \frac{1}{4G_{\rm N}}\sup_{\sigma}\infp_\gamma
|\gamma|\,,
\ee
where $|\surf|$ is the area of $\surf$. Importantly, it was shown that, for a nested pair of boundary regions, such as the regions $B$ and $ABC$ appearing in SSA \eqref{eq:ssa}, there exists a slice $\bkslice$ that contains both $\HRT(B)$ and $\HRT(ABC)$, and furthermore on which both surfaces are minimal. While $\bkslice$ does \emph{not} in general contain the HRT surfaces for the other regions $AB,BC$ appearing in the inequality, it does contain surfaces $\tilde\gamma(AB),\tilde\gamma(BC)$ homologous to those regions that are smaller (by a focusing argument) than $\HRT(AB),\HRT(BC)$ respectively. With this collection of surfaces defined on $\bkslice$, the static argument can then be applied. The same argument goes through for MMI. Since SSA and MMI are the only three-party inequalities obeyed by RT, we learn that the three-region HRT cone equals the three-region RT cone. Unlike in the static case, the proofs of SSA and MMI (as well as the equivalence between maximin and HRT, and consistency of HRT with boundary causality \cite{Headrick:2014cta}) rely on dynamical assumptions about the spacetime: the null energy condition, the Einstein equation, and AdS boundary conditions; indeed,  NEC violation generically leads to SSA violation \cite{Callan:2012ip}. Thus, while the RT cone is essentially a graph-theoretic object, the HRT cone connects directly to the physics of gravity and holography.

Unfortunately, when we go to more than three regions, the maximin proof strategy fails, as higher inequalities contain crossing pairs of regions on the smaller side of the inequality, where we need all HRT surfaces to be minimal on a common Cauchy slice. In \cite{Rota:2017ubr}, it was shown that the maximin strategy cannot be used to prove any entropy inequalities beyond SSA and MMI. Nonetheless, there has been a plethora of partial progress in characterizing the HRT cone:
\begin{enumerate}
    \item In \cite{Bao:2018wwd}, it was shown that for  black brane spacetimes in which the membrane theory of \cite{Mezei:2018jco} can be applied, in the late-time and large boundary region limit, all of the RT inequalities hold. 
    \item In \cite{Erdmenger:2017gdk,Caginalp:2019mgu}, the inequalities were numerically tested in 2+1 bulk dimensions and trivial topology.
    \item In \cite{Czech:2019lps}, it was proved that the inequalities hold in 2+1 dimensions for topologically trivial spacetimes. 
    \item In \cite{Grado-White:2024gtx}, the inequalities were tested numerically in 2+1 dimensions and non-trivial topology, and a proof was given for spacetimes with $\pi_1(\mathcal{M}) = \mathbb{Z}$.
    \item In \cite{Bousso:2024ysg}, a new proof strategy, using the original HRT formula rather than maximin, was used to re-prove SSA and MMI; it may be possible to apply it to higher entropy inequalities.
\end{enumerate}
All of these developments point to the likely equality of the HRT and RT cones.

A closely related question is whether a general time-dependent holographic state, with $\mathsf{N}$ boundary regions specified, admits a graph model, in the same sense as for RT: a weighted graph with $\mathsf{N}+1$ external vertices such that the entropy of any subset of these  regions is given by the corresponding min cut on the graph. If all the HRT surfaces are contained and minimal on a common Cauchy slice, then we can construct a graph model the same way as for RT. However, as noted above, this is generically not the case, so any general graph model is non-trivial to construct. In fact, it is not obvious if a graph model even exists. Since, as mentioned above, the RT cone equals the min cut cone, equality of the HRT and RT cones would be an immediate corollary of the existence of a graph model.

The graph model is also of interest beyond the issue of the entropy cones. By capturing all the entropies of the regions in a purely \emph{spatial} construct, it would imply that, as far as the entanglement structure of the state is concerned, the time direction can be ``integrated out.'' More precisely, in fixing the boundary regions we are fixing a boundary Cauchy slice, and so it is the bulk time within the corresponding WDW patch that is being integrated out. This is perhaps not so surprising, since we know that time evolution within a WDW patch is a mere gauge transformation in gravity, and so in principle should be irrelevant for physically well-defined quantities. Graph models are also related to tensor network models of holographic states \cite{Bao:2018pvs}, and so may help to clarify the problem of constructing tensor networks for time-dependent holographic states. On the other hand, if a graph model does \emph{not} in general exist, then that would tell us that the bulk time direction \emph{cannot} be integrated out in describing the entanglement structure of a holographic state; specifically, the fact that the HRT surfaces of crossing regions are not mutually achronal must be encoded in the values of the entanglement entropies themselves. For these reasons, we believe the existence of a covariant graph model is an important question about holography.

In this work, we will attempt --- but not entirely succeed --- to answer these questions about the graph model and HRT cone. To this end, we will focus on a reformulation of the HRT formula that is inspired by, but different from, maximin;
namely the \emph{minimax} prescription introduced in \cite{Headrick:2022nbe}. In this formulation, one first finds the maximal area surface on a given timelike or null hypersurface $\ts$ homologous to $D(A)$, dubbed \emph{time-sheet}, and then minimizes over time-sheets:
\be
S(A) =\frac{1}{4G_{\rm N}} \, \infp_\tau\sup_\gamma|\gamma|.
\ee
We prove several properties of minimax surfaces and time-sheets, a minimax time-sheet being one that contains the HRT surface and on which it is maximal.

For multiple regions, we define the notion of a \emph{cooperating} time-sheet configuration, as follows: Intersecting time-sheets cut each other up into partial time-sheets, and cut the surfaces they host into partial surfaces; a set of minimax time-sheets cooperates if every partial HRT surface is maximal on its partial time-sheet. The value of cooperating time-sheets is that, as we show, they define a graph model. We also show that, for a pair of crossing regions with connected HRT surfaces, there does exist a cooperating time-sheet pair. However, we have not been able to prove that a cooperating configuration exists for a general set of boundary regions, and our investigations of specific examples are inconclusive; while in some cases we can construct cooperating time-sheets, in others we can actually argue \emph{against} their existence. This leaves several logically possible outcomes for this story:
\begin{enumerate}
\item Cooperating time-sheets always exist, but one must be more clever than we were to construct them in specific examples.
\item Cooperating time-sheets do not always exist, but a different (perhaps more flexible) construction yields a graph model.
\item A graph model does not exist in general, but the HRT cone equals the RT cone for all $\mathsf{N}$ for some other reason.
\item The HRT cone does not equal the RT cone for all $\mathsf{N}$. 
\end{enumerate}
We leave the solution to this multiple-choice problem to future work.

The rest of the paper is organized as follows. In the remainder of this section we provide a self-contained summary of our main results. Section \ref{sec:max_to_min} gives a precise definition of the minimax prescription and shows that it agrees with HRT. In doing so, we discuss the stability of the surface in relation to its extremality. We also introduce yet another reformulation in the form of a relaxation of the main minimax prescription. In section \ref{sec:properties}, we prove several properties of minimax surfaces. Some of these are already known in the literature, but we provide a minimax-based proof for them, while others are new results. In section \ref{sec:conjecture} we define the cooperating property for sets of time-sheets and explain how to construct a graph model from it. We give supporting evidence for the existence of cooperating time-sheets, but also discuss scenarios where their existence may be challenged.\footnote{In appendix \ref{app:geodesic_seams}, we provide details on computing the spacelike and timelike geodesics used in the examples to test the conjecture.} We conclude in section \ref{sec:discussion} with a discussion of possible applications and future directions. 

\subsection{User-friendly summary}\label{sec:summary}

In this section we provide an  
overview of the main results; proofs and details are given in the main text. In what follows, the bulk spacetime will be denoted by $\hat{\M}$, the conformal boundary by $\hat{\N}$, and the future and past boundaries by $\hat{\I}^{+}$ and $\hat{\I}^{-}$ respectively (we use the hatted notation to avoid confusion with the regulated spacetime that we will use in the bulk of the paper, which will be introduced in subsection \ref{sec:regulator}). For a fixed boundary region $A$ on a boundary Cauchy slice $\bdyslice \subset \hat{\N}$,  the conformal boundary can then be decomposed as
(cf.\ figure \ref{fig:time-sheet-intersection})
\be
    \hat{\N} = D(A) \cup D(A^c) \cup J^{+}(\partial A) \cup J^{-}(\partial A)\,,
\ee
where $D(A)$ is the boundary domain of dependence of $A$, $A^c = \bdyslice\setminus A$ is the complement of $A$  and $J^{\pm}(\partial A)$ are the boundary causal future and past of the entangling surface $\partial A$.

\subsubsection*{Minimax}

\begin{figure}
    \centering
    \includegraphics[width=0.6\textwidth]{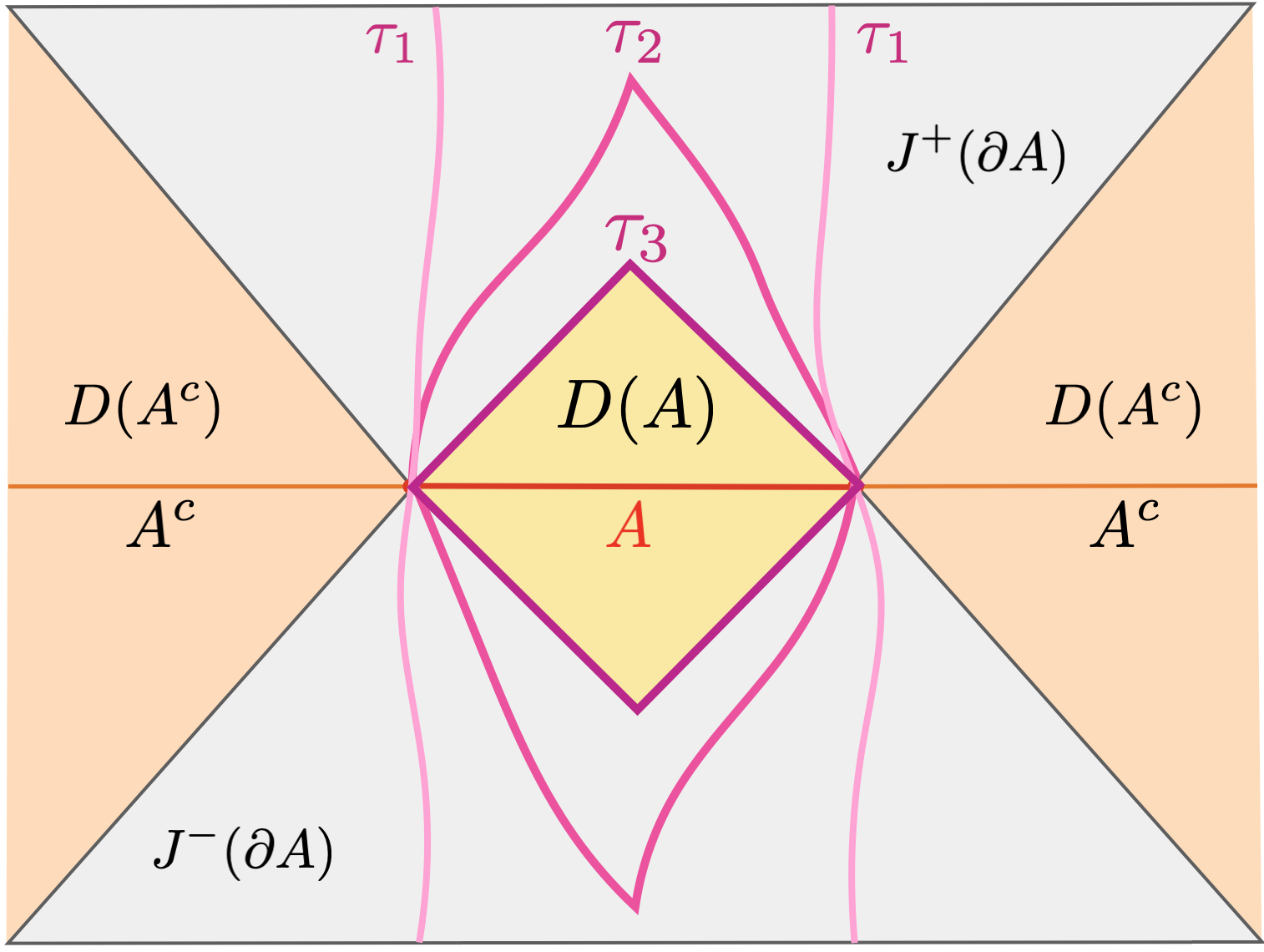}
    \caption{Illustration of the decomposition  $\hat{\N}=D(A) \cup D(A^c) \cup J^{+}(\partial A) \cup J^{-}(\partial A)$ of the conformal boundary for a given boundary region $A$, with the left and right edges understood to be identified. We also illustrate the relative homology condition by showing the intersection of three valid time-sheets $\tau_1, \tau_2$, $\tau_3$ with the boundary $\hat{\N}$: (i) $\tau_1$ is smooth, timelike, and continues all the way to $\hat\I^{\pm}$; (ii) $\tau_2$ is piecewise timelike and contains seams; (iii) $\tau_3$ is piecewise null and intersects $\hat{\N}$ exactly at the boundary of $D(A)$ (for example $\tau_3$ may be the boundary of the entanglement wedge). Notice that all the time-sheets contain $\partial A$ and do not intersect $D(A)$ or $D(A^c)$, as required by the relative homology condition.}
    \label{fig:time-sheet-intersection}
\end{figure}

\begin{figure}
    \centering
    \includegraphics[width=0.8\textwidth]{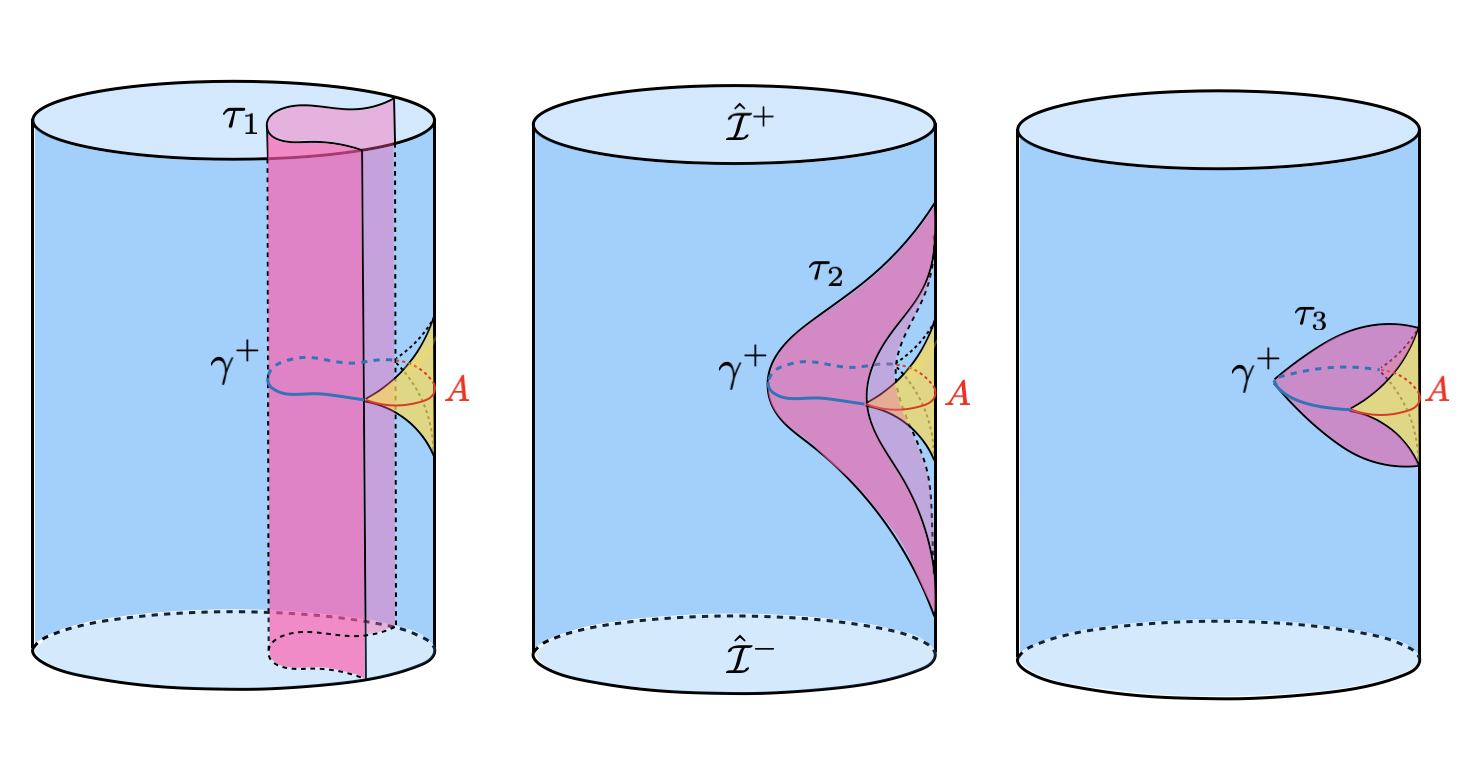}
    \caption{The time-sheets $\tau_1$, $\tau_2$, $\tau_3$ from figure \ref{fig:time-sheet-intersection} are shown from the bulk perspective, in separate panels. For all figures, the boundary subregion $A$ is indicated in red, its boundary domain of dependence $D(A)$ is indicated in yellow, the time-sheets are indicated in magenta, and the minimax surface $\minimax$ is indicated in blue. {\bf Left:} The time-sheet $\tau_1$ is smooth and ends on the conformal boundary, and past and future infinity. 
    {\bf Middle:} The time-sheet $\tau_2$ has a seam and fully ends on $\hat\N$. 
    {\bf Right:} The time-sheet $\tau_3$ is the boundary of the entanglement wedge $\mathcal{H}(A)$; it is a piecewise-null time-sheet with a spacelike seam, and ends on $\partial D(A)$.
    }
    \label{fig:time-sheets}
\end{figure}

In the minimax prescription \cite{Headrick:2022nbe}, the order of extremization is swapped relative to maximin. In doing this swap, the central geometrical construct shifts from Cauchy slices to \emph{time-sheets}. A time-sheet is defined as a piecewise timelike or null hypersurface homologous to $D(A)$ relative to $\hat{\I}^{+} \cup \hat{\I}^{-} \cup J^{+}(\partial A) \cup J^{-}(\partial A)$. In essence, this means that a time-sheet contains $\partial A$ and is allowed to end on $\hat{\I}^{+} \cup \hat{\I}^{-} \cup J^{+}(\partial A) \cup J^{-}(\partial A)$ on the boundary; see figures \ref{fig:time-sheet-intersection} and \ref{fig:time-sheets}  for examples of allowed possibilities. Consequently, there exists a codimension-0 bulk homology region $R(A)$ interpolating between the time-sheet $\ts$ and a portion of the boundary that includes all of $D(A)$ and none of $D(A^c)$. Given a boundary subregion, we call $\tsset$ the set of all valid time-sheets and $\bksliceset$ the set of all valid Cauchy slices, i.e.\ those containing $\partial A$. We can then define $\surfset_\ts$ as the set of all codimension-2 surfaces of the form $\gamma = \sigma \cap \ts$, for $\sigma \in \bksliceset$. The minimax prescription states 
\be
\label{eq:minimax0}
   S(A)  =  \frac{1}{4G_{\rm N}} \infp_{\ts \in \tsset} \sup_{\surf \in \surfset_\ts} |\surf|\,.
\ee
That is, on each time-sheet we find the maximal surface, and then pick the minimal one among 
choices of time-sheet. 
Following the notation of \cite{Headrick:2022nbe}, we denote the minimax surface by $\minimax$. Any time-sheet on which $\minimax$ is maximal is called a \textit{minimax time-sheet}.

\subsubsection*{Minimax = maximin = HRT}

A priori, the minimax prescription \eqref{eq:minimax0} could output a different value than maximin \eqref{eq:maximin}. However it was shown in \cite{Headrick:2022nbe} that the two prescriptions give the same value under standard holographic conditions (i.e.\ the Einstein field equations, null energy condition, and AdS boundary conditions). In the main text below, we will reiterate this argument and in addition show that minimax is in fact the minimal codimension-2 extremal surface in the correct homology class. More specifically, we will show:
\begin{itemize}
    \item $\area(\text{minimax})$ = $\area(\text{maximin})$, and
    \item minimax = HRT.
\end{itemize}
Note that while the equivalence between maximin and HRT has been previously shown \cite{Wall_2014}, and is  
invoked in the first argument, the second one does not rely on maximin and is thus more direct.

Leading up to the minimax = HRT result, we discuss the extremality of the minimax surface in relation to its stability, and show that a stable and extremal minimax surface always exists and that any unstable surface that may arise on some time-sheet is necessarily non-extremal. Furthermore, this stable and extremal surface can be shown to be the minimal one among all extremal surfaces, proving equivalence with the HRT prescription.

We also explore a relaxation of the minimax prescription, where instead of maximizing over codimension-2 surfaces that are achronal in the spacetime, we maximize over surfaces that are achronal just within the time-sheet. We call this the \textit{relaxed minimax}\footnote{This relaxed minimax prescription is not to be confused with the \emph{convex-relaxed} minimax prescription described in \cite{Headrick:2022nbe}.}
prescription. We show that when NEC is obeyed, the two prescriptions coincide. This relaxation of minimax will become a useful tool in our proof of entropy inequalities. 

\subsubsection*{Properties of minimax}
In section \ref{sec:properties}, we will prove various properties of minimax surfaces, some of which will be utilized in the following sections: 
\begin{enumerate}
    
    \item Minimality of entanglement wedge:
    As shown in \cite{Headrick:2022nbe}, the entanglement horizon $\hor(A)$ of the minimax surface $\minimax$ is a minimax time-sheet, i.e.\ a time-sheet on which the minimax surface is globally maximal. However, it is not the only one, as minimax time-sheets are highly non-unique. Among the set of all spacetime homology regions bounded by minimax time-sheets, we show that the entanglement wedge $\ew_A$ is the smallest one (i.e.\ it is contained in all the others). 

    \item Area bound:
    The causal holographic information upper-bounds the holographic entanglement entropy. A proof of this statement using maximin surfaces was first given in \cite{Wall_2014}. 

    \item Entanglement wedge nesting: Given two boundary regions $A$ and $B$ with $D(B) \subset D(A)$
    the entanglement wedge of $B$ is contained in the entanglement wedge of $A$, i.e.\ $\ew_B \subset \ew_A$. As a corollary of this statement, two minimax surfaces for the same boundary region must be spacelike related. 
    
    \item Pairwise cooperation:
    Whenever two minimax time-sheets intersect, they partition each other into what we call \textit{partial minimax time-sheets} (and consequently, the minimax surfaces are cut into partial minimax surfaces). We prove that there always exists a pair of minimax time-sheets $(\ts_1, \ts_2)$ such that each partial minimax surface is maximal on the respective partial minimax time-sheet.  
\end{enumerate}

The result above naturally raises the question of whether the cooperating property holds more generally, for multiple pairs of fixed time-sheets.
Specifically, we wonder if an arbitrary number of intersecting time-sheets can be arranged so that the partial minimax surfaces are maximal on their respective partial time-sheets. We motivate this \emph{cooperating conjecture} in section \ref{sec:conjecture} by introducing some of its powerful consequences, and further, we discuss its validity by considering explicit examples in 2+1 dimensions. 

\subsubsection*{The spacetime graph model} 

One of the key implications of the cooperating conjecture is that our minimax construction allows for the definition of a graph model for time-dependent states, which in turn implies the equivalence of the RT and HRT cones. The minimax prescription associates to a boundary domain of dependence $D(A)$ a spacetime homology region $R(A)$ bounded by the minimax time-sheet. Therefore, given a holographic spacetime $\M$ with $\mathsf{N}$ boundary regions specified, the minimax time-sheets of every joint boundary region define a partitioning of $\M$ into spacetime cells. We propose a graph model for time-dependent states where each such bulk cell corresponds to a vertex, and two vertices are joined with an edge if they share a part of a time-sheet. The edge weight is the area of the portion of minimax surface on the shared portion of the time-sheet. 

The only missing ingredient for this graph to be a working graph model of holographic entanglement is for the entropy of any subset of the regions to be given by its corresponding min cut on the graph. A sufficient condition for this to happen is if the cooperating property from the previous section holds for the entire collection of time-sheets in the spacetime, i.e.\ if every partial minimax surface on every partial time-sheet is maximal on the respective partial time-sheet. If this cooperating property holds, static and time-dependent holographic states share the same exact graph model of holographic entanglement. Any property obeyed by one is obeyed by the other, and in particular, the two entropy cones must be equivalent. 

Except for the case of two intersecting time-sheets, we do not have a general proof of the existence of cooperating time-sheets, which we leave as a conjecture. However, we show supporting evidence of its validity in subsection \ref{sec:favor} by doing direct analytic and numerical computations for a specific set of configurations of time-sheets in pure AdS$_3$ and in asymptotically AdS$_3$ with spherically symmetric NEC-obeying matter. In subsection \ref{sec:against}, we also present a scenario in pure AdS$_3$ where the existence of a cooperating configuration for three time-sheets is challenged. The setup entails an apparently frustrated configuration of nearly-intersecting but relatively-boosted HRT surfaces. In this extreme scenario cooperation appears difficult to achieve, suggesting that either a more clever construction is needed or the cooperating conjecture, as currently stated, may be false. These potential counterexamples occur in a topologically trivial $2+1$ dimensional bulk, where we know that the RT inequalities are obeyed \cite{Czech:2019lps}; thus, while they call into question the cooperating conjecture, they do not directly argue against the equality of the HRT and RT cones.
We leave further investigations to future studies.

\section{The minimax prescription}\label{sec:max_to_min}

In this section, we will review the minimax prescription \cite{Headrick:2022nbe} and prove its equivalence to two other covariant holographic entanglement entropy formulas, the original HRT formula and its maximin version.

Throughout this work we will take the bulk spacetime manifold $\mathcal{M}$ to be classical, smooth, and asymptotically AdS. We will also assume matter to satisfy the null energy condition, $T_{ab}k^a k^b \geq 0$ for any null vector $k^a$. Further, we will often impose a generic condition such that along any segment of any null geodesic with tangent vector $k^a$ there exists a point for which $R_{ab}k^a k^b \neq 0$. Together with the Raychaudhuri equation, these conditions will ensure the focusing of any  null congruence with initial non-positive expansion. Since we will frequently use this focusing argument throughout many proofs in the paper, we will refer to it as ``by focusing'' for convenience. 

This section is organized as follows. The first two subsections are largely based on \cite{Headrick:2022nbe}. In subsection \ref{sec:regulator}, we explain a choice of regulator to deal with ultraviolet divergences and in subsection \ref{sec:minimax} 
we introduce the minimax prescription and explain its relation to the HRT and maximin formulas. Subsections \ref{sec:stability} and \ref{sec:relaxation} consist of new material. In subsection \ref{sec:stability}, we define a notion of stability for minimax surfaces, and prove that such a surface is extremal if and only if it is stable. This gives a different way, without reference to maximin, of proving that the (stable) minimax surface is the HRT surface. Finally, in subsection \ref{sec:relaxation} we introduce a new protocol for computing minimax entropies, obtained as a relaxation of the main prescription, and we show that under standard holographic assumptions the two prescriptions coincide. This new reformulation will be important later in the paper.

\subsection{Entanglement wedge cross section regulator} \label{sec:regulator}

Generally, the entropies we are interested in may be infinite, due to UV divergences (and potentially IR divergences). Here and below, we will attempt to be careful and work only with UV-regulated quantities. In particular, we choose the entanglement wedge cross section (EWCS) as our regulator \cite{Dutta:2019gen}, which we now review. 

Let $A$ be a boundary region contained in a boundary Cauchy slice $\Sigma$. If the entangling surface $\partial A$ is non-empty, the entanglement entropy $S(A)$ will be divergent due to the UV modes across $\partial A$. In the bulk, this is manifested as the area of the HRT surface diverging as it approaches the boundary. We can regulate this divergence as follows. Denote by $B$ the complement of $A$ in $\Sigma$, with a small neighborhood removed around $\partial A$, such that $A$ and $B$ do not touch. Let $\surf^0$ be the bulk extremal surface homologous to the buffer, and therefore anchored to $\partial(AB)$ (or, if there is more than one, the surface closest to the buffer). Morally, $\surf^0$ is the HRT surface for $AB$; technically, however, defining it that way would make the definition circular. The \textit{entanglement wedge cross section}, denoted by $E_w(A:B)$, is the area of the minimal extremal surface homologous to $A$ relative to $\surf^0$.On the boundary, the fact that $A$ and $B$ do not touch means that there is no UV divergent entanglement between them. In the bulk, this is reflected in the fact that the extremal surface computing $E_w(A:B)$ ends on $\surf^0$; since it does not extend all the way to the boundary, its area is finite.

\begin{figure}
    \centering
    \includegraphics[width = \textwidth]{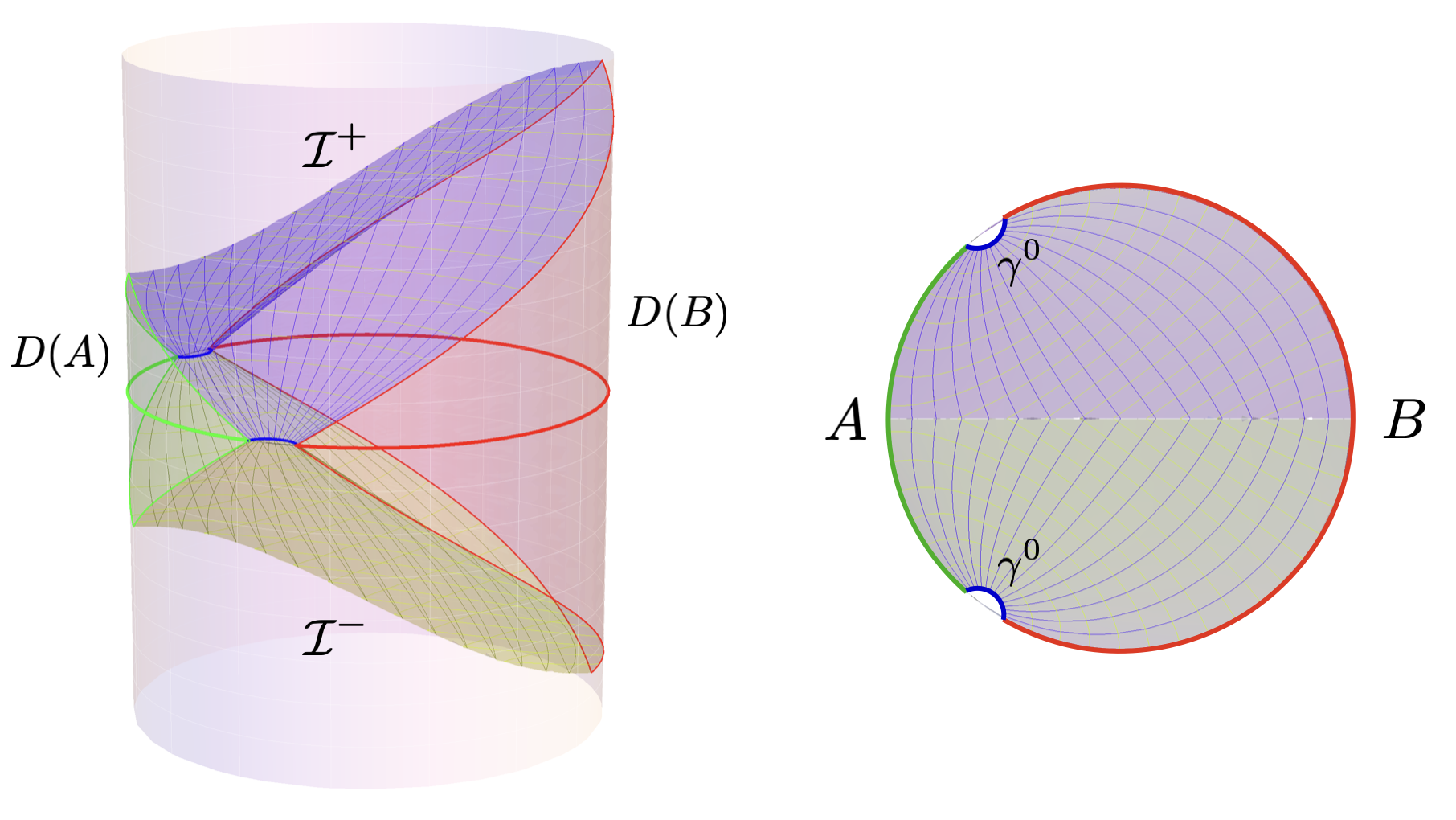}
    \caption{The regulated spacetime $\M$ for two regions $A$ and $B$. {\bf Left:} We present the full regulated spacetime, showing how it is bounded by ingoing null congruences emitted from $\surf^0$ (the HRT surface of $AB$ in the full spacetime). The remaining volume of the AdS cylinder (defined as the future and past of the two small spatial regions bounded by $\surf^0$) is removed. {\bf Right:} same spacetime but seen from the top, where we get a better view of $\surf^0$.}
    \label{fig:ewcs}
\end{figure}

As noted in \cite{Headrick:2022nbe}, using the EWCS as a regulator is equivalent to applying the HRT formula in the spacetime region defined by the $AB$ entanglement wedge $\hat{\mathcal{W}}(AB)$, with homology now computed relative to $\surf^0$. ($\hat{\mathcal{W}}(AB)$ is defined as the causal domain of the homology region interpolating between $AB$ and $\surf^0$.) This relative homology constraint can be thought of as analogous to the constraint arising from an end-of-the-world brane, in that the EWCS is allowed to end anywhere along $\surf^0$ (with the distinction that $\surf^0$ is codimension 2, while the end-of-the-world-brane is codmiension 1). We will take this perspective in the remainder of the paper, and simply ignore the portion of the spacetime outside of $\hat{\mathcal{W}}(AB)$. We use hatted symbols to denote quantities defined in the full spacetime (e.g.\ $\hat{\mathcal{I}}^\pm$), and reserve un-hatted symbols to denote quantities defined relative to the regulated spacetime. In the regulated spacetime, the conformal boundary $\N$ will be comprised solely of $D(A)$ and $D(B)$, and $\mathcal{I}^{\pm}$ will be given by the future/past boundaries of $\hat{\mathcal{W}}(AB)$. This regulation carries over analogously for more than two relevant boundary regions: given a partitioning of a boundary spatial slice into subregions $A_1,\ldots, A_\mathsf{N}$, plus a purifier $ A_{\mathsf{N} + 1}$, we remove small neighborhoods from this slice around each entangling surface, and discard any spacetime outside the joint entanglement wedge of the remaining portion of the boundary. We denote by $\gamma^0$ the joint ``HRT'' surface of $A_1,\ldots, A_\mathsf{N + 1}$, which is the intersection of $\mathcal{I}^+$ and $\mathcal{I}^-$. We show in figure \ref{fig:ewcs} the geometric construction of the regulated spacetime for two regions. 

For the purposes of the remainder of the paper, the upshot of this construction is that our boundary regions $A_i$ will be composed of entire connected components of the conformal boundary $\N = \bigcup_i D(A_i)$, and the bulk spacetime consists of the joint entanglement wedge of these regions. Of course, if the entangling surfaces of the regions we are interested in are empty (i.e.\ if they are made up of entire connected components of the boundary) then the ``regulated spacetime'' will simply be the original spacetime.

\subsection{Minimax surfaces}\label{sec:minimax}

We start with some basic definitions. Throughout this paper, a \emph{surface} is a spacelike codimension-2 submanifold of the bulk, and a \emph{slice} is a bulk Cauchy slice, which we take to be acausal, i.e.\ spacelike achronal. A surface is \emph{spacelike-homologous} to a boundary region $A$ if it is homologous to $A$ on some slice relative to $\gamma^0$. Such a surface is \emph{extremal} if it extremizes the area functional among such surfaces, which is true if and only if the trace of its extrinsic curvature vanishes everywhere and it meets $\gamma^0$ orthogonally. A surface $\surf$ is an \emph{HRT} surface for $A$ if it minimizes the area among extremal surfaces spacelike-homologous to $A$. In a completely generic spacetime, $A$ has a unique HRT surface. With some fine-tuning or symmetries, there may be multiple HRT surfaces (all with the same area, by definition). Nonetheless, for simplicity we will sometimes refer to ``the'' HRT surface $\HRT(A)$; if there are multiple ones, then the statements made hold for any of them. Its area is denoted $S(A)$. (From now on we work in units where $ 4G_{\rm N} =1$.) The causal domain of the homology region interpolating between $\HRT(A)$ and $A$ is the entanglement wedge $\ew(A)$, and the bulk part of its boundary is the entanglement horizon $\hor(A)$.

The \emph{maximin} formula \cite{Wall_2014} is
\be\label{maximin}
S_-(A):=\sup_\sigma\infp_{\surf}|\surf|\,,
\ee
where the $\sup$ is over slices $\sigma$ and the $\inf$ is over surfaces in $\sigma$ homologous to $\sigma\cap D(A)$ (relative to $\surf^0$). It was argued in \cite{Wall_2014} that the sup and inf in \eqref{maximin} are achieved, and that a maximin surface, subject to an addition stability criterion, is an HRT surface. This implies the existence of an HRT surface, as well as $S_-(A)=S(A)$.

\begin{figure}
    \centering
    \includegraphics[scale = 0.5]{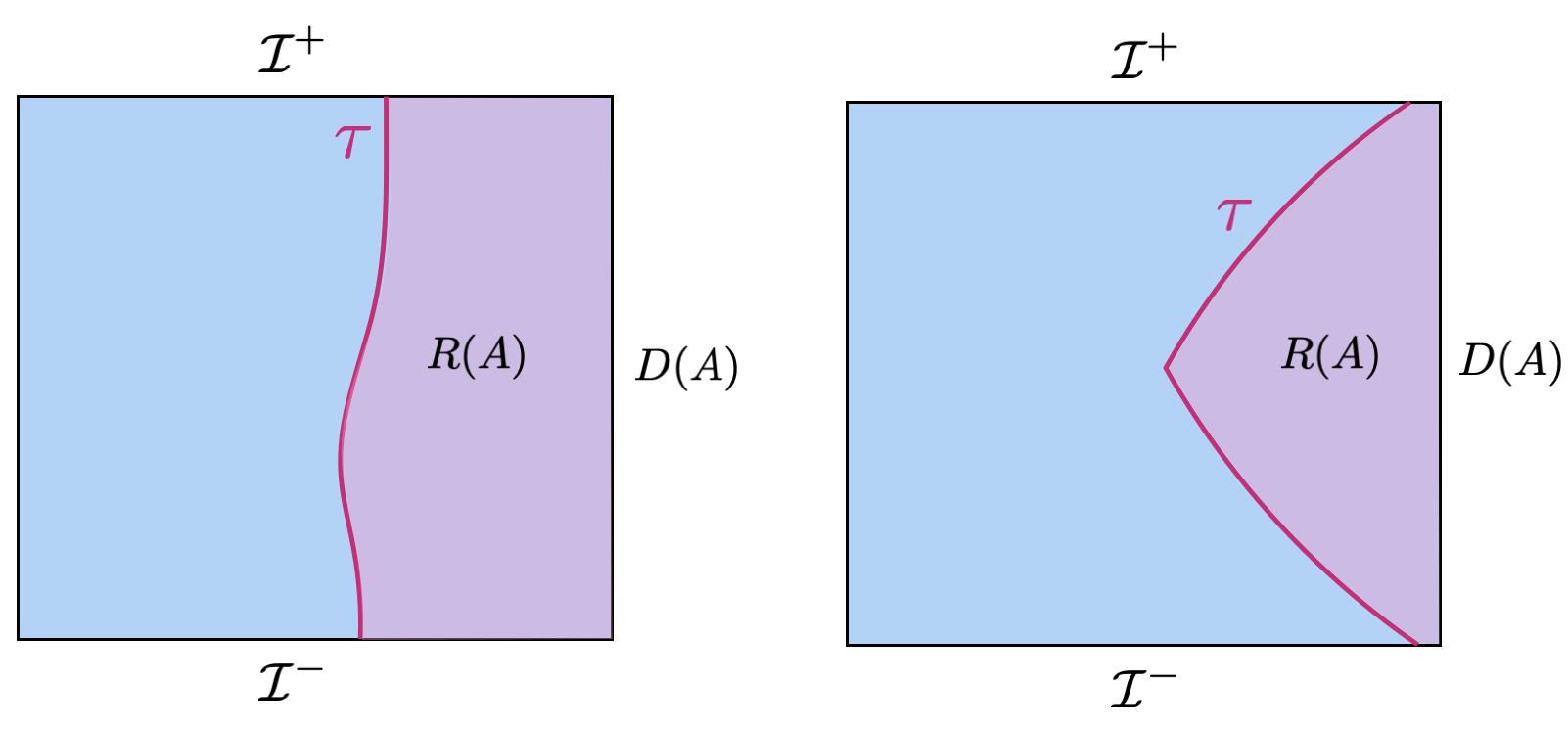}
    \caption{Illustration of two possible time-sheets in the regulated spacetime. We show them schematically as two-dimensional cross sections of the spacetime (such as those shown in figure \ref{fig:time-sheets}). {\bf Left:} A smooth timelike time-sheet homologous to $D(A)$ relative to $\I$. The codimension-0 homology region $R(A)$ is shaded in pink. {\bf Right:} The  entanglement horizon, a piecewise null time-sheet. Its homology region is the entanglement wedge, $R(A) = \ew_A$.
    }
    \label{fig:ts-regulated}
\end{figure}

To define minimax, we introduce the timelike counterparts of Cauchy slices, called time-sheets.
\begin{definition}[Time-sheet]
Given a boundary region $A$, a time-sheet for $A$ is a piecewise timelike or null hypersurface homologous to $D(A)$ relative to $\mathcal{I}^{+}\cup \mathcal{I}^{-}$, not intersecting $\N$. 
\end{definition}
By relative homology, we mean that the region $\mathcal{I}^{+}\cup \mathcal{I}^{-}$ does not count when computing homology, so that the time-sheet $\ts$ is allowed to end anywhere on this part of the boundary; in other words, there exists a codimension-0 homology region $R(A)$ that is bounded by $D(A)$, $\ts$, and an arbitrary subset of  $\mathcal{I}^{+}\cup \mathcal{I}^{-}$. It follows from the definition that, if $\surf^0$ is non-empty, then a time-sheet $\ts$ intersects it. 
We also note that a time-sheet is allowed to have discontinuous tangent plane. We refer to this as a seam, which may have spacelike, timelike, and null pieces.
A time-sheet may also be disconnected --- which can happen when the boundary region is itself made of disconnected components or when the HRT surface includes compact disconnected components, such as the bifurcation surface of a black hole. The set of possible time-sheets for a given boundary region $A$ will be denoted by $\tsset_A$ (or by $\tsset$ if the subregion is clear from context). See figure \ref{fig:ts-regulated} for an illustration of the time-sheets in the regulated geometry. We denote by $\tilde{\surfset}_\ts$ the set of codimension-2 surfaces of the form $\bkslice \cap \ts$ for some Cauchy slice $\bkslice$. 
For convenience of maximizing over a closed set we define the set $\surfset_\ts$ to be the closure of $\tilde{\surfset}_\ts$. Taking the closure allows the surface to have null pieces despite the fact that we defined Cauchy slices to be acausal.\footnote{The reader may ask whether we could instead have defined $\surfset_{\ts}$ to be the set of achronal codimension-2 surfaces directly, by defining Cauchy slice to be merely achronal rather than acausal (i.e.\ allowing it to have null pieces). The choices are equivalent except when $\ts$ is null, for example for the entanglement horizon for which any subset of the future (or past) is achronal. The supremum would thus be infinite, and so a time-sheet with null portions would not be a minimax time-sheet. However, we will find it convenient to retain the entanglement horizon as a minimax time-sheet, and thus restrict our surfaces to be in $\surfset_{\ts}$.
\label{achronal-issues-footnote}}

The paper \cite{Headrick:2022nbe} defined and studied the minimax quantity
\be\label{Splus}
   S_+(A)  := \adjustlimits \inf_{\ts \in \tsset_A} \sup_{\surf \in \surfset_\ts} |\surf|\,.
\ee
In this paper, we will be interested not only in the minimax area, but also in the surfaces and time-sheets yielding that quantity. We will assume that the sup and inf are achieved, so that there exist minimax time-sheets and minimax surfaces.\footnote{After theorem \ref{thm:minimax=HRT}, we will refine the definition of minimax surface to incorporate a natural stability condition.} This is a reasonable assumption for the following reason. The future and past boundaries of $\mathcal{M}$ are made up of singularities, Cauchy horizons, and/or the joint entanglement horizons of the boundary regions (if the spacetime has been regulated, as described in the previous subsection). Near these boundaries, the spatial components of the metric (or at least some of them) become small, while near the conformal boundaries the spatial components become large. Therefore, one would expect there to exist a saddle point of the area somewhere in the interior. It would be interesting to give a careful proof, under some assumptions, of the existence of minimax time-sheets and surfaces, but we leave this problem to future work.

For the rest of this section, we fix the boundary region $A$ and drop the dependence of each quantity on $A$, for brevity. Given a minimax pair $(\surf^+,\ts^+)$, surfaces in $\ts^+$ away from $\surf^+$ will typically have areas smaller than $\surf^+$. Therefore, $\ts^+$ can be slightly deformed away from $\surf^+$ while keeping the time-sheet minimax. This implies that minimax time-sheets are far from unique; in fact, away from minimax surfaces they are quite ``floppy.'' A useful time-sheet that is always minimax is the entanglement horizon. However, due to this ``floppiness,'' one would also expect to find minimax time-sheets that are smooth and everywhere timelike. In this paper, we make the assumption that one such everywhere timelike and smooth time-sheet exists. The issue of uniqueness of minimax \emph{surfaces} is more subtle, as we will discuss in the next subsection.

It was argued in \cite{Headrick:2022nbe}, assuming the existence of a maximin/HRT surface $\HRT$, that $\HRT$ is a minimax surface and $S_+=S$. First, by focusing, $\HRT$ is maximal on the entanglement horizon $\hor$, which is a time-sheet for $A$. So it is a candidate surface for the minimization. To show that it is indeed minimal amongst all maximal surfaces on the time-sheets in $\tsset$ we proceed by contradiction. Assume that there is a time-sheet $\ts'$ on which there exists a maximal surface $\surf'$ with smaller area than $\HRT$. Let $\sigma_{\rm max}$ be a slice on which $\HRT$ is minimal, and define $\surf'':=\sigma_{\rm max}\cap\ts'$; $\surf''$ is homologous to $\sigma_{\rm max}\cap D(A)$ via the homology region $\sigma_{\rm max}\cap R(A)$, and is no larger than $\surf'$, hence smaller than $\HRT$. But this contradicts the minimality of $\HRT$ on $\sigma_{\rm max}$.

\subsection{Stability and extremality of minimax surfaces}\label{sec:stability}

Above, we assumed the existence of an HRT surface. We can alternatively use minimax to \emph{prove} that one exists and is minimax, so that $S_+=S$. As always in this paper, we are assuming the existence of a minimax surface; the point here is that, for the reasons given in the previous subsection, this is intuitively clear in a well-behaved holographic spacetime, whereas the existence of an extremal surface seems a priori more mysterious. A naive argument is that a minimax surface extremizes its area against deformations in both the spatial and temporal normal directions and is therefore extremal. On the other hand, by focusing, every extremal surface is maximal on its ``entanglement horizon.'' So, a minimax surface is a least-area extremal surface, i.e.\ an HRT surface.

The problem with the above argument is that there can exist a minimax surface that, under an arbitrarily small deformation of the time-sheet in its neighborhood, ceases being maximal on it. This can easily be arranged starting from any minimax pair $(\surf^+,\ts^+)$, and wiggling $\ts^+$ far away from $\surf^+$ to make a new time-sheet $\ts^{\prime}$ with a surface $\surf^{\prime}$ fine-tuned to have the same area as $\surf^+$. The new surface $\surf^{\prime}$ is maximal on $\ts^{\prime}$ and therefore minimax, but is obviously not extremal.  A small deformation of $\ts^{\prime}$ near $\surf^{\prime}$ can decrease its area at first order, making $\surf^{\prime}$ no longer maximal on it; in other words, it is unstable.

In this section, we will make the notion of stability precise, prove that a stable minimax surface is extremal (and an extremal one is stable), and prove that any minimax time-sheet contains exactly one stable minimax surface. Together, these facts imply that an extremal surface exists. Combined with the fact that an extremal surface is maximal on its horizon, we have that a stable minimax surface exists and is HRT, and vice versa.

\begin{figure}
    \centering
    \includegraphics[width=0.65\textwidth]{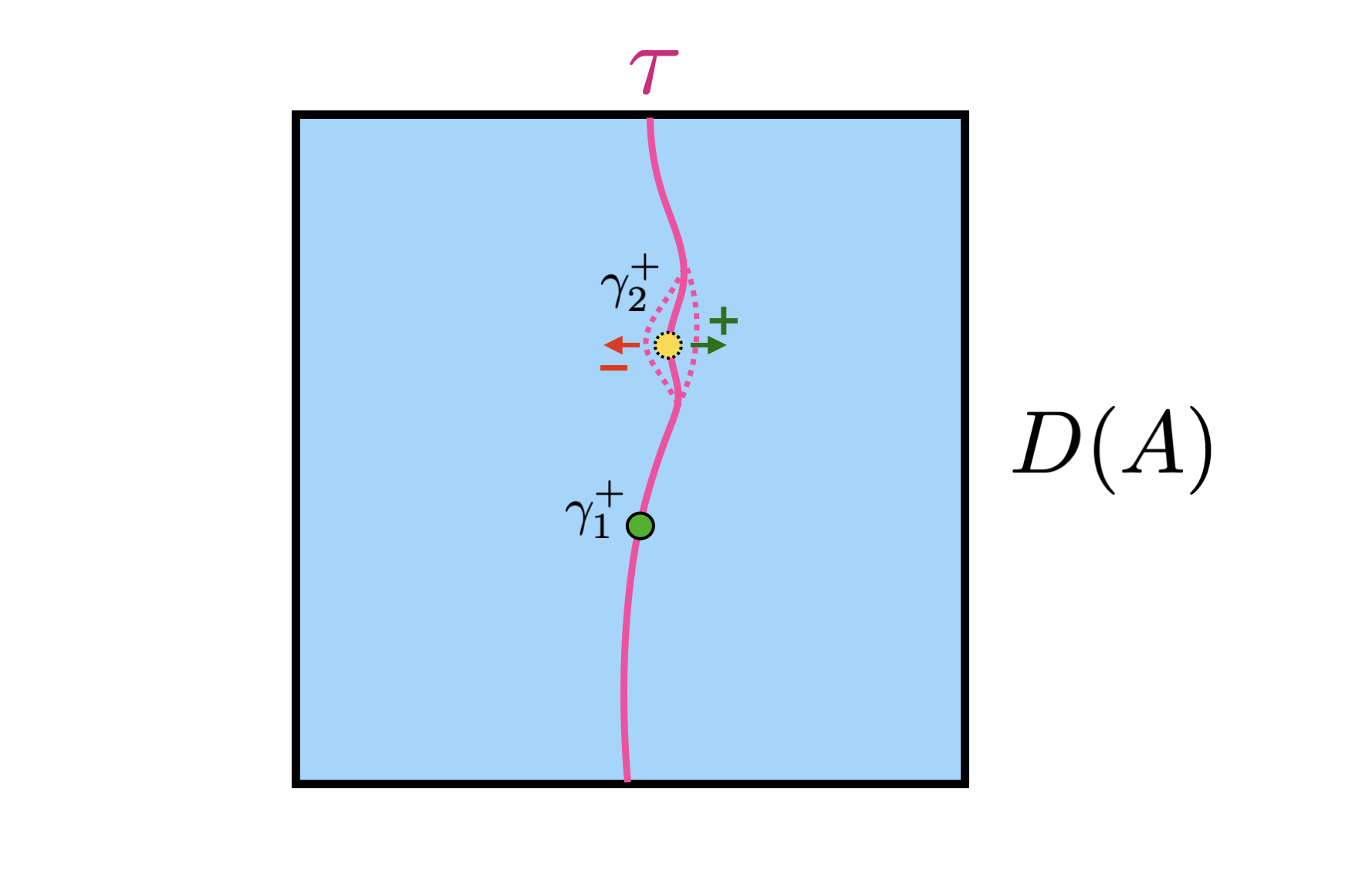}
     \caption{Suppose that $(\minimax_1,\ts)$ and $(\minimax_2,\ts)$ are both minimax, and that $\minimax_1$ is extremal while $\minimax_2$ is not. Then there exists a deformation of $\tau$ in a neighborhood of $\minimax_2$ that decreases its area, which makes $(\minimax_2,\ts)$ unstable. The opposite deformation increases its area, making $(\minimax_1,\ts)$ unstable. 
     However, $(\minimax_1, \ts')$ \emph{is} stable, where $\ts'$ is the time-sheet obtained by the first deformation, since $\minimax_1$ is the unique maximal surface on it. Since there exists such a time-sheet, we say that $\minimax_1$ is stable. Thus we see that stability and extremality are correlated, a relation that we prove in theorem \ref{thm:neitherboth}.}
    \label{fig:instability}
\end{figure}

\begin{definition}[Stable]
The minimax pair $(\minimax,\tau)$ is stable if, for any open neighborhood $\mathcal{N}$ of $\minimax$ and normal vector field $\eta$ on $\ts$, there exists an $\epsilon>0$ such that $\mathcal{N}$ contains a maximal surface on the time-sheet obtained by deforming $\ts$ by $\epsilon\eta$. The minimax surface $\minimax$ is stable if there exists a time-sheet $\tau$ containing $\minimax$ such that $(\minimax,\tau)$ is a stable minimax pair.
\end{definition}

Note that a given surface may be stable on one time-sheet and unstable on another; see figure \ref{fig:instability}. The example showcases another equivalent characterization of stable surfaces which is that a minimax time-sheet \textit{must} contain $\minimax$ (or, in case of multiple degenerate ones, one of them), whereas for an unstable surface $\surf$ one can merely \textit{choose} $\tau$ to pass through $\surf$.

We also note that a minimax surface $\minimax$ that is extremal is clearly stable. Extremality implies that there is no small time-sheet deformation that decreases the area of $\minimax$, and if $\minimax$ is the unique maximal surface on the time-sheet then there always exists a small enough $\epsilon$ for which the deformed time-sheet keeps $\minimax$ maximal on it. If instead $\minimax$ is not the unique maximal surface (sharing the time-sheet with unstable degenerate surfaces), we follow the instructions in figure \ref{fig:instability} to render it so. So extremality implies stability. In the remainder of this section, we will show that the converse holds.

In the following lemma we show how stability is only an issue when there are degeneracies on the time-sheet. However when that does happen, one is always guaranteed to find at least one stable minimax.
\begin{lemma}\label{lem:unique-is-stable}
At least one maximal surface on a minimax time-sheet is stable.
\end{lemma} 
\begin{proof}
If the unique maximal surface $\minimax$ is unstable, the time-sheet $\ts$ can be deformed to obtain a new maximal surface $\surf'$ far from $\minimax$. Then the difference 
\be
\delta = |\surf'| - |\minimax|
\ee
between $\minimax$ and the new maximal surface $\surf'$ can be made arbitrarily small for arbitrarily small $\epsilon$, contradicting the non-degeneracy. If there are degeneracies and all the maximal surfaces are unstable, the same argument applies: under an arbitrarily small deformation of $\ts$ all surfaces cease to be maximal, and the new maximal surface $\surf'$, outside $\mathcal{N}$, has area arbitrarily close to the $\minimax_i$.
\end{proof}

\subsubsection*{Stability $\iff$ extremality}

We have shown that at least one stable minimax surface always exists, but is it extremal? To answer this question, we begin with showing that a minimax surface can be deformed in all of its directions without obstructions.

\begin{lemma}\label{lem:no-null-segments}
    A minimax surface (i) is smooth and (ii) does not contain null-related points.
\end{lemma}
\begin{proof}
    (i) Let the minimax surface lie on some fully timelike and smooth time-sheet $\ts$, which, as explained at the end of subsection \ref{sec:minimax}, can be always assumed to exist. 
    So at least in the spacelike direction the surface is smooth. The only possibility then is that the surface is kinky as a submanifold on the time-sheet. However, the surface can then be smoothened out in a timelike direction to increase its area, contradicting its maximality. On the other hand, if the time-sheet was not smooth to start with, the maximal surface could inherit a kink. However, in such a case, the minimization over time-sheets step would render this configuration irrelevant, since a nearby time-sheet which resolves the kink would have a smaller-area maximal surface.

    (ii) The proof for this part follows standard arguments (see theorem 14 of \cite{Wall_2014} or claim 1 of \cite{Grado-White:2020wlb}, as well as the older but useful reference \cite{Galloway:1999ny} for details regarding the null congruences). Nonetheless, we quickly go through the argument to show how it gets modified in replacing Cauchy slices with time-sheets. First, we take care of the case where the null segment is itself part of the minimax surface. Let $\bar{\surf}$ be the null portion of $\minimax$. Deform $\minimax$ by applying a small deformation with support on a neighborhood of $\bar{\surf}$, to make it everywhere spacelike. Under this deformation, its area increases, contradicting the maximality of $\minimax$. Second, we consider the case where a null geodesic $\mathpzc{n}$ connects two points $x$ and $y$ on $\minimax$ through spacetime external to $\ts$. By shooting a future directed null congruence from points near $x$ and a past directed null congruence from points near $y$, the achronality condition on $\minimax$ imposes that the expansion of the former is larger than the expansion of the latter. Invoking focusing, this constrains the sign of the expansion at $x$ and a contradiction can be derived involving the mean curvature of $\minimax$ at the point $x$.
    \end{proof}

\begin{corollary}
A stable minimax surface $\minimax$ is extremal.
\end{corollary}
\begin{proof}
Lemma \ref{lem:no-null-segments} ensures that $\minimax$ can be deformed slightly in any direction while remaining achronal. Since it is maximal on the time-sheet $\ts$, it is extremal against deformations tangent to $\ts$. Since $\minimax$ is stable, under small deformations of the time-sheet the maximal surface moves continuously, in a direction not tangent to $\ts$. The maximal area is minimal and therefore extremal with respect to such deformations. Hence the area is extremal against arbitrary deformations.
\end{proof}

Lemma \ref{lem:unique-is-stable} ensures that a stable, and therefore extremal, minimax surface exists.  Now we discuss the case where there are multiple maximal surfaces on a given minimax time-sheet. We will use the following fact about extremal surfaces:

\begin{lemma}\label{not-timelike-related}
  If two minimax surfaces (for the same boundary region) share the same timelike minimax time-sheet, they are not both extremal.
\end{lemma}
\begin{proof}
Let $\minimax_1$ and $\minimax_2$ be the two minimax surfaces for the boundary region $A$, sharing the same timelike minimax time-sheet $\ts$. Suppose they are both extremal. Then, except on connected components where they might coincide, $\minimax_1$ and $\minimax_2$ must be disjoint (see corollary B.4 of \cite{Grado-White:2024gtx}).

Consider the entanglement wedges $\ew_{1}$ and $\ew_{2}$ of $\minimax_1$ and $\minimax_2$, whose boundaries $\hor_1$ and $\hor_2$ (the entanglement horizons) are the null hypersurfaces generated by orthogonally shooting a null congruence from $\minimax_1$ and $\minimax_2$ in the direction of $D(A)$. Since $\minimax_{1,2}$ are extremal, by focusing, they are the maximal area surfaces on $\hor_1$ and $\hor_2$ respectively. Next, consider the region
\be
    \mathcal{R} := \ew_1 \cap \ew_2,
\ee
obtained by intersecting the entanglement wedges of $\minimax_1$ and $\minimax_2$. The bulk part of the boundary $\partial \mathcal{R}$ of this region will be a null hypersurface made of components of $\hor_1$ and $\hor_2$. Because both $\hor_1$ and $\hor_2$ are homologous to $D(A)$, $\partial \mathcal{R}$ will be as well, by construction. Furthermore, $\partial \mathcal{R}$ will contain the codimension two surface $\minimax_3 = \hor_1 \cap \hor_2$ whose area must be smaller than both $\minimax_1$ and $\minimax_2$ by focusing. The surface $\minimax_3$ will also be maximal on $\partial \mathcal{R}$ since areas of cross sections decrease along $\hor_1$ and $\hor_2$. But this is a contradiction to the minimality of $\minimax_1$ and $\minimax_2$, since $\minimax_3$ is a maximal surface on an hypersurface homologous to $D(A)$, with less area than both $\minimax_1$ and $\minimax_2$.   
\end{proof}

The above result guarantees that, if there are multiple maximal surfaces on a given minimax time-sheet, at most one of them is extremal. On the other hand, by lemma \ref{lem:unique-is-stable}, at least one of them is stable, and thus extremal. Once again this exemplifies the characterization of stability, that a minimax time-sheet \emph{must} pass through an extremal (and therefore stable) surface, i.e.\ the existence of a stable/extremal surface is \emph{necessary} to support the  existence of a minimax time-sheet. We can sum the above results up in the following two theorems.

\begin{theorem}\label{thm:neitherboth}
    A minimax surface is either both extremal and stable or neither, and furthermore at least one stable and extremal minimax surface exists.
\end{theorem}
\begin{theorem}\label{thm:minimax=HRT}
    A stable minimax surface is an HRT surface.
\end{theorem}
\begin{proof}
Let $A$ be a boundary region. Let $\sets{E}_A$ be the set of extremal surfaces in the correct homology class for $A$. Since the minimax surface $\minimax$ has been proven to be extremal, we have
    \begin{equation}
        \sets{E}_A = \{\gamma_1,\gamma_2,\dots,\minimax\}.
    \end{equation}
    To show that the minimax surface coincides with the HRT surface we have to show that it is the least area surface in $\sets{E}_A$. Note that each one of the $\gamma_i$ lives on a time-sheet (the entanglement horizon) on which it is maximal. So clearly, by the minimization step of minimax, $\minimax$ is the least area surface in $\sets{E}_A$.
\end{proof}

Of course, these theorems do not preclude a region $A$ having multiple, degenerate minimax/HRT surfaces. However, these surfaces must be maximal on different time-sheets. In subsection \ref{sec:nesting} below, we will prove a stronger result: degenerate minimax surfaces must be fully spacelike-separated. The same result is also obtained from maximin, where it is shown that degenerate HRT surfaces share a common Cauchy slice \cite{Wall_2014,Grado-White:2024gtx}.

Unstable minimax surfaces are uninteresting for our purposes, so from now on by \emph{minimax surface} we mean \emph{stable minimax surface}.

\subsection{Relaxed minimax}
\label{sec:relaxation}

While the definition of minimax given above follows in natural analogy to maximin, we now define a modification that, while less symmetric with maximin, will be useful in proving later results. In particular, we will relax the achronality conditions in the maximization step. In the definition of minimax used in previous sections, we considered surfaces in $\surfset_\ts$ that are achronal in the full spacetime $\M$. Here, we instead consider the set of acausal surfaces \emph{within} the time-sheet $\ts$, regarded as a Lorentzian manifold with its induced metric. The maximization step will be performed over the closure of this set, which we call $\newsurfset$. (The closure is taken for the same reasons explained in footnote \ref{achronal-issues-footnote}
; in particular, mere achronality within a time-sheet would, for partly null time-sheets, permit the surface to cross the same generator arbitrarily many times, leading to a divergence in its area.) We call this new minimax prescription \textit{relaxed minimax}, defined as
\be\label{eq:relaxed-minimax}
     \tilde{S}_{+}(A) = \adjustlimits \inf_{\ts \in \tsset} \sup_{\surf \in \newsurfset} |\surf|.
\ee
As noted in subsection \ref{sec:minimax}, the definition of time-sheet implies that $\ts$ intersects $\gamma^0$ (if the latter is non-empty). A general acausal surface on $\ts$ does not necessarily intersect $\gamma^0$, but a maximal one will, so the relaxed-minimax surface does as well.

It is straightforward to see that the set $\newsurfset$ is larger than the original minimax set $\surfset_\ts$: let $\surf\in \surfset_\ts$, then $\surf = \ts \cap \bkslice$ for some Cauchy slice $\bkslice$. The slice separates the spacetime $\mathcal{M}$, so $\surf$ separates the time-sheet $\ts$ and therefore is in $\newsurfset$. The converse is clearly not true; it is easy to construct a time-sheet $\ts$ and a surface $\surf\in \newsurfset$ that does not lie on a Cauchy slice. Since this is a weaker constraint, we have
\be
    \sup_{\surf \in \newsurfset} |\surf| \geq \sup_{\surf \in \surfset_\ts}|\surf|
\ee
and hence
\be\label{eq:weaker-constraint}
    \adjustlimits \inf_{\ts \in \tsset}\sup_{\surf \in \newsurfset} |\surf| \geq \adjustlimits \inf_{\ts \in \tsset}\sup_{\surf \in \surfset_\ts}|\surf|\,.
\ee
So at first glance we might expect the resulting output to be a surface with greater area to the original minimax protocol. Here, however, we show that under usual holographic conditions (where NEC is obeyed), the two formulas coincide. In particular, we show that every minimax/HRT surface is a relaxed-minimax surface 
and vice-versa (so in particular, the two formulas output surfaces with the same area). 
\begin{lemma}\label{lemma:relaxed_equal_area}
The relaxed-minimax surface $\newminimax$ has the same area as the HRT surface $\minimax$,
\be
\tilde{S}_+=S_{+}. 
\ee 
Furthermore, the HRT surface is a relaxed-minimax surface. 
\end{lemma}
\begin{proof}
From \eqref{eq:weaker-constraint}, one always has that 
\be\label{eq:newsurf-geq-surf}
|\newminimax| \geq |\minimax|.
\ee
Now consider the HRT surface $\minimax$, and consider the minimizing time-sheet formed by its entanglement horizon $\hor$. Any other achronal-in-$\hor$ surface which intersects each generator at most once will have less area than the HRT surface, by focusing. Therefore, the HRT surface is the maximal spacelike codimension-2 surface on $\hor$. Thus, it is a valid candidate maximal surface for the relaxed minimax prescription. Therefore, we must have 
\be\label{eq:newsurf-leq-surf}
    |\newminimax| \leq |\minimax|.
\ee
Combining equations \eqref{eq:newsurf-geq-surf} and \eqref{eq:newsurf-leq-surf}, $|\newminimax| = |\minimax|$. Further, it follows from the above discussion that $\minimax$ is the maximal area surface in $\Gamma_{\hor}$, and has no more area than any other  $\tilde{\surf}\in \newsurfset$ on any other time-sheet. Thus, $\minimax$ is a  relaxed-minimax surface. 
\end{proof}

\begin{figure}
    \centering
    \includegraphics[width = 0.6\textwidth]{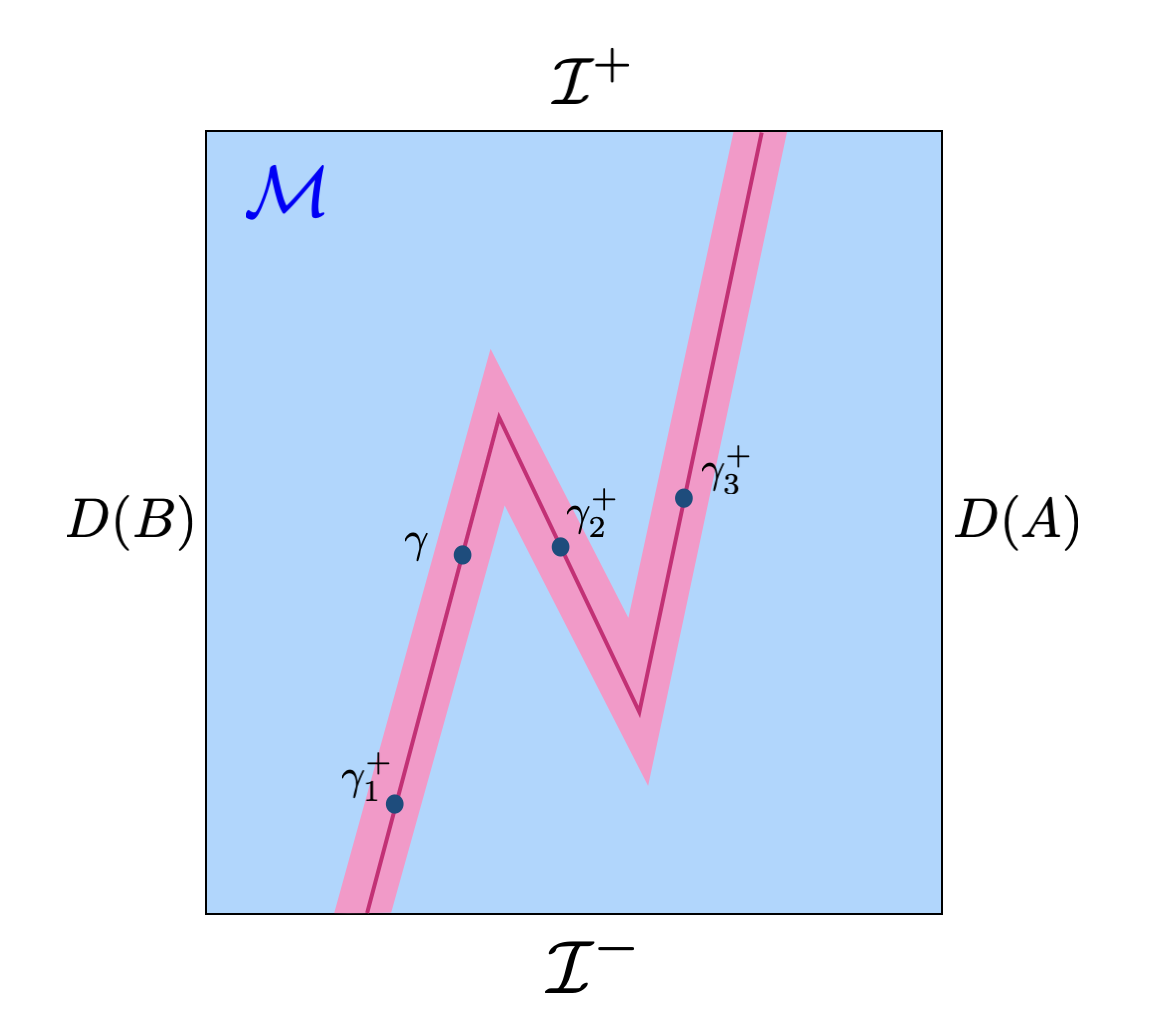}
    \caption{Penrose diagram of a NEC-violating spacetime where $S_{\text{relaxed-minimax}} \neq S_{\text{minimax}}$. The spacetime $\mathcal{M}$ is made of a sphere $\mathsf{S}^{d-2}$ fibered over Minkowski. Outside of the pink region, we choose the area of the sphere $\mathsf{S}^{d-2}$ to be very large, so any time-sheet outside of the pink region will not be a minimizing time-sheet and can be discarded. Time-sheets in this pink region must be made of three segments seamed together. Suppose $\minimax_1,\minimax_2$ and $\minimax_3$ are the minimax surfaces in the bottom, middle and top segments. Then the relaxed-minimax surface is $\minimax_1 \cup \minimax_2 \cup \minimax_3$ since they are achronal through the time-sheet. However, $\minimax_1$ and $\minimax_2$ are not achronal in $\M$ and so it is not a valid minimax surface. The minimax surface will be $\gamma \cup \minimax_2 \cup \minimax_2$ where $\gamma$ is the maximal surface on the bottom segment that is achronal to both $\minimax_2$ and $\minimax_2$, which will have smaller area than $\minimax_1$.}
    \label{fig:counterexample}
\end{figure}

Notice that the NEC was a fundamental assumption for the above result: there exist more general spacetimes where the areas of relaxed-minimax surfaces and minimax surfaces differ; see figure \ref{fig:counterexample} for an example. It is important to stress that lemma \ref{lemma:relaxed_equal_area} does not preclude the existence of relaxed-minimax surfaces that are not minimax/HRT. A priori the set of relaxed minimax surfaces may be larger, including surfaces $\newminimax$ that are non-achronal in the full spacetime. These surfaces would still be extremal, since (i) they arise from a minimaximization procedure, and further (ii) they can be deformed freely without obstructions (lemma \ref{lem:no-null-segments} ensures they can't contain null segments, though they are allowed to have null-related points through $\mathcal{M}$ by definition). However, we now prove that under standard holographic assumptions non-achronal relaxed minimax surfaces cannot exist, thereby establishing, perhaps surprisingly, the equivalence between these two homology constraints.

\begin{theorem}\label{thm:relaxed-achronal}
An extremal relaxed-minimax surface must be achronal.
\end{theorem}
\begin{proof}
Let $\rmms$ be an extremal relaxed-minimax surface for region $A$ and $\rmmt$ its minimax time-sheet in $\tsset$. Proceeding by contradiction as in the proof of lemma \ref{not-timelike-related}, we now argue that the lack of achronality of $\rmms$ allows us to construct a time-sheet $\nmmt$ on which the maximal surface $\nmms$ has strictly smaller area than $\rmms$, a contradiction from the latter's assumed minimax property.

Define $J:=J^+(\rmms)\cup J^-(\rmms)$. The boundary of $J$ is made up of congruences of null geodesics shot orthogonally from $\rmms$. Since $\rmms$ is extremal, the expansion of the congruence is initially zero, and by focusing it is everywhere non-positive. Therefore no geodesic can reach $\mathcal{N}$, since its expansion would be infinite there. Therefore $J$ does not intersect $\mathcal{N}$.

Next, we show that $J$ contains $\rmmt$. Suppose there is a point $x\in \rmmt\setminus J$. Since $J$ is closed, there is an open neighborhood $\mathcal{O}$ of $x$ outside $J$. There is an achronal surface $\surf_\mathcal{O}\subset\mathcal{O}\cap \rmmt$. Then $\rmms\cup\surf_\mathcal{O}$ is achronal in $\rmmt$ and has larger area than $\rmms$. This contradicts the assumption that $\rmms$ is the maximal achronal surface in $\rmmt$.

\begin{figure}
    \centering
    \includegraphics[width=\linewidth]{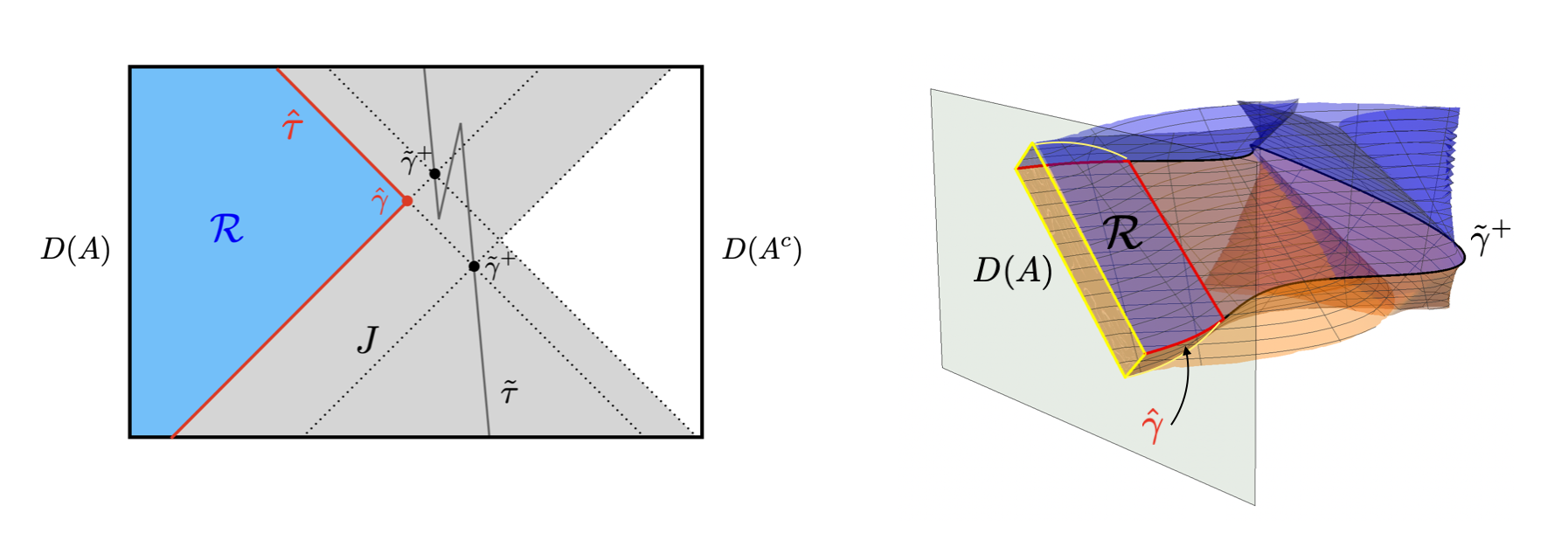}
    \caption{Setup of theorem \ref{thm:relaxed-achronal}. {\bf Left:} We have drawn a cartoon of the proof in the regulated spacetime, by taking a 1 + 1 dimensional cross section. The jagged gray line is  the minimax time-sheet $\rmmt$ which contains two causally related parts of the relaxed minimax surface $\rmms$ (black dots). From $\rmms$ we shoot orthogonal null congruences to define the set $J$, shown as the shaded area in gray. The part of the complement of $J$ that is path connected to $D(A)$ is shown in blue as the region $\mathcal{R}$, whose boundary is the time-sheet $\nmmt$ (red line) containing $\nmms$ (red dot). {\bf Right:} Here we have drawn the setup in the unregulated spacetime, where $\rmms$ meets $A$ (its domain of dependence shown in yellow) at the boundary (green plane). The future and past directed null congruences are shown in blue and orange respectively, and due to the non-achronality of $\rmms$, these sheets self-intersect, enclosing the volume $\mathcal{R}$ which meets the boundary at $\partial D(A)$. So the time-sheet defined as $\nmmt:= \partial \mathcal{R}\setminus\N$ is homologous to $D(A)$. We show in solid red the edge of $\nmmr$, which defines the surface $\mathcal{R}$ composed of some parts of $\rmms$ and a part of the intersection of the lightsheets.}
    \label{fig:intersecting-ew}
\end{figure}

Next, we use $J$ to define a time-sheet. By the homology condition on $\rmmt$, every path connecting $D(A)$ to $D(A^c)$ intersects $\rmmt$. Since $\rmmt\subset J$, every such path intersects $J$. Let $\mathcal{R}$ be the part of the complement of $J$ that is path-connected to $D(A)$. $\mathcal{R}$ touches $\mathcal{N}$ precisely on $D(A)$. Therefore the bulk part of its boundary, which we call $\nmmt$, is homologous to $D(A)$. It is a subset of the boundary of $J$, and is therefore made of null congruences. Hence it is a time-sheet in $\tsset$. Figure \ref{fig:intersecting-ew} shows the setup built so far.

Next, we find the maximal surface on $\nmmt$. Let $\nmmt^\pm$ be the parts of $\nmmt$ that bound $\mathcal{R}$ to the future and past respectively, and define $\nmms:=\nmmt^+\cap\nmmt^-$. This is a spacelike, achronal surface. Every point on $\nmmt$ is connected by a null geodesic to a point on $\nmms$, and, by focusing, the generators of $\nmmt$ have non-positive expansion. Therefore $\nmms$ is maximal on $\nmmt$.

Finally, we show that $|\nmms|<|\rmms|$, which is contradicts the assumption that $\rmms$ is minimax. First, we note that we cannot have $\nmms=\rmms$, since the former is achronal while the latter is not. A part of $\nmms$ that is disjoint from $\rmms$ lies both on a future-directed null congruence starting on $\rmms$ and on a past-directed null congruence starting on a different part of $\rmms$. By focusing, the area of the former is less than or equal to the area of each of the latter, and therefore less than their sum. On the other hand, if there is no part of $\nmms$ that is disjoint from $\rmms$, i.e.\ if $\nmms\subset\rmms$, then $\rmms\setminus\nmms$ is non-empty. Either way, we conclude that $|\nmms|<|\rmms|$.
\end{proof}

\begin{corollary}
    A relaxed-minimax surface is a minimax surface.
\end{corollary}

Hence the relaxed and original minimax prescriptions coincide. Henceforth, for ease of notation, we will use $\minimax$ to refer to minimax or relaxed minimax surfaces interchangeably. Note that once again, the NEC and other holographic assumptions were crucial in establishing this. To best appreciate this equivalence, consider rewriting the relaxed minimax prescription as an extremization procedure analogous to HRT. For HRT, one extremizes over codimension-2 spacelike surfaces that are spacelike homologous to the boundary region. For relaxed minimax, one drops the spacelike homology constraint and extremizes simply over spacelike surfaces that are homologous to the boundary region \emph{through a larger collection of codimension-1 slices}. In general these are different constraints, but as we have seen they are one and the same in the holographic context, a remarkable statement.\footnote{We defer the discussion of interesting consequences of this equivalence to the discussion in section \ref{sec:discussion}.}

\section{Properties of minimax surfaces}\label{sec:properties}

In this section, we present various properties of minimax surfaces, some of which are novel results while others are new minimax-based proofs of known results.

\subsection{The entanglement wedge is the smallest minimax homology region} 

\begin{figure}
    \centering
    \includegraphics[width=0.6\textwidth ]{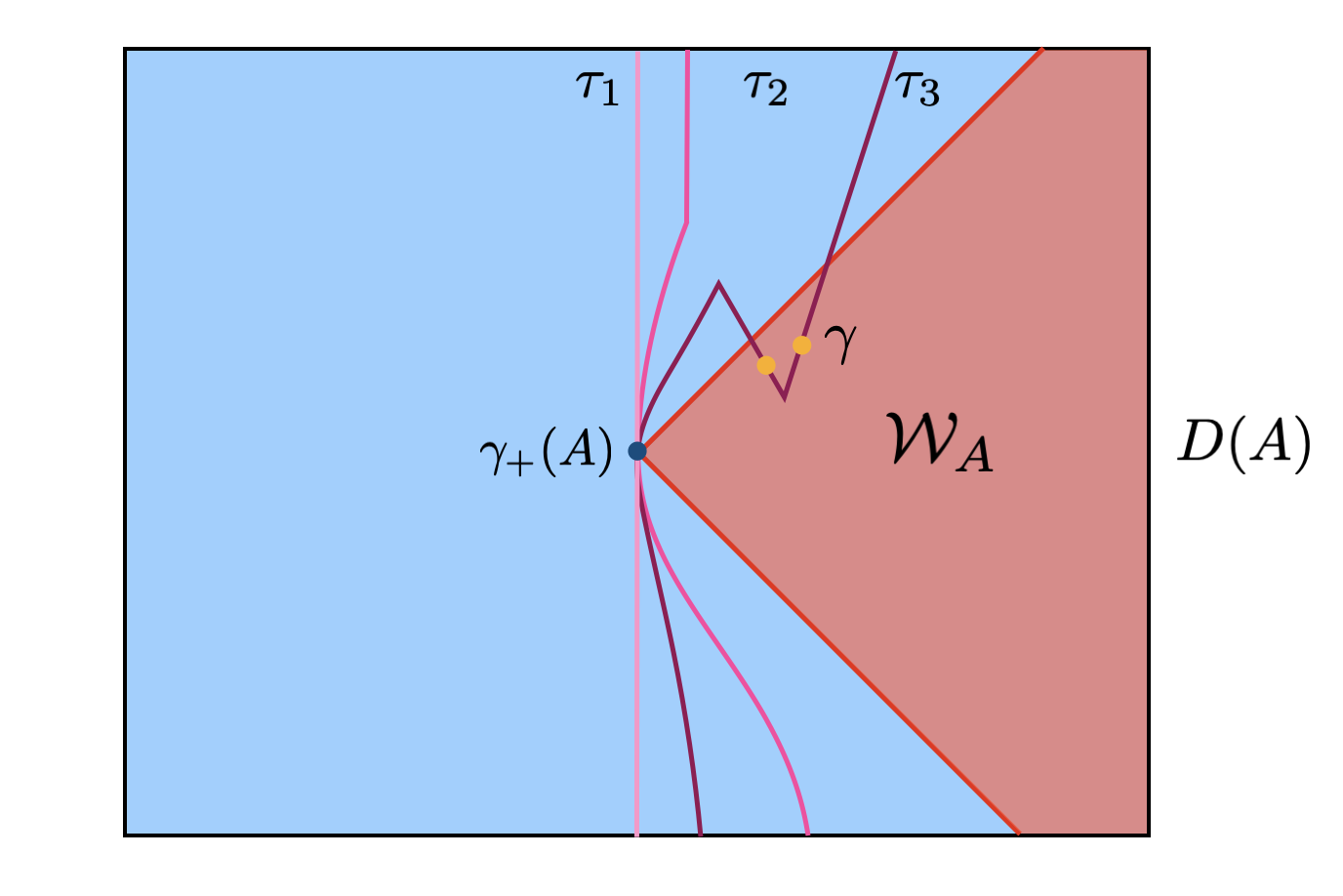}
    \caption{Figure illustrating the proof of theorem \ref{thm:ew_smallest} through the use of a Penrose diagram in 1+1 dimensions. The minimax surface $\minimax(A)$ is shown as a blue dot, its entanglement horizon $\hor(A)$ is depicted by the pair of null red lines, and its entanglement wedge $\ew_A$ is the shaded red region bordered by $\hor(A)$. We also show three different types of time-sheets containing $\minimax(A)$ in various shades of pink. $\tau_1$ is smooth and timelike while $\tau_2$ is piecewise-timelike with a seam; their respective homology regions $R_1(A)$ and $R_2(A)$ both contain $\ew_A$. On the other hand, the third time-sheet, $\tau_3$, has tent-like seams which allow it to enter $\ew_A$ --- however, in doing so, it admits an additional surface $\gamma$ (shown in yellow) achronal to $\minimax(A)$, which leads to contradiction between $\minimax(A)$ being minimax and $\tau_3$ being a minimizing time-sheet.
   }
    \label{fig:ew-smallest-proof}
\end{figure}

In this section, we prove a theorem that constrains the types of minimax time-sheets that can arise in the minimax prescription. As discussed in subsection \ref{sec:minimax}, the set of minimax time-sheets is quite large: far from $\minimax$ the time-sheet can be deformed quite a bit while keeping $\minimax$ maximal on it. In the set $\tsset_A$ of time-sheets homologous to $D(A)$ there exist ``kinky'' time-sheets that contain $\minimax$ and penetrate the entanglement wedge $\ew_A$ (for example, they may feature tent-shaped spacelike seams that allow the time-sheet to change direction in time, see figure \ref{fig:ew-smallest-proof} for an example of such a time-sheet). Below we show that these time-sheets are never in the minimax set. In other words, we show that the entanglement wedge is the smallest minimax homology region, namely the homology region bounded by a minimax time-sheet that is contained in all the others. This result is particularly suggestive in the context of bulk reconstruction, as it reinforces the natural expectation that the reconstructible bulk region should be, in some sense, the smallest one (a sharp example where this expectation has appeared before is in the presence of degenerate HRT surfaces for the same boundary region, in which case the true entanglement wedge is believed to be the one contained in all the others \cite{Grado-White:2024gtx}).

\begin{theorem}\label{thm:ew_smallest}
Given a boundary region $A$ and minimax surface $\minimax(A)$, the entanglement wedge $\ew_A$ is the smallest spacetime homology region among homology regions bounded by minimax time-sheets.
\end{theorem}

\begin{proof}
To reach a contradiction suppose that $\ew_A$ was not the smallest homology region bounded by a minimizing time-sheet. Then, there must be a minimax time-sheet $\ts$ that intersects the interior of $\ew_A$ (which can happen in situations such as  depicted by $\tau_3$ in  figure \ref{fig:ew-smallest-proof}).  This means there is a bulk Cauchy slice $\sigma$ containing both $\minimax(A)$ and an extra (spatially null-homologous) surface $\gamma$ in the part of $\ts$ inside the interior of $\ew_A$.  The resulting achronal multi-component surface $\minimax(A) \cup \gamma = \sigma \cap \tau \in \Gamma_{\ts}$ then has manifestly larger area than $\minimax(A)$, which precludes $\minimax(A)$ from being minimax, contradicting our starting assumption.  Hence $\ts$ could have not been a minimizing time-sheet. 
\end{proof}

If there are degenerate surfaces for the same boundary region, the above proof applies independently to each surface, ensuring that for each one, there cannot be minimax time-sheets that penetrate the surface's entanglement wedge. Then, the smallest homology region will correspond to the smallest entanglement wedge.

\subsection{The causal information surface has larger area than the minimax surface}

In this subsection we give a minimax-based proof of a known result, concerning the causal holographic information $\chi(A)$ of a boundary region $A$, which is defined as the area of the codimension-2 spacelike rim of the causal wedge --- the bulk region that can be causally reached from $D(A)$  \cite{Hubeny:2012wa}. This surface is known as the causal information surface, and denoted by $\Xi_A$.

The causal information surface satisfies various global properties summarized in \cite{Hubeny:2013gba}, two important ones being that it always lies closer to the boundary than the HRT surface \cite{Hubeny:2012wa, Headrick:2014cta, Wall_2014} (a fact crucial for establishing compatibility of holographic entanglement with field theory causality) and that it is the minimal area codimension-2 surface lying on the boundary of the causal wedge.  

The causal holographic information has been argued to represent some sort of coarse-grained entropy \cite{Hubeny:2012wa, Kelly:2013aja}. Therefore, it is natural to expect that generally $\chi(A) \geq S(A)$. This is clear for static states: both the RT surface and the causal information surface lie on the same time-reflection symmetric slice, so the inequality follows from the global minimality of the RT surface on the slice. For HRT, general arguments in favor of the bound were made in \cite{Hubeny:2012wa}, and a proof using the maximin prescription was given in \cite{Wall_2014}. The minimax formulation allows for a very simple (and slightly more direct) proof, which we now present.

\begin{theorem}\label{lem:causal-info-surface-thm}
    Given a boundary region $A$, the causal holographic information $\chi(A)$ upper-bounds the holographic entanglement entropy $S(A)$.
\end{theorem}
\begin{proof}
    The proof is constructive. To show that $\chi(A) \geq S(A)$ we will define a time-sheet which: (i) contains the causal information surface $\Xi_A$, (ii) is spacetime-homologous to $D(A)$, and (iii) on which $\Xi_A$ is maximal. Given such a time-sheet, the result will follow since $\Xi_A$ will be a (not necessarily minimal) candidate minimax surface.
    
    The most straightforward guess for this special time-sheet would be the boundary of the causal wedge, however this can be quickly ruled out since $\Xi_A$ is minimal, not maximal, on it. To fix this, consider the boundaries of the future and past bulk domains of influences of $D(A)$ (whose intersection forms the causal information surface) and keep only the portions of these two boundaries that are not part of the causal wedge, i.e.\ the two null sheets that start from $\Xi_A$ and continue to the future and to the past away from $D(A)$ (in other words, part of the boundary of the causal shadow \cite{Headrick:2014cta}). This null hypersurface is homologous to $D(A)$, and further, $\Xi_A$ is maximal on it: by focusing (in this case used analogously to the area theorem), the null congruences through $\Xi_A$ must expand towards the boundary in order to provide generators of $\partial I^\pm[D(A)]$, a set defined from $\N$. Since $\minimax_A$ is obtained by minimizing over all time-sheets, $|\Xi_A| \geq |\minimax_A|$ and the theorem follows. See figure \ref{fig:causal-surface-thm} to see a visual of the construction.
\end{proof}

\begin{figure}
    \centering
    \includegraphics[width=0.8\textwidth]{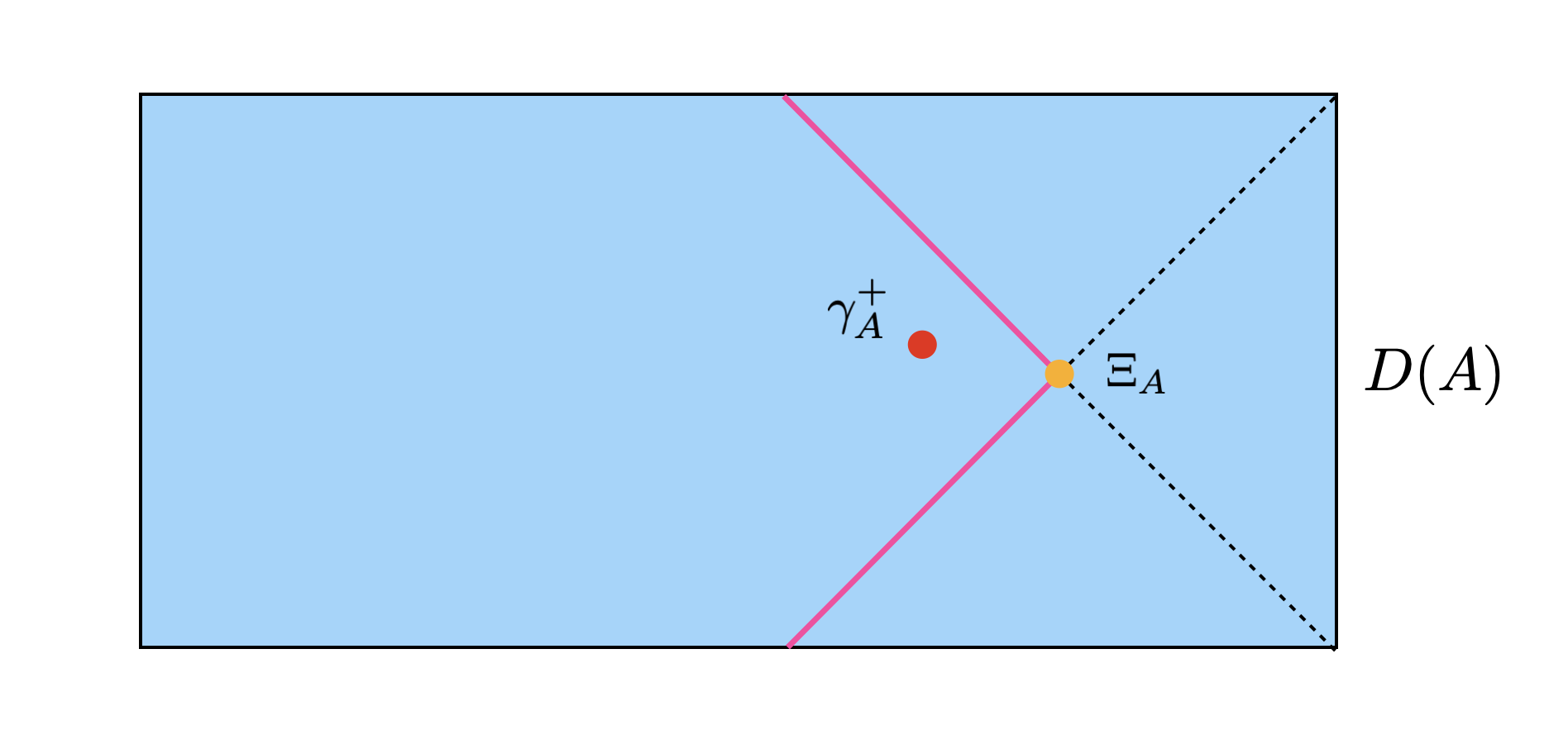}
    \caption{Schematic illustration of the construction in the proof of lemma \ref{lem:causal-info-surface-thm}. The causal information surface $\Xi_A$ is shown in yellow, while the minimax surface $\minimax$ is shown in red. The black dotted lines indicate the boundary of the causal wedge, and the pink solid lines form the time-sheet on which $\Xi_A$ is maximal.}
    \label{fig:causal-surface-thm}
\end{figure}

\subsection{Entanglement wedge nesting \& causality of minimax surfaces} \label{sec:nesting}

We showed in subsection \ref{sec:stability} that two minimax surfaces for the same boundary region cannot lie on the same time-sheet. We will now strengthen this statement by showing that they must be fully spacelike-related. We will show this as a consequence of entanglement wedge nesting.
 
\begin{theorem}
Given boundary regions $A,B$ with $D(B)\subseteq D(A)$, we have $\ew_B \subseteq \ew_A$.
\end{theorem} 

\begin{proof}
The proof follows a similar strategy employed in theorem \ref{thm:relaxed-achronal}. Let $\minimax_A$ and $\minimax_B$ be the minimax surfaces for the two boundary subregions $A$ and $B$, with minimax time-sheets $\ts_A\in \tsset_A$ and $\ts_B\in \tsset_B$. Define $J_A := J^+(\minimax_A) \cup J^{-}(\minimax_A)$, $J_B:= J^+(\minimax_B) \cup J^{-}(\minimax_B)$ and $J := J_A \cup J_B$. 

By theorem \ref{thm:ew_smallest}, $\ts_A \subset J_A$ and $\ts_B \subset J_B$, so $J$ contains both of them. By the homology conditions on $\ts_A$ and $\ts_B$, every path connecting $D(B)$ with $D(B^c)$ intersects both $\ts_A$ and $\ts_{B}$ since $D(B) \subseteq D(A)$. So, every such path will also intersect $J$. Let $\mathcal{R}$ be the part of the complement of $J$ that is path connected to $D(B)$. $\mathcal{R}$ touches $\N$ precisely on $D(B)$, and the bulk part of its boundary forms a time-sheet homologous to $D(B)$ which we will call $\hat{\ts}$. Entanglement wedge nesting is the statement $\mathcal{R} = \ew_B$, so we assume this is not the case for the sake of contradiction.

Call $\hat{\surf}$ the edge of $\hat{\ts}$; this is a spacelike achronal surface. Since $\hat{\ts}$ is made of null congruences from $\minimax_A$ and $\minimax_B$ which are extremal, the generators of $\hat{\ts}$ have non-positive expansion making $\hat{\surf}$ maximal on $\hat{\ts}$. 

We now show that $|\hat{\surf}| < |\minimax_B|$, contradicting the minimality of $\minimax_B$ as a minimax surface. Note that $\hat{\surf}$ can't fully coincide with $\minimax_B$ by assumption. Then, the surface $\hat{\surf}$ is made of portions of $\minimax_B$ whenever $\minimax_B$ is spacelike separated with $\minimax_A$, and a portion that lies both on a future-directed null congruence starting on one of the two minimax surfaces and a past-directed null congruence starting on the other one whenever $\minimax_A$ and $\minimax_B$ are causally related. By focusing we arrive at the desired contradiction $|\hat{\surf}| < |\minimax_B|$. So it must be that $\mathcal{R} = \ew_B$, and hence $\ew_B \subseteq \ew_A$.
\end{proof}

The case where $D(A)$ and $D(B)$ coincide corresponds to a strengthening of lemma \ref{not-timelike-related}. 

\begin{corollary}
\label{cor:spacelike_related}
In case of multiple degenerate minimax surfaces $\minimax_i(A)$ for the same boundary region $A$, the $\minimax_i(A)$ must all be fully spacelike separated from one another (and be maximal on different minimax time-sheets). 
\end{corollary}

\subsection{Cooperating pairs of time-sheets}\label{sec:coop}

In this section we introduce and prove a very interesting property of time-sheets, which we will see to be a special case of a more general conjecture discussed in more detail in section \ref{sec:conjecture}.

Consider two crossing boundary subsystems $\mathscr{J}_1 = AB$ and $\mathscr{J}_2 = BC$ (if they do not cross, the following statements are trivially true). Furthermore we will initially assume that the minimax surfaces $\minimax_1$ and $\minimax_2$, and correspondingly their minimax time-sheets $\ts_1$ and $\ts_2$, are each connected; we will relegate the more general case involving multiple components until after the proof of theorem \ref{thm:coop-2party}. Since $\mathscr{J}_1$ and $\mathscr{J}_2$ cross $\ts_1$ and $\ts_2$ will intersect --- cutting each other in what we call \textit{partial time-sheets}. Due to this intersection, $\minimax_1$ and $\minimax_2$ will also be cut into \textit{partial minimax surfaces}.

\begin{figure}
    \centering
    \includegraphics[width=0.9\linewidth]{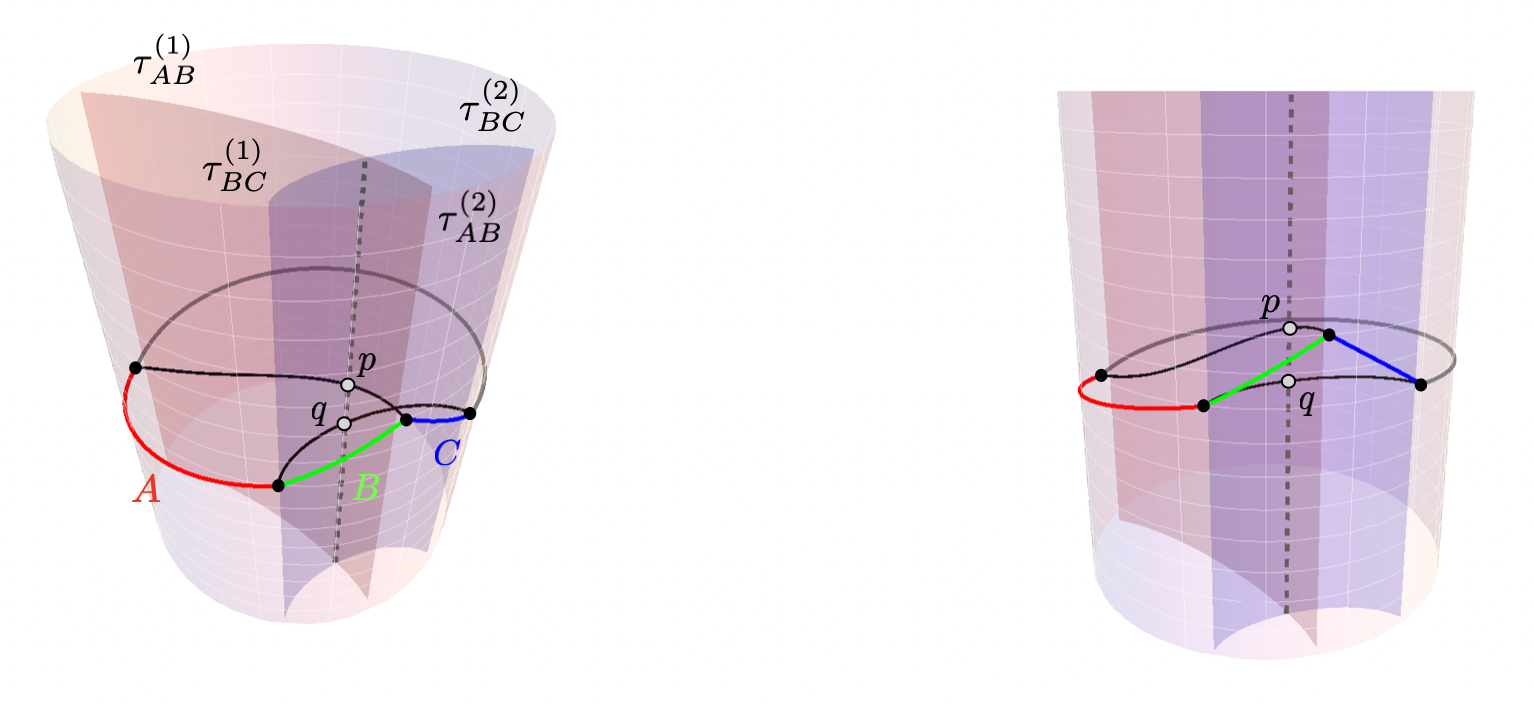}
    \caption{Two viewpoints of a configuration $(\ts_{AB},\ts_{BC})$ of minimax time-sheets for the crossing regions $AB$ and $BC$. We show the construction in the unregulated spacetime for clarity. {\bf Left:} The regions $A$ (red), $B$ (green), $C$ (blue) and the purifier (gray) are relatively boosted on the boundary, but within a common boundary Cauchy slice. In black we show the two minimax surfaces $\minimax_{AB}$ and $\minimax_{BC}$. The two time-sheets intersect along the dashed gray line forming the partial time-sheets $\ts_{AB}^{(1)},\, \ts_{AB}^{(2)},\, \ts_{BC}^{(1)}$ and $\ts_{BC}^{(2)}$. {\bf Right:} With this viewpoint we make it clearer that the two minimax surfaces do not intersect and do not lie on the same bulk slice. The intersection of the time-sheets induces a partitioning of the minimax surfaces into four parts: $\minimax_{AB}$ is split into two parts that meet the intersection seam at the locus $p$. The surface $\minimax_{BC}$ is cut similarly but through the locus $q$. Theorem \ref{thm:coop-2party} states that there exists a pair $(\ts_{AB},\ts_{BC})$ for which these four partial surfaces are all the maximal achronal surfaces on the respective partial time-sheets. The maximization problem on each partial time-sheet has fixed boundary condition on the boundary, but it is free to end on the intersection seam wherever and however. In the example shown in the figure, while not immediately apparent, the configuration is not cooperating since the minimax surface for $AB$ does not meet the intersection seam orthogonally (a necessary condition for cooperation) --- since the time-sheets have been constructed in a naive way by extending the minimax surfaces upward and downward in time statically.} 
    \label{fig:coop-2-party}
\end{figure}

The property we show here is that there exists a pair of minimax time-sheets $(\ts_1, \ts_2)$ for which the maximal achronal surface on each partial time-sheet is its corresponding partial minimax surface. Note that this property is highly non-trivial: the maximization problem on a partial time-sheet has a boundary condition that depends on the intersection seam $\ts_1 \cap \ts_2$ of both time-sheets, so in general one needs $\ts_1$ and $\ts_2$ to work together in order for the property to hold. For this reason, we say that when this property holds, the pair of time-sheets $(\ts_1, \ts_2)$ is \textit{cooperating}.\footnote{ 
We postpone a more formal definition to the general case presented in definition \ref{def:cooperating}.}

To examine whether we have sufficient freedom in choosing the time-sheets so as to ensure the cooperating property, let us specify the relevant constructs further.  Let $\confset$ be the set of all pairs $(\ts_1, \ts_2)$ of minimax time-sheets. A specific pair $\conf \in \confset$ will be called a \textit{configuration}, and we denote the set of cooperating configurations by $\coopset \subset \confset$.  See figure \ref{fig:coop-2-party} for a drawing of the setup and a visual explanation of the cooperating property.
Note that while $\confset$ is certainly smaller than $\tsset_{\mathscr{J}_1}\times\tsset_{\mathscr{J}_2}$ (the set of \emph{all} pairs of time-sheets as opposed to only the minimax ones), the choice of minimax time-sheet for a given minimax surface $\minimax$ is still highly non-unique: away from $\minimax$, the time-sheet can be deformed by some amount while remaining a minimax time-sheet.  Since there is an infinite family of suitable deformations, i.e.\ the set $\confset$ is infinite, it seems at least \emph{plausible} that one retains enough wiggle room to find a cooperating configuration. A necessary condition for that to happen is for the minimax surfaces to meet the intersection seam of the two time-sheets orthogonally.

While \emph{locally} the intersection of the time-sheet is constrained by the orthogonality condition, \emph{globally} a variety of complications could arise. However, we will further show that cooperating configurations are particularly simple in that they partition the spacetime into four connected codimension-0 volumes (which we call \emph{cells}) that extend all the way to the conformal boundary. We call such a configuration to be \textit{minimally intersecting}. Intuitively, one can view this property as saying that a cooperating pair of time-sheets, as viewed from the top (i.e.\ suppressing the time direction) behaves similarly to pairs of RT surfaces on a slice in the static context (i.e.\ they do not weave back and forth intersecting each other multiple times, e.g.\ as in figure \ref{fig:multiple-intersections}). The set of minimally intersecting configurations will be denoted by $\mintset \subset \confset$.

We now present the proof that a cooperating pair of time-sheets exists.
\begin{theorem}
\label{thm:coop-2party}
For a crossing pair of boundary subsystems $\mathscr{J}_1,\mathscr{J}_2$ admitting connected minimax time-sheets, there exists a pair $(\ts_1, \ts_2)$ of cooperating time-sheets. 
\end{theorem}
\begin{sproof}\renewcommand{\qedsymbol}{}
Consider the maximal surfaces $\psurf{i}{\alpha}$\footnote{Note that the maximal surfaces are generically \emph{not} the partial minimax surfaces, which we will denote as  $\pminimax{i}{\alpha}$; while $\pminimax{1}{1}$ and $\pminimax{1}{2}$ are  connected and achronal,  $\psurf{1}{1}$ need not be connected with or achronal to $\psurf{1}{2}$, being individually maximized on the respective partial time-sheets.} on the corresponding partial time-sheets $\ts_{i}^{(\alpha)}$, with $i = 1,2$ and $\alpha$ labeling the individual parts of each $\tau_i$. We then minimize the total sum of the areas of the maximal surfaces $\sum |\psurf{i}{\alpha}|$ over configurations $\conf\in\confset$, i.e.\ over all pairs $(\ts_1, \ts_2)$ of minimax time-sheets. For the configuration $\conf^{\star}$ that achieves this minimum, we prove that all the maximal surfaces for a given time-sheet must join, i.e.\ that 
\begin{equation}\label{eq:connected}
    \bigcup_\alpha \psurf{i}{\alpha} \text{ is connected,} \quad( i = 1,2)\,.
\end{equation}
Since these surfaces lie on minimax time-sheets, and since they join smoothly, it follows that they must form the two relaxed minimax surfaces $\minimax_1$ and $\minimax_2$ (note that the maximal segments are individually achronal on their partial time-sheets, so since they join, their union will be achronal on the full time-sheet). We prove  \eqref{eq:connected} by contradiction. If the maximal surfaces were disconnected, we could find a deformation of the time-sheets (within $\confset$) such that the total sum of areas decreases, contradicting the aforementioned minimality. 
\end{sproof}
\begin{proof}
The set $\confset$ of configurations will be made of pairs $(\ts_1, \ts_2)$ of minimax time-sheets for the boundary regions $\mathscr{J}_1$ and $\mathscr{J}_2$. Let $\pts{i}{\alpha}$ be the partial time-sheets formed from the intersection of a configuration $\conf \in \confset$. On each partial time-sheet $\pts{i}{\alpha}$, consider the maximal achronal-in-$\pts{i}{\alpha}$ surface $\psurf{i}{\alpha}$. Then, consider the total sum of the areas of these maximal surfaces, i.e.
\be
\sum_{\alpha,i} \max_{\psurf{i}{\alpha} \subset \ts_{i}^{(\alpha)}} |\gamma_{i}^{(\alpha)}|= \max_{\psurf{i}{\alpha} \subset \ts_{i}^{(\alpha)}} \sum_{\alpha,i}  |\gamma_{i}^{(\alpha)}|,
\ee
and finally minimize this over all configurations $t \in \confset$:
\be\label{eq:minimax_lemma}
\min_{\confset} \sum_{\alpha,i} \max_{\psurf{i}{\alpha} \subset \ts_{i}^{(\alpha)}} |\gamma_{i}^{(\alpha)}|  = \min_{\confset}  \max_{\psurf{i}{\alpha} \subset \ts_{i}^{(\alpha)}} \sum_{\alpha,i}  |\gamma_{i}^{(\alpha)}|.
\ee
Denote the minimizing pair by $\conf^{\star} = (\ts_1^{\star}, \ts_{2}^{\star})$\footnote{
As for minimax itself,
we work under the assumption that a solution exists.}. We want to show that $\conf^{\star}$ is such that
\be\label{eq:join-smoohtly}
\bigcup_\alpha \psurf{1}{\alpha} \text{ is connected, \ \ and }
\bigcup_\alpha \psurf{2}{\alpha} \text{ is connected.}
\ee
Since each $\psurf{i}{\alpha}$ is maximal and achronal in $\pts{i}{\alpha}$, if \eqref{eq:join-smoohtly} holds, it implies that the $\cup_\alpha \psurf{i}{\alpha}$ are the maximal surfaces achronal on $\ts_1^{\star}$ and $\ts_2^{\star}$ respectively, and since $\ts_1^{\star}$ and $\ts_2^{\star}$ are minimax time-sheets, it further follows that the $\cup_\alpha \psurf{i}{\alpha}$ are the (relaxed) minimax surfaces, i.e.
\begin{equation}
\bigcup_\alpha \psurf{1}{\alpha} = \bigcup_\alpha \gamma_1^{+(\alpha)} = \minimax_1\,,\qquad
\bigcup_\alpha \psurf{2}{\alpha} = \bigcup_\alpha \gamma_2^{+(\alpha)} = \minimax_2\,.
\end{equation}
To show that $\conf^{\star}$ is such that \eqref{eq:join-smoohtly} holds, we proceed by contradiction: we assume the partial maximal surfaces to be disconnected and show that that contradicts minimality of $\conf^\star$. Since we will need to take care of a variety of special cases, we split the proof into small sections for improved readability.

\paragraph{Case I: $\conf^{\star}\in \mintset$, and generic intersections:}
We begin with assuming that  $\conf^{\star}\in \mintset$, i.e.\ that $\conf^{\star}$ is minimally intersecting. As a consequence, the two time-sheets intersect forming four partial time-sheets. Let $\mathpzc{i} := \ts^{\star}_1 \cap \ts^{\star}_2$ be the intersection locus of $\ts^{\star}_1$ with $\ts^{\star}_2$. In general, $\mathpzc{i}$ may be timelike, spacelike or null (or a piecewise combination of the three). To show \eqref{eq:join-smoohtly}, assume, for contradiction, that the two pairs of maximal surfaces do not join (i.e.\ they intersect $\mathpzc{i}$ at different locations). For the moment, we are further going to assume that all of the loci of intersections are different (even between maximal surfaces for different time-sheets). We want to show that if so, then there exists a deformation of $\ts_1^{\star}$ and $\ts_2^{\star}$ (in the space $\confset$) that would decrease $\max \sum |\psurf{i}{\alpha}|$, contradicting the minimality in \eqref{eq:minimax_lemma}. Assume for now that the maximal surfaces $\psurf{i}{\alpha}$ are unique.\footnote{Note that there may be other equal-area minimax surfaces on the full time-sheet. These however will be unstable (or non-achronal in the full spacetime). However, there will be a unique (stable, extremal, achronal) HRT surface by the arguments in subsection \ref{sec:stability}. The set $\confset$ is the set of configurations containing the $\mathscr{J}_i$ HRT surfaces, which are then maximal on their respective time-sheets.
So when deforming the time-sheet, we do so in a neighborhood of any non-HRT minimax surface, while keeping the HRT surface fixed.}

Firstly, note that since the $\psurf{i}{\alpha}$ are maximal, they must intersect $\mathpzc{i}$ orthogonally. So the $\psurf{i}{\alpha}$ must end on a timelike portion of $\mathpzc{i}$ (since if they ended on a spacelike or null portion, one could always add positive area to the surface, contradicting maximality).
Since we assumed for contradiction that each $\bigcup_\alpha \psurf{i}{\alpha}$ is disconnected (for fixed $i = 1,2$), at most one partial maximal surface on a given time-sheet can intersect the minimax surface at $\mathpzc{i}$ (and at most two partial maximal surfaces can intersect with either minimax surface). Without loss of generality, assume $\psurf{1}{1}$ does not intersect with either $\minimax_1$ or $\minimax_2$, namely $\psurf{1}{1}\cap\mathpzc{i}\cap\minimax_1 = \emptyset$ and $\psurf{1}{1}\cap\mathpzc{i}\cap\minimax_2 = \emptyset$. Then, consider a deformation of $\ts_2^{\star}$ in a $\delta$-sized
neighborhood $\sets{U}_\delta$ of $\psurf{1}{1} \cap \mathpzc{i}$, such that $\ts_2^{\star}$ moves towards $\psurf{1}{1}$ by an amount of order $\epsilon$, i.e.\
\begin{equation}\label{eq:deformation}
    \ts_2^{\star} \to \tilde{\ts}_2 = \ts_2^{\star} + \epsilon \, \eta
\end{equation}
where $\eta$ is a normal vector field with support on $\sets{U}_\delta$. Because $\psurf{1}{1}\cap\mathpzc{i}\cap\minimax_i = \emptyset$, this neighborhood $\sets{U}_\delta$ can be chosen such that it also does not contain $\minimax_i$, and so by stability, $\epsilon$ can be chosen such that $\minimax_2$ remains maximal. Therefore, the deformed configuration remains in $\confset$.

We now consider the areas of the new maximal segments $\psurf{i}{\alpha}$ following the deformation \eqref{eq:deformation}. The surface $\psurf{1}{1}$ will decrease in area by order $\epsilon$ by construction, while the surfaces $\psurf{1}{2}, \psurf{2}{1}$ and $\psurf{2}{2}$ will all be unchanged since we are only deforming $\ts_2^{\star}$ in the neighbourhood $\sets{U}_{\delta}$.\footnote{Note, however, that though $\ts_1^{\star}$ is not deformed, the intersection $\mathpzc{i}$ is. As such, the maximal partial surfaces might again want to jump either outside (in the case of $\gamma_1^{(1)}$) or inside (in the case of $\gamma_2^{(1)}$, $\gamma_1^{(2)}$ and $\gamma_2^{(2)}$) the deformation to maximize their area. However, we can always choose an $\epsilon$ small enough such that this does not occur.} All together then, the total area of the four maximal surfaces decreases
\be
\max \sum_i |\psurf{i}{\alpha}| \to \max \sum_i |\psurf{i}{\alpha}| - O(\epsilon), 
\ee
in contradiction to the minimality of \eqref{eq:minimax_lemma}. 

\begin{figure}
    \centering
    \includegraphics[width=0.9\linewidth]{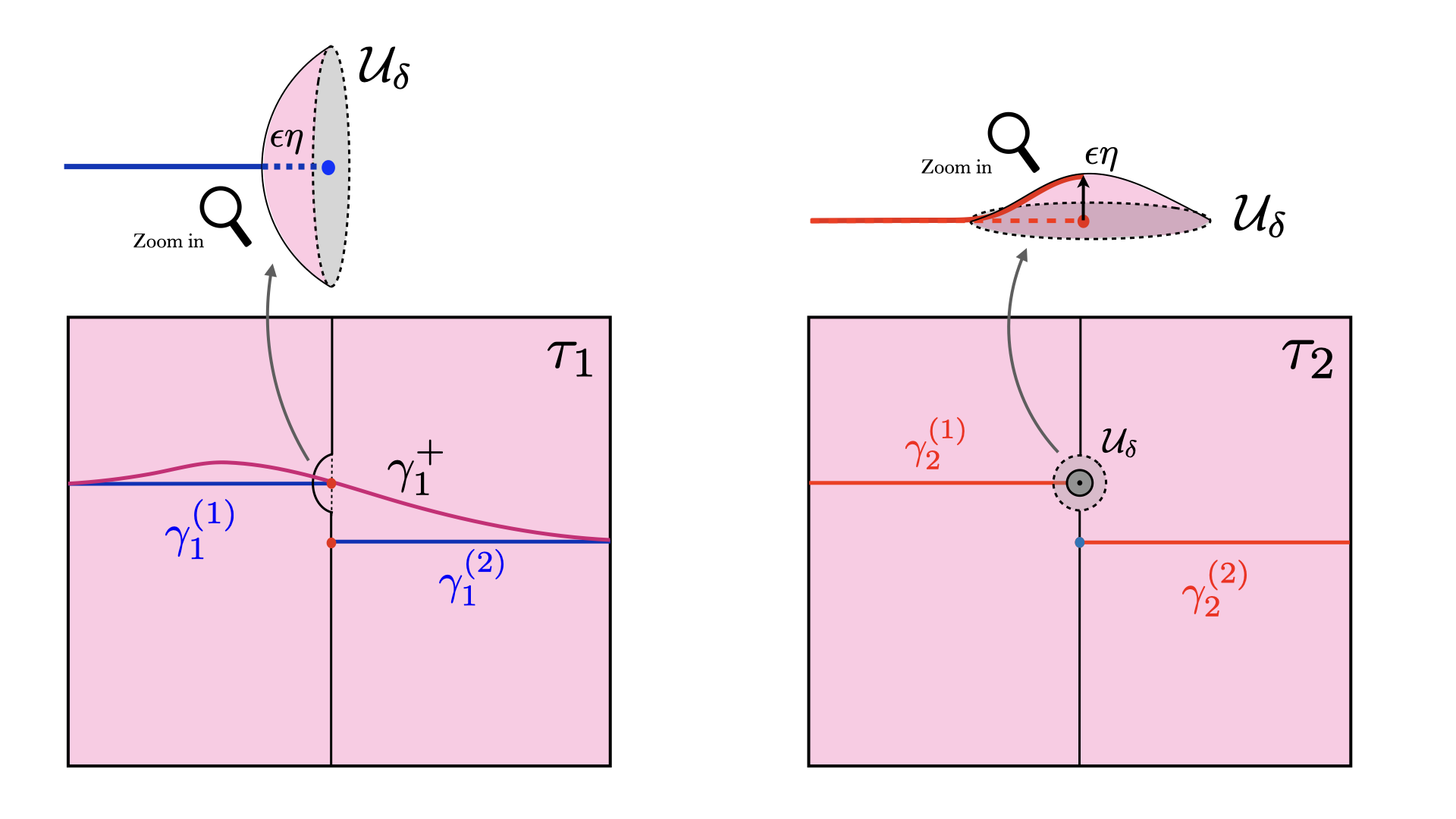}
    \caption{Cartoon showing the deformation used to reach a contradiction in theorem \ref{thm:coop-2party} for the case where $\gamma_1^{(1)}$, $\gamma_2^{(1)}$, and $\minimax_1$ all intersect at the intersection seam. {\bf Left:} We show time-sheet $\ts_1$, containing the minimax surface $\minimax_1$ (purple). The two maximal partial surfaces $\psurf{1}{1}$ and $\psurf{1}{2}$ (blue) do not intersect at the intersection seam, however $\psurf{1}{1}$ intersects $\minimax_1$. {\bf Right:} We show time-sheet $\ts_2$ and the two partial maximal surfaces $\psurf{2}{1}$ and $\psurf{2}{2}$ (red). Because $\psurf{1}{1}$ and $\psurf{2}{1}$ intersect, a deformation of $\ts_2$ towards $\psurf{1}{1}$ in a normal direction will increase the area of  $\psurf{2}{1}$ (shown in the zoomed in area in the top right), but decrease the area of $\psurf{1}{1}$ by a greater amount (shown in the zoomed in area in the top left). Altogether the deformation decreases the total area of the configuration.}
    \label{fig:cooperating-thm}
\end{figure}

\paragraph{Case II: $\conf^{\star}\in \mintset$, worst case scenario of intersections:}
Above we assumed that no two maximal partial surfaces intersected $\mathpzc{i}$ at the same point. Without loss of generality, consider now the worst case scenario in which both $\psurf{1}{1} \cap \psurf{2}{1}$ and $\psurf{1}{2} \cap \psurf{2}{2}$ are non-empty, with $\psurf{1}{1} \cap \psurf{2}{1} \cap \minimax_1$ non-empty at $\mathpzc{i}$.\footnote{At most one minimax surface can intersect either $\psurf{1}{1} \cap \psurf{2}{1}$ or $\psurf{1}{2} \cap \psurf{2}{2}$.} There exists deformations of $\ts_2^{\star}$ within $\confset$ towards $\psurf{1}{1}$, but they will inevitably increase the area of $\psurf{2}{1}$ (since the deformation \eqref{eq:deformation} would push it in a spacelike direction). However, while on $\psurf{1}{1}$ the deformation has the effect of excising a portion of the surface with area of order $\epsilon$, on $\psurf{2}{1}$ 
it merely deforms a  $\delta-$sized neighborhood portion of it under the normal field $\epsilon \eta$, hence increasing the area of $\psurf{2}{1}$ by an amount of order $\epsilon\delta$. Thus, $\epsilon$ can be taken small enough such that there is still a deformation that decreases the total area in \eqref{eq:minimax_lemma}, and we arrive at the same contradiction. See figure \ref{fig:cooperating-thm} for a cartoon of the deformations involved in this step of the proof.

\paragraph{Case III: $\conf^{\star}\in \mintset$, degeneracies on the partial time-sheets:}
We now take care of the case where the maximal surfaces $\psurf{i}{\alpha}$ are not unique on their respective partial time-sheets. For example, suppose there are two maximal surfaces $\psurf{1}{1}$ and $\bar{\gamma}_1^{(1)}$on $\pts{1}{1}$. We assume, again for contradiction, that neither $\psurf{1}{1}$ nor $\bar{\gamma}_1^{(1)}$ intersect $\psurf{1}{2}$. In the worst case configuration of partial surfaces (i.e.\ when $\minimax_2\cap\psurf{1}{2}\cap\psurf{2}{1}$ is non-empty) we may want to deform $\ts_2^{\star}$ to decrease the area of $\psurf{1}{1}$. However, if we perform the deformation \eqref{eq:deformation} around only one, say $\psurf{1}{1}$, the new maximal surface on $\pts{1}{1}$ would simply be $\bar{\gamma}_1^{(1)}$, so our deformation would not decrease the area of \eqref{eq:minimax_lemma}. However, we can simultaneously perform the deformation in neighborhoods of both $\psurf{1}{1}$ and $\bar{\gamma}_1^{(1)}$ (note this deformation can be done, as $\psurf{1}{1}\cap\psurf{1}{2}$ and $\bar{\gamma}_1^{(1)}\cap\psurf{1}{2}$ are both empty). Thus, there will still exist a deformation that decreases the area, yielding a contradiction, as above. At least one maximal partial surface therefore needs to smoothly join with $\psurf{1}{2}$. This concludes the minimally intersecting part of the proof, showing that when $\conf^{\star}\in \mintset$, $\conf^{\star}\in \coopset$.

\begin{figure}
    \centering
    \includegraphics[width=0.45\textwidth]{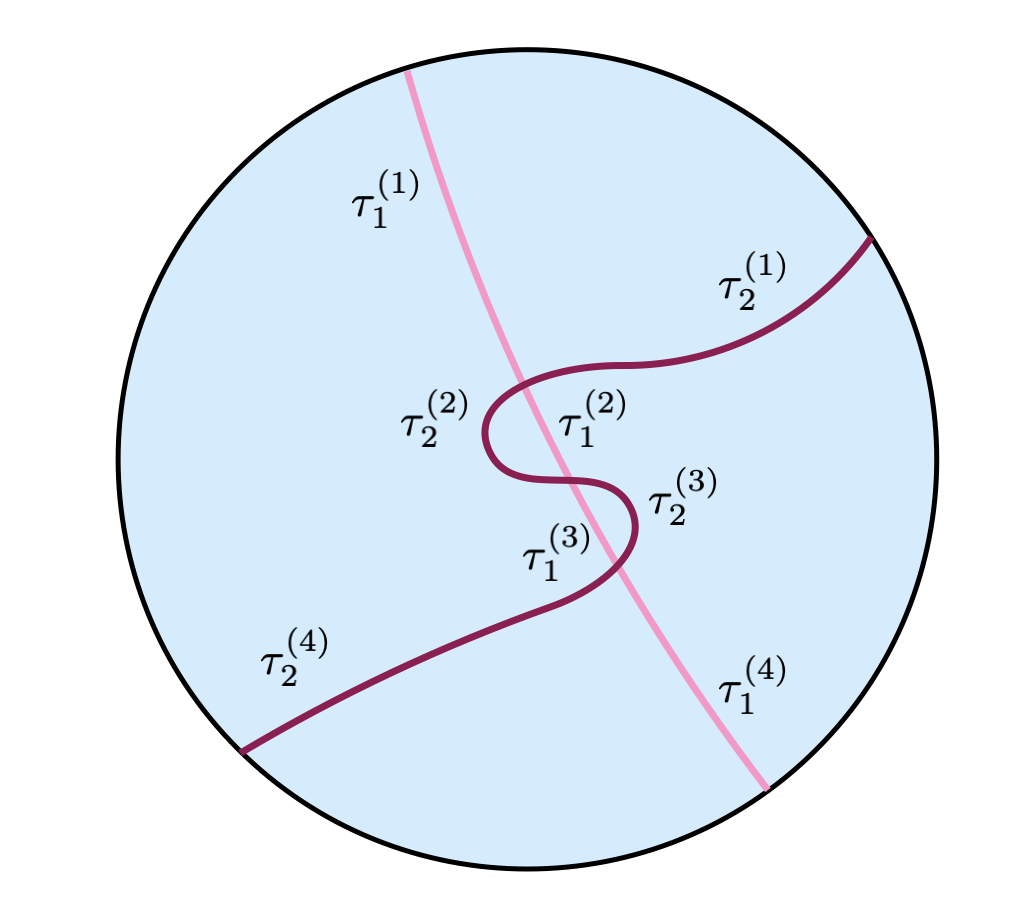}
    \caption{Example of a non-minimally intersecting configuration of time-sheets, where one of the two time-sheets weaves back and forth to intersect the other multiple times. Theorem \ref{thm:min-int} ensures that this configuration is never cooperating.}
    \label{fig:multiple-intersections}
\end{figure}

\paragraph{Case IV: $\conf^{\star} \notin \mintset$:}
Finally, suppose now that $\conf^{\star}\notin \mintset$. This could manifest in a multitude of topologically different configurations, see figure \ref{fig:multiple-intersections} for an example. In particular, for any two connected components of  $\conf^{\star}$, there may be multiple connected components of the seam $\mathpzc{i}_j$. However, note that the seam $\mathpzc{i}_j$ will not self-intersect, as this could only occur if the time-sheet self-intersects (which is not allowed by the homology condition).
As a consequence of this simple structure, around any given component of the seam, $\mathpzc{i}_j$, we can consider as above whether the partial surfaces $\psurf{i}{j}$, $\psurf{i}{j+1}$ connect across it. Note, however, that the arguments above did not make any reference to the structure of the time-sheets away from a neighborhood of $\mathpzc{i}_j$. Thus, by arguments identical to those given above, the partial surfaces $\psurf{i}{j}$, $\psurf{i}{j+1}$ must join across $\mathpzc{i}_j$. 
\end{proof}
In the above proof we have worked under the assumption that $\ts_1$ and $\ts_2$ are single connected components. More generally one could have regions $\mathscr{J}_1$ and $\mathscr{J}_2$ at least one of which has minimax surface which consists of multiple connected components. Such a situation needs to be considered already
in the proof of SSA. However, since no single surface can have closely lying or nearly intersecting components, one would only ever have pairwise intersections of time-sheets without them coming dangerously close to one another. Thus, while we have not spelled out the details of the proof for multiple connected components, we suspect that it should go through by applying the same procedure to each connected component individually.

We are now going to prove a theorem that holds for a general pair $\mathscr{J}_1,\mathscr{J}_2$, whose time-sheets are not necessarilty connected, and that ensures that any pair of connected components of a cooperating configuration is minimally intersecting. This defines for us a more general notion of a \emph{minimally intersecting configuration}, that applies even for disconnected time-sheets. We will see a further  generalization of this definition, for arbitrarily many regions, in section \ref{sec:conjecture}.

\begin{theorem}\label{thm:min-int}
Given a crossing pair of boundary regions $\mathscr{J}_1,\mathscr{J}_2$, any cooperating time-sheet configuration $\conf^{\star} \in \coopset$ is minimally intersecting, i.e.\ $\conf^{\star}\in \mintset$.
\end{theorem}

\begin{proof}
    Consider a cooperating configuration $\conf^{\star} = (\ts_1, \ts_2)$ that is not minimally intersecting. Then, there exists a pairwise cell that does not reach the conformal boundary, meaning that there is a cell bounded by exactly two partial time-sheets $\ts_1^{(\alpha)}$ and $\ts_2^{(\beta)}$ (for some fixed $\alpha$ and $\beta$).
    We consider the case where we have one such cell, but the proof extends trivially to the case were we have multiple ones (see figure \ref{fig:multiple-intersections}). We will now show that a non-minimally intersecting cooperating configuration is inconsistent globally.
    
    Suppose that $|\gamma_1^{+(\alpha)}| > |\gamma_2^{+(\beta)}|$. Then, consider the time-sheet $\tilde{\ts}_1$ made by taking $\ts_1$ and replacing $\ts_1^{(\alpha)}$ with $\ts_2^{(\beta)}$. This time-sheet will be in the same homology class as $\ts_1$. Note however that it is not a minimax time-sheet since we have deformed it near the minimax surface.\footnote{Since the configuration is supposed to be cooperating, the two minimax surfaces $\minimax_1$ and $\minimax_2$ must be contained in $\ts_1^{(\alpha)}$ and $\ts_2^{(\beta)}$ respectively, for if not there would be  partial maximal achronal surfaces on  $\ts_1^{(\alpha)}$ and $\ts_2^{(\beta)}$ that are not the corresponding partial minimax surfaces.}
    Since the configuration was cooperating, the maximal (disconnected) surface $\gamma$ on $\tilde{\ts}$ will consist of $\minimax_1$ except on $\ts_2^{(\beta)}$, where it will contain the maximal segment $\gamma_2^{+(\beta)}$, more explicitly
    \begin{equation}
        \gamma := (\minimax_1 \setminus \gamma_1^{+(\alpha)}) \cup \gamma_2^{+(\beta)}.
    \end{equation}
    This surface will have less area than $\minimax_1$ since $|\gamma_1^{+(\alpha)}| > |\gamma_2^{+(\beta)}|$. Now, consider the maximal connected surface $\tilde{\gamma}^{+}$ on $\tilde{\ts}$. This surface will have less area than the maximal disconnected surface $\gamma$ since being connected is a stricter constraint. All in all, we have
    \begin{equation}\label{eq:min-int-ineq}
        |\tilde{\gamma}^{+}| \leq |\gamma| < |\minimax_1|.
    \end{equation}
    So, we have found a new minimal maximal surface contradicting the fact that $\ts_1$ was a minimax time-sheet. If $|\gamma_1^{+(\alpha)}| < |\gamma_2^{+(\beta)}|$, we repeat the same argument but for the other time-sheet instead. Finally, consider the case $|\gamma_1^{+(\alpha)}| =|\gamma^{+(\beta)}|$. We repeat the same exact logic as above, reaching the final inequality $|\tilde{\gamma}^+| \leq |\gamma| = |\minimax_1|$. If $|\tilde{\gamma}^+| < |\minimax_1|$, we arrive at the same contradiction as above. If $|\tilde{\gamma}^+| = |\minimax_1|$, we have an additional minimax surface, but one with sharp corners where the time-sheet $\ts_1^{(\alpha)}$ is sewn together with $\ts_2^{(\beta)}$. Thus, we can always form a new surface with smaller area by rounding off these corners, again contradicting the minimality of $\minimax_1$. 
\end{proof}

\section{The  cooperating conjecture}\label{sec:conjecture}

Having introduced the minimax prescription, we can now explore some of its interesting consequences and applications. Contrary to other covariant constructions, here the full spacetime plays a central role. The homology condition, which is usually enforced through a spatial slice, here takes the form of a bulk homology. We begin in subsection  \ref{sec:graph-model} by presenting the construction of a spacetime graph model for time-dependent states, and explain how a generalization of the cooperating property presented in \ref{sec:coop} is sufficient for this graph to be equivalent to the RT graph. In subsections \ref{sec:favor} and \ref{sec:against} we explore the validity of this cooperating conjecture --- showing that it's consistent with boundary causality and with the NEC for a class of explicit examples, and comment on some challenging scenarios. We end in \ref{sec:geometric-proof} by providing additional intuition for the cooperating property, explaining how all of the holographic entropy inequalities are geometrically proved with it. The geometric picture makes it clear how minimax resolves the conceptual issues regarding the HRT surfaces not lying on a common bulk slice.

\subsection{A graph model for time-dependent states}\label{sec:graph-model}

\begin{figure}
    \centering
    \includegraphics[width=\textwidth]{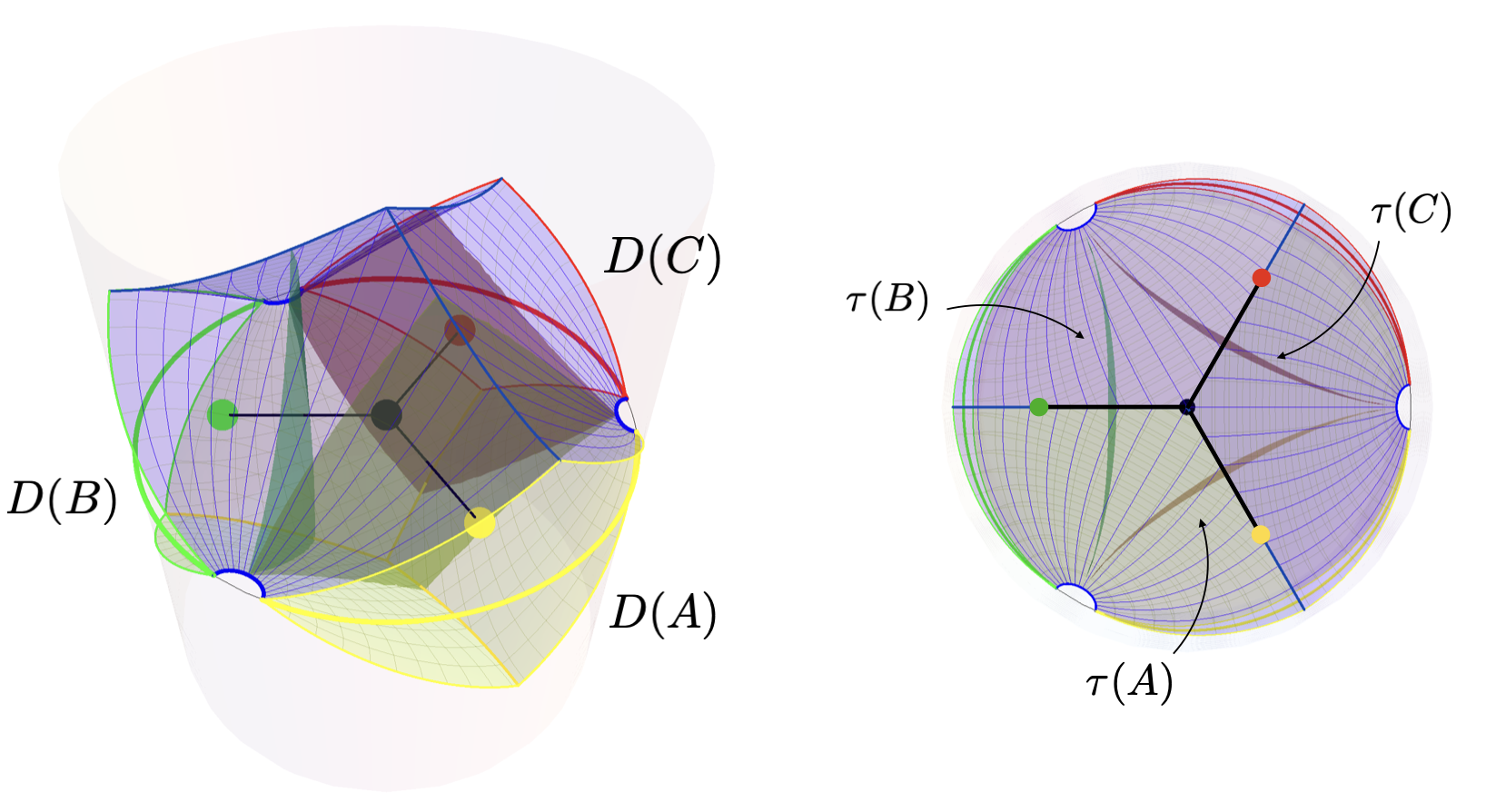}
    \caption{Spacetime graph inside the regulated spacetime for three boundary regions (with their domains of dependence on the boundary shown in red, green and yellow). {\bf Left:} We show the full spacetime, the three time-sheets are shown with the respective colors and they end on the joint HRT surface $\hat{\gamma}_{\text{HRT}}(ABC)$ and on $\mathcal{I}$ (shown by the ``tent'' of null congruences shot from $\hat{\gamma}_{\text{HRT}}(ABC)$). In this simple example the time-sheets do not intersect, and there are only four spacetime cells formed by the union and intersections of the various homology regions: $R(A)$ (yellow node), $R(B)$ (green node), $R(C)$ (red node) and $R^c(A)\cap R^c(B) \cap R^c(C)$ (black node). {\bf Right:} Same figure but shown from the top for additional clarity.}
    \label{fig:spacetime-graph}
\end{figure}

As presented in section \ref{sec:intro}, the minimax prescription suggests a way to construct a spacetime graph. Fix a boundary Cauchy slice with $\mathsf{N}$ boundary regions specified. Let $\conf$ be a set of minimax time-sheets for all elementary and composite regions. The time-sheets of $\conf$ partition the bulk into connected codimension-0 cells. We associate to each cell a vertex in the graph. If two cells share a common partial time-sheet, we connect their vertices with an edge $e$, which is assigned a weight $w(e)$ equal to the area of $\minimax(e)$, the part of the minimax surface lying on the shared partial time-sheet $\ts(e)$. Cells adjacent to the boundary are called boundary vertices; each one is associated with one elementary region. Let $G(\conf)$ be the resulting weighted graph. See figure \ref{fig:spacetime-graph} for the construction of a spacetime graph for a simple example with non-intersecting time-sheets.

Now that we have defined a spacetime graph, we can compute minimal cuts on it. A \textit{cut} $U$ of $G(\conf)$ is a subset of its vertices. The \textit{cut of edges} $\mathscr{C}(U)$ is the set of all edges with one end in $U$ and the other in $U^{c}$; its total weight is the sum of their weights:
\begin{equation}
\|\mathscr{C}(U)\| := \sum_{e \in \mathscr{C}(U)} w(e).
\end{equation} 
A cut $U$ is said to be \emph{homologous} to a composite region $\mathscr{J}$, denoted $U \sim \mathscr{J}$, if 
the boundary vertices in $U$ are precisely the boundary vertices associated to $\mathscr{J}$. One can interpret $\mathscr{C}(U)$ and $U$ as the graph-theoretic analogs of a surface and its corresponding homology region respectively. We call the minimal weight for a given region its graph entropy $S_{G(\conf)}(\mathscr{J})$:
\begin{equation}\label{eq:discrete-entropy}
S_{G(\conf)}(\mathscr{J}):=    \min_{U \sim \mathscr{J}} \| \mathscr{C}(U)\|\,.
\end{equation}

Given that the graph $G(\conf)$ is constructed out of minimax time-sheets and surfaces, each of which computes the HRT entropy of a region $\mathscr{J}$, one might expect that the graph entropy equals the HRT entropy. Indeed, the time-sheet for $\ts(\mathscr{J})$ in the configuration $\conf$ corresponds tautologically to a particular cut $U_{\ts(\mathscr{J})}\sim\mathscr{J}$, which simply includes all vertices corresponding to cells in the homology volume for $\ts(\mathscr{J})$. This however only implies a bound,
\be\label{Sgraph}
S_{G(\conf)}(\mathscr{J})\le S(\mathscr{J})\,,
\ee
as the graph could include a shortcut: a cut $U\sim \mathscr{J}$ smaller than $U_{\ts(\mathscr{J})}$.

What we will now show is that the cooperating condition on $\conf$ forbids shortcuts, ensuring equality in \eqref{Sgraph}. We already defined the cooperating property for pairs of time-sheets in section \ref{sec:coop}. We will now define it for general sets of time-sheets.

\begin{definition}[Cooperating time-sheet configuration]\label{def:cooperating}
A configuration of minimax time-sheets is said to be cooperating if the maximal achronal surface on each partial time-sheet is the corresponding partial minimax surface.
\end{definition}

\begin{theorem}\label{thm:cooperating}
If the time-sheet configuration $\conf$ is cooperating, then for any boundary region $\mathscr{J}$,
\be\label{Sgraphequal}
S_{G(\conf)}(\mathscr{J})= S(\mathscr{J})\,.
\ee
\end{theorem}
\begin{proof}
Following \cite{Bao:2015bfa}, we construct a  mapping from graph cuts homologous to $\mathscr{J}$ to time-sheets homologous to $D(\mathscr{J})$. Let $U \sim \mathscr{J}$. For each vertex $u \in U$ there is an associated spacetime cell $R(u)$. Thus, to the entire cut $U$, there is an associated spacetime volume
\begin{equation}
    R(U) := \bigcup_{u \in U} R(u)\,.
\end{equation}
The corresponding time-sheet is then defined as the bulk part of its boundary $\partial R(U)$,
\begin{equation}\label{eq:J-cut-time-sheet}
    \ts(U) := \bigcup_{(u,u') \in \mathscr{C}(U)} \left( R(u) \cap R(u') \right) = \bigcup_{e \in \mathscr{C}(U)}  \ts(e)\,,
\end{equation}
which is by construction in the correct homology class for $\mathscr{J}$. Since the weight of an edge has been defined as the area of the partial minimax surface on the partial time-sheet, the total weight of the cut $U$ quals the area of a (possibly disconnected) surface made of various portions of minimax surfaces along the time-sheet $\ts(U)$:
\be
\|\mathscr{C}(U)\|=\sum_{e \in \mathscr{C}(U)}|\minimax(e)|\,.
\ee

This is where the cooperating property enters the proof: if each partial minimax surface $\minimax(e)$ is also the maximal segment on the corresponding partial time-sheet, then the total weight $\|\mathscr{C}(U)\|$ is greater than or equal to that of the maximal achronal surface on $\ts(U)$, since it is being maximized subject to a weaker constraint (each partial surface being achronal, versus their union being achronal):
\be
\|\mathscr{C}(U)\| \ge \max_{\surf\in\ts(U)}|\surf|\,,
\ee
where the maximum is over achronal surfaces in $\ts(U)$. Since this is true for all cuts $U\sim\mathscr{J}$, the 
graph entropy cannot be smaller than the 
HRT area:
\be\label{coopineq}
    S_{G(\conf)}(\mathscr{J})\ge S(\mathscr{J})\,.
    \ee
Given \eqref{Sgraph}, this implies \eqref{Sgraphequal}.   
\end{proof}

\subsection*{The cooperating conjecture}

\begin{figure}
    \centering
    \includegraphics[width=\linewidth]{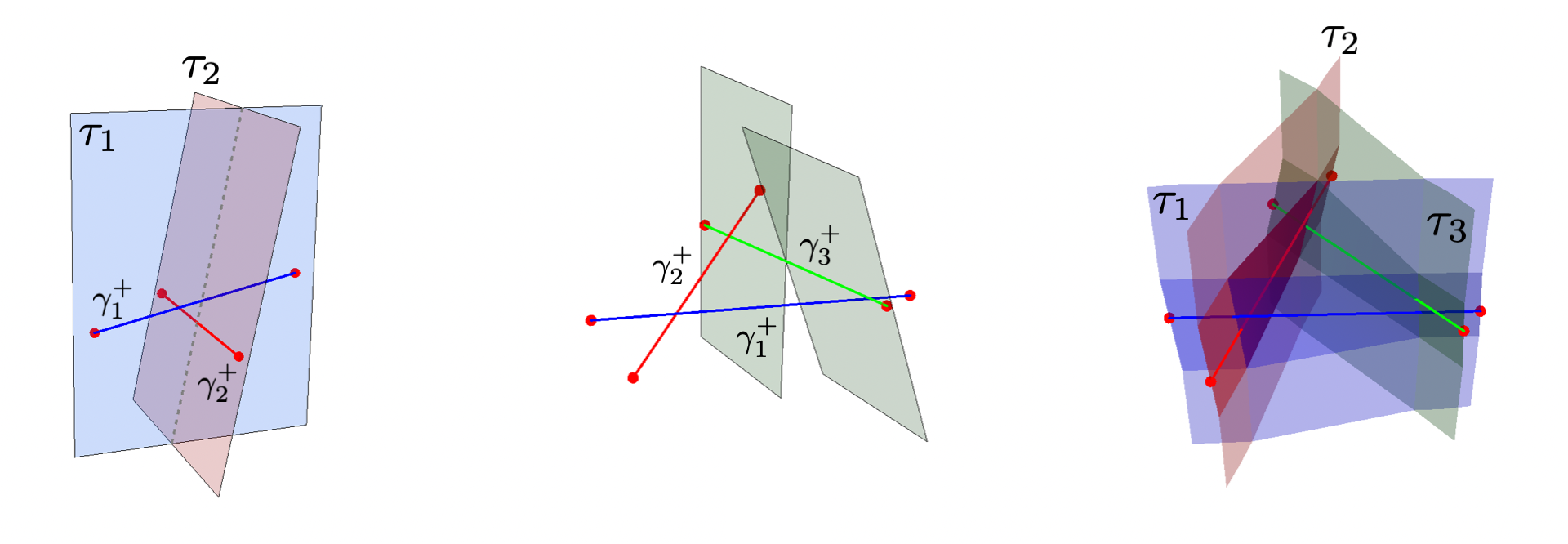}
    \caption{Toy example in Minkowski spacetime illustrating the jump in complexity from two to three intersecting time-sheets. {\bf Left:} Two boosted minimax surfaces $\minimax_1$ (blue) and $\minimax_2$ (red) which do not intersect. To make the configuration cooperating one imposes orthogonality between the surfaces and the intersection locus (dashed gray line). Two planar time-sheets are enough to meet this constraint. {\bf Middle:} Three boosted minimax surfaces $\minimax_1$ (blue), $\minimax_2$ (red) and $\minimax_3$ (green). The orthogonality constraint now gives six boosts for only three intersections: planar time-sheets (here we show only one of them in green) cannot be made to work. {\bf Right:} The configuration can be made cooperating but one needs more degrees of freedom, allowing the time-sheets to bend. (In fact, here we can achieve a cooperating configuration by piecewise-planar sheets; we indicate the pieces by change in the shading.)}
    \label{fig:coop-in-R3}
\end{figure}

Having motivated the need for a generalization of theorem \ref{thm:coop-2party}, we now proceed to describe it in more detail. The notation below will be the same as that introduced in subsection \ref{sec:coop}, in particular $\confset$ is the set of configurations $\conf$, $\coopset$ is the set of cooperating configurations and $\mintset$ is the set of minimally intersecting configurations. However, their definitions must be adapted to accomodate for arbitrarily many time-sheets. Let $\mathscr{R} \subset \pwset$ be a (not necessarily complete) collection of $\textsf{r} \coloneq |\mathscr{R}|$ boundary regions.\footnote{Just like in the 2-party case, it is only interesting to consider the case where there are crossings. Here, however, we do not make further assumptions about the regions (in particular, they can be composed of disconnected components).} Denote with $\tsset_{\mathscr{I}}^{\,+}$ the set of minimax time-sheets for some region $\mathscr{I}$, then we define
\begin{equation}\label{eq:configuration}
  \confset_{\mathscr{R}} := \prod_{\mathscr{I} \in \mathscr{R}}\tsset_{\mathscr{I}}^{\,+} = \{(\ts_{\mathscr{I}_1}, \ts_{\mathscr{I}_2},\dots, \ts_{\mathscr{I}_{\textsf{r} }}) \ \ |\ \   \ts_{\mathscr{I}_i} \in \tsset_{\mathscr{I}_i}^{\,+} \text{ for every } i \in [\textsf{r} ]\}
\end{equation}
to be the set of configurations $\conf$ for the regions specified by $\mathscr{R}$ expressed as $\textsf{r}-$tuples of minimax time-sheets. As an example, choose $\mathscr{R} = \{\{1,2\},\{2,3\},\{3,4\}\}$ corresponding to choosing regions $AB$, $BC$ and $CD$ on the boundary. Then, a configuration will be given by a triplet $(\ts_{AB}, \ts_{BC},\ts_{CD}) \in \tsset_{AB}^{\,+}\times \tsset_{BC}^{\,+}\times\tsset_{CD}^{\,+}$. See figure \ref{fig:coop-3-party} for a drawing of this example.

A key distinction with the setup from subsection \ref{sec:coop} is that, in this case, we typically have partial time-sheets that do not reach the asymptotic boundary, and are instead bounded by the intersection with two (or possibly more) other time-sheets. This leads to a non-trivial increase in complexity for cooperation.

To better appreciate this increase in complexity, let us consider a simple example in Minkowski space $\mathbb{R}^{2,1}$. First consider the case of only two intersecting time-sheets, each containing one HRT surface (which will be straight segment). Recall that a necessary condition for the partial minimax surfaces to be maximal on their respective time-sheets is that the surfaces meet each intersection orthogonally. Orthogonality is imposed by requiring that the tangent vectors on the HRT surfaces be perpendicular to the tangent vectors at the intersection loci. This constraint leads us to two boosts for two time-sheets: so a cooperating configuration can always be achieved 
by choosing simple planes for the two time-sheets. Now let us repeat the same problem with three pairwise-intersecting time-sheets.
Each intersection will give us a pair of boosts, so in total we will have six boosts for only three time-sheets. This is an overconstrained problem if we restrict to using only planes, so we need to allow for the time-sheets to bend. See figure \ref{fig:coop-in-R3} for a visual illustrating this example.

To make matters worse, a general configuration $\conf \in \confset_{\mathscr{R}}$ may have time-sheets which intersect in complicated ways. In fact, two different configurations for the same boundary regions may differ in terms of numbers of time-sheet intersections, type of intersections (timelike, spacelike or null) and even codimension of intersections. Of course, configurations will all be similar in a neighborhood of the minimax surfaces $\minimax_{\mathscr{I}_i}$, since the time-sheets are constrained to contain its minimax surface (with the minimax surface maximal on it). To simplify this large variability, we want to find a similar definition for the set $\mintset$ of minimally intersecting configurations introduced in subsection \ref{sec:coop}.
\begin{definition}[Minimally intersecting time-sheets]
A configuration is said to be minimally intersecting if for every pair of time-sheets in the configuration, any two connected components partition the bulk into cells that extend to the conformal boundary.
\end{definition}
Again, the intuition behind this definition is to have a configuration that looks and behaves analogously to RT surfaces on a bulk time slice, which may intersect each other (depending on the boundary region), but  will do so at most once for any pair of connected components (as a consequence of minimality). We note that the set $\mintset_{\mathscr{R}}$ is never empty, as the set of entanglement horizons for the regions in $\mathscr{R}$ is always a minimally intersecting configuration. This is because once a set of null generators crosses another set, they cannot weave back to intersect multiple times. However, although a configuration comprising entanglement horizons lies in $\mintset$, it does not lie in $\coopset$, because any surface $\minimax_{\mathscr{J}}$ which crosses a horizon $\ts_{\mathscr{I}}$ (with $\mathscr{I} \ne \mathscr{J}$) fails orthogonality.

While minimally intersecting configurations share some properties with collections of RT surfaces, to draw complete analogy to the static case and to allow for the construction of the spacetime graph model we need the (much stronger) cooperating constraint specified in the definition \ref{def:cooperating}. Our main conjecture is a generalization of theorem \ref{thm:coop-2party} for an arbitrary configuration:

\begin{conjecture}[Cooperating conjecture]\label{con:cooperating-conjecture}
Any choice $\mathscr{R}$ of boundary regions admits a cooperating configuration of time-sheets $\mathfrak{t} \in \mathfrak{T}_{\mathscr{R}}$.
Moreover, any cooperating configuration is minimally intersecting, 
\begin{equation}
\emptyset \ne \coopset_\mathscr{R}\subset\mintset_\mathscr{R} \ .
\end{equation}
\end{conjecture}

\subsection*{Comments}

The first comment is the most obvious: it is a conjecture and therefore we do not, at the moment, have a general proof. It is however possible to test the conjecture for a class of configurations, which we review in the remainder of the paper. In particular, we focus on configurations for which cooperation naively appears difficult to achieve, choosing examples with nearby pairs of intersecting (or nearly intersecting) HRT surfaces which are highly boosted relative to one another. The most natural choice of time-sheets for these configurations leads to seams that move towards or away from each other, in a manner seemingly in tension with the local maximality condition for cooperation. We find that in the first class of configurations studied below, a more clever choice of time-sheets resolves this tension. In the second class, however, we are unable to find a suitiably clever choice, leaving cooperation still apparantly difficult to achieve. Due to these challenging configurations, the validity of the conjecture is in question. However, we expect that it may still hold true, either with some mild modifications or through a more complex construction of cooperating time-sheets.

On the other hand,
the proof of theorem \ref{thm:min-int} regarding minimal intersecting configurations can be easily recycled for the more general case. One simply applies the same argument individually to any pair of connected time-sheets that intersect, forming bounded cells that do not extend to the conformal boundary. This guarantees that if a cooperating configuration exists, it will be minimally intersecting. The minimally intersecting property is necessary for the spacetime graph model to reduce in the static case to the same topology as the RT graph model, so this is a useful (though not strictly necessary\footnote{For the purpose of reproducing a given holographic entropy vector using the more compact representation in terms of a graph model, the graph is actually not unique:  one can easily add extra vertices and edges which don't participate in any min-cut and therefore preserve the entropy vector.  We suspect that a time-sheet configuration which is optionally non-minimally-intersecting would merely correspond to such a reducible graph.}) result.

To prove the cooperating conjecture, one might be tempted to apply theorem \ref{thm:coop-2party} to all possible pairings of time-sheets in the configuration. This of course can be done, but it will not get us far into a proof. In a general configuration there are partial time-sheets that do not extend to the conformal boundary and their shape cannot be fixed by decoupling the multiple constraints, i.e.\ if $(\ts_1^*, \ts_2^*)$ is a cooperating pair and $(\ts_2^\dagger, \ts_3^\dagger)$ is also cooperating, it is not at all clear whether $\ts_2^*$ and $\ts_2^{\dagger}$ are compatible solutions.
See figure \ref{fig:coop-3-party} for an illustration.
    
Finally, we remark that the conjecture is a dynamical statement. Therefore, a proof should make use of standard holographic assumptions such as the null energy condition (NEC), the Einstein equations, and AdS boundary conditions.

\begin{figure}
    \centering
    \includegraphics[width=0.9\linewidth]{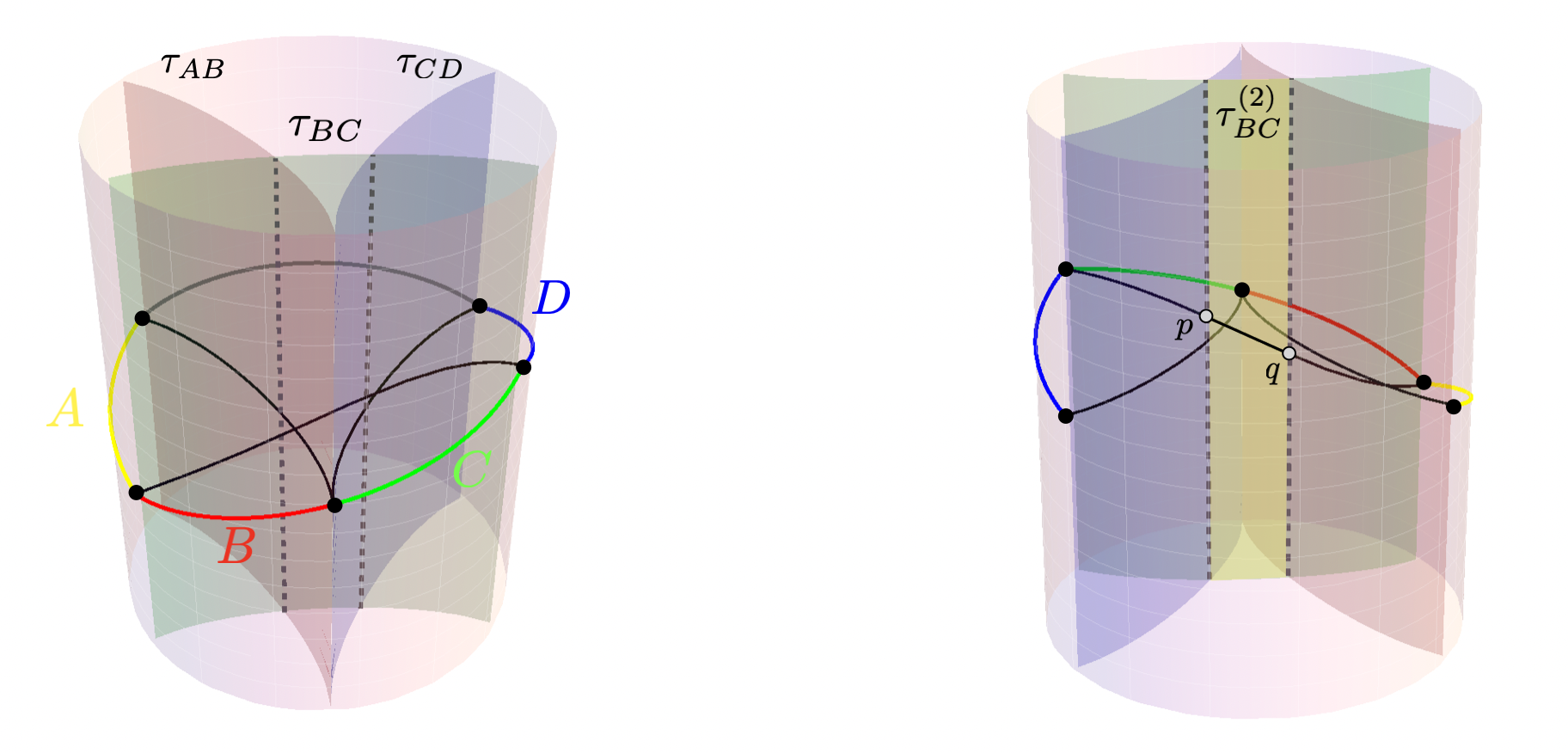}
    \caption{Two viewpoints of a configuration $(\ts_{AB}, \ts_{BC}, \ts_{CD})$ of minimax time-sheets for a collection of 3 boundary regions $AB$, $BC$ and $CD$. We have showed the construction in the unregulated spacetime for clarity. {\bf Left:} The regions $A$ (yellow), $B$ (red), $C$ (green) and $D$ (blue) are boosted on the boundary and they lay on a common Cauchy slice (we show in gray the purifier region). In black we show the three minimax surfaces $\minimax_{AB}$, $\minimax_{BC}$ and $\minimax_{CD}$. The time-sheets intersect twice (shown as dashed gray lines), one between $\ts_{AB}$ and $\ts_{BC}$ and the other between $\ts_{BC}$ and $\ts_{CD}$. {\bf Right:} Same configuration but as viewed from the back. Contrary to the 2-party scenario shown in figure \ref{fig:coop-2-party}, now there is a partial time-sheet $\ts_{BC}^{(2)}$ (shaded in yellow) which does not extend to the conformal boundary, and it is bounded by the two interection seams. On it there is a partial minimax surface going from $p$ to $q$. Therefore, on $\ts_{BC}^{(2)}$, the maximization problem of finding the maximal achronal segment imposes 
    boundary conditions on both endpoints. Indeed, 
    the configuration shown is not cooperating, since the $pq$-portion  on $\ts_{BC}^{(2)}$ does not meet the intersection seams orthogonally.}
    \label{fig:coop-3-party}
\end{figure}

\subsection{Evidence in favor of the conjecture}\label{sec:favor}

The only scenario where we have full control is one where the configuration is made of exactly two connected time-sheets, which has been already discussed in subsection \ref{sec:coop}. There, one can prove the conjecture in full generality. In the absence of a proof for the general case, we can test the conjecture by considering setups with more than two intersecting time-sheets of varying degree of complexity. To this end, in the remainder of this section, we restrict to $2+1$ dimensions (where computations are tractable) and give strong evidence that the conjecture holds for a large family of configurations and spacetimes. 

The family can be best understood by starting with its simplest member. This consists of three HRT surfaces in 2+1 dimensions, where the first one intersects the other two orthogonally, and those two are boosted relative to each other. For concreteness, we take the first HRT surface to lie in the $t=0$ slice and be associated to exactly half of the boundary. We call this configuration the \emph{boosted H}, 
as pictorially suggestive from the configuration, cf.\  the left panel of figure \ref{fig:spine}.
The six endpoints of the three HRT surfaces must be mututally spacelike in order to be placed on a common boundary Cauchy slice. As a result, there is a limit to how much the two HRT surfaces can be boosted, as we will cover in more detail below.\footnote{
    Note that this does not guarantee that the three HRT surfaces themselves can be placed on a common \emph{bulk} Cauchy slice, and even when they can, this Cauchy slice need not be a maximin slice, in which case we would not achieve automatic reduction to an RT-type scenario.  This makes the present family a non-trivial testing ground for our minimax construction.
} 
The objective is to show that this configuration can be made cooperating. We will not provide an explicit construction for the three time-sheets, but will instead 
describe them by  their intersections, represented by timelike curves. 
These two intersections will generate precisely seven partial surfaces: six that connect a boundary endpoint to an intersection seam of the time-sheets, and one middle segment along the static spine bounded by the two intersection seams. Here, one can give strong evidence in support for the conjecture by explicitly showing local maximality of these partial surfaces. 

A necessary condition for the partial minimax surfaces to be maximal is for them to meet the intersection seam orthogonally. So given a point and a tangent vector, we are instructed to find an acceleration for the two seams. A natural choice for the intersection seams is to chose them to be timelike geodesics, since they naturally converge in AdS. We will show that this ansatz suffices.

Because of the converging property of timelike geodesics in AdS, proving the conjecture for the boosted H allows for a generalization of the setup to an infinite family of configurations, where the first HRT surface (which we call the \emph{spine}) is once again chosen to lie in the $t=0$ slice and be associated to exactly half of the boundary, while arbitrarily many HRT surfaces (which we call the \textit{ribs}) intersect the spine orthogonally and are relatively boosted with one another, cf.\  the right panel of figure \ref{fig:spine}. Though we have convinced ourselves of the above intuition, in the computations below we are going to restrict to the case of two ribs without performing any explicit computation for the general case.

In the following subsections, we test the cooperating conjecture for the boosted H family in two bulk spacetimes: (1) pure AdS$_3$, and (2) asymptotically AdS$_3$ with spherically symmetric matter that satisfies the null energy condition.

We will show that local maximality is guaranteed only when (i) boundary causality holds and (ii) the bulk energy-momentum tensor satisfies the NEC. We interpret this to mean that, whatever GR theorem this cooperating property corresponds to, it seems to be highly sensitive to the dynamical properties that we know the covariant holographic entanglement entropy
obeys. We consider this a non-trivial evidence in support of the conjecture.

\begin{figure}
    \centering
    \includegraphics[width=0.9\linewidth]{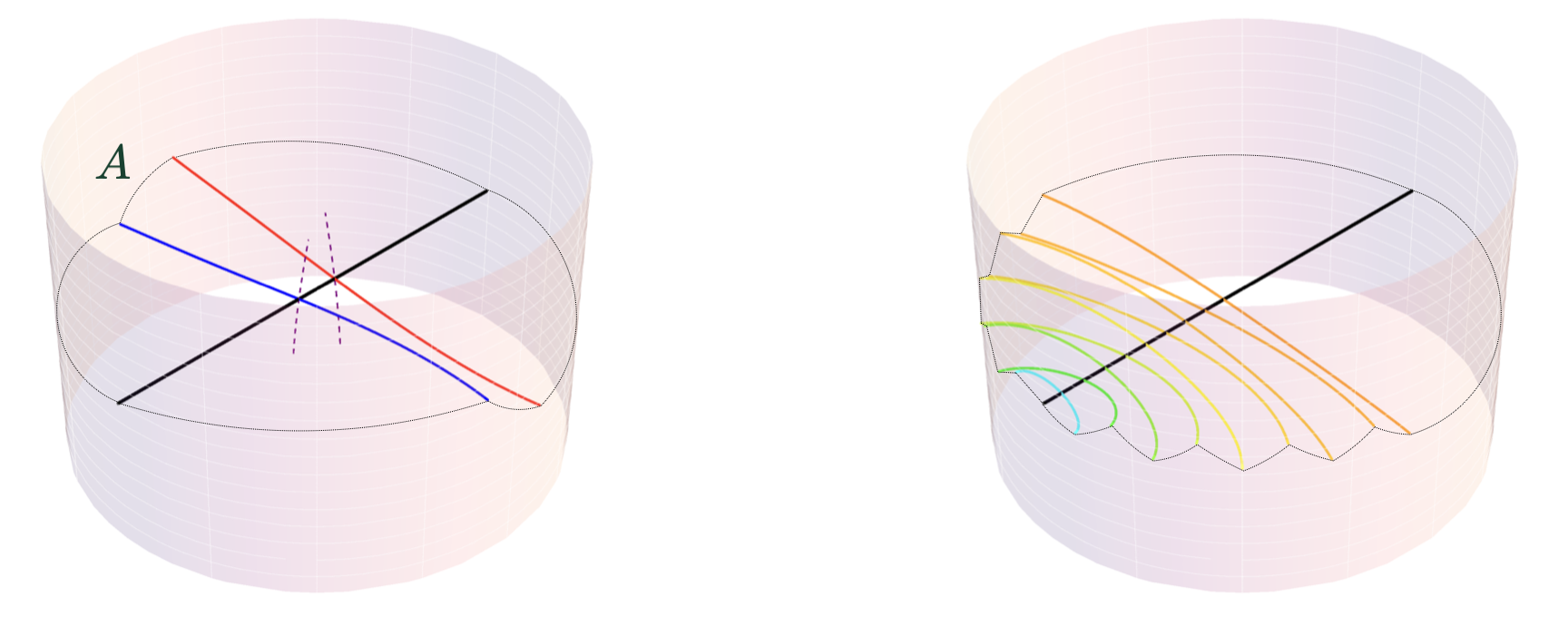}
    \caption{{\bf Left:} The boosted H configuration in global AdS$_3$. In black we depict the static RT surface that covers half the boundary, while the blue and red surfaces represent the other two HRT surfaces boosted relative to each other. The amount of boost is constrained by boundary causality, all the boundary regions must lie on a common boundary Cauchy slice (shown as the dashed line). For nearby ribs, this amounts to enforcing that the region $A$ (and its opposite counterpart) is spacelike. The dashed purple lines indicate our ansatz for the intersection seams of the three time-sheets (not shown). The intersections, being timelike geodesics, converge toward one another allowing the configuration to cooperate. {\bf Right:} A generalization of the boosted H configuration in global AdS$_3$ for nine ribs shown in various colors.}
    \label{fig:spine}
\end{figure}

\subsubsection{Pure AdS$_3$}
The setup is AdS$_3$ in global coordinates $(t, \theta, \phi)$ with line element
\be
\dd s^2 = \frac{1}{\cos^2{\theta}}\left(-\dd t^2 + \dd \theta^2 + \sin^2{\theta}\,\dd\phi^2\right).
\ee
The first HRT surface (the spine) will be taken to be static, anchored at the $\phi = 0$ and $\phi = \pi$ points on the boundary. The other two (the ribs) will be composed of nearby HRT surfaces for boundary regions with an opening angle $\Delta \phi \in (0,\pi)$ centered at $\phi = 0$ relatively boosted with one another. 

More precisely, the configuration is constructed as follows -- we refer the reader to appendix \ref{app:geodesic_seams} for a detailed computation of the spacelike and timelike geodesics. Begin by defining a static rib with opening angle $\Delta \phi$. We now boost the boundary region with some boost parameter $\chi$, so the region's endpoints have time separation $\Delta t = \Delta \phi \tanh{\chi}$. The second rib is constructed in a similar manner, by enlarging the region to have an initial opening angle $\Delta\phi' = \Delta\phi + \delta \phi$. We boost the boundary region oppositely by $-\chi$, obtaining $\Delta t' = -\Delta\phi' \tanh{\chi}$. The boost parameter $\chi$ cannot be arbitrarily large as it is constrained by boundary causality; the six endpoints must be mutually spacelike in order to be placed on a common boundary Cauchy slice, see figure \ref{fig:spine}. We denote the maximum boost allowed by causality  with $\chi_{\text{c}}$, 
and it reads\footnote{If the two ribs are near each other, which is the scenario we are considering, this constraint is enforced by demanding the region bounded by the ribs' endpoints (labeled $A$ in figure \ref{fig:spine}) 
 to be spacelike.}
\be\label{eq:causality-bound-AdS3}
\chi_{\text{c}} = \frac{1}{2}\log\left(\frac{\Delta\phi'}{\Delta\phi}\right).
\ee
We now need to define the timelike intersection seams of the time-sheets. As mentioned, we will show that timelike geodesic seams, shot orthogonally from the intersection point of the ribs with the spine, will suffice to make the configuration cooperating. By finding the tangent vectors to the spine and one of the two ribs at the intersection point $(0,\,\theta_0,\,0)$, the orthogonal timelike vector is found to be  
\be\label{eq:ortho-vector}
\left(\frac{1+\cos{\Delta\phi}}{1+\cos{\Delta t}}\tan\frac{\Delta \phi}{2},\ 0 ,\ \tan\frac{\Delta t}{2}\right).
\ee
From this, one solves for the timelike geodesic $\gamma_{\perp}$ to quadratic order in the affine parameter $\lambda$, obtaining a parameterization $( t(\lambda),\theta(\lambda), \phi(\lambda))$ which reads
\be
\gamma_{\perp}^{\mu}(\lambda) = \left(\frac{1+\cos{\Delta\phi}}{1+\cos{\Delta t}}\tan\frac{\Delta \phi}{2}\,\lambda, \,  \theta_0 - \left(\frac{(t_1^2 - \phi_1^2) \tan{\theta_0}}{2}\right)\,\lambda^2, \, \tan\frac{\Delta t}{2}\lambda \right).
\ee
with $t_1$ and $\phi_1$ are the velocities at $\lambda = 0$ given in \eqref{eq:ortho-vector}.

\begin{figure}
    \centering
    \includegraphics[width=\linewidth]{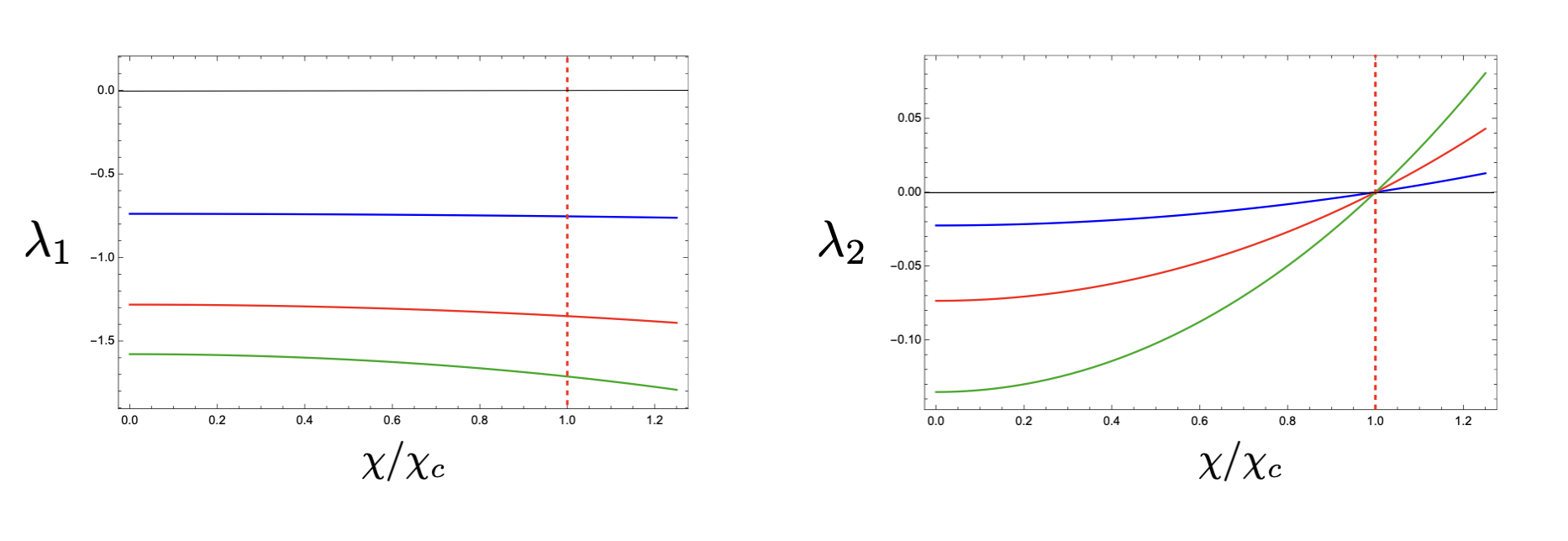}
    \caption{Plots of the two eigenvalues $\{\lambda_1,\lambda_2\}$ of the Hessian $\mathbf{H}(\Delta(0,0))$ as a function of the normalized boost $\chi/\chi_{c}$ for the boosted H configuration with parameters $\Delta\phi_0 = \pi/2$ (blue), $\Delta\phi_0 = \pi/3$ (red), $\Delta\phi_0 = \pi/4$ (green)  and $\delta\phi = \pi/5$ for all of them. The first eigenvalue is never positive, while the second one changes sign precisely at $\chi = \chi_{c}$. }
    \label{fig:AdS3-coop-plots}
\end{figure}

We now move to the issue of cooperation. We need to check that all seven partial HRT surfaces are locally maximal with respect to varying the two endpoints along the timelike seams. To do that we define the geodesic distance $\Delta$ between two points $(t_1, \theta_1 ,\phi_1)$ and $(t_2, \theta_2 ,\phi_2)$ in AdS$_3$,
\be\label{eq:distance-fn}
\Delta := \cos{(t_1 - t_2)}\sec{\theta_1}\sec{\theta_2} -\cos{(\phi_1-\phi_2)}\tan{\theta_1}\tan{\theta_2}.
\ee
We begin by considering the partial surface on the spine with one endpoint on the boundary located at $(0,\theta_\infty, 0)$ (with $\theta_\infty$ some cutoff radius) and the other endpoint on the intersection with the first timelike seam at $(0,\theta_0,0)$. To show maximality, we compute the distance $\Delta(\lambda)$ between the fixed boundary endpoint and any other point on the timelike geodesic $\gamma_{\perp}^{\mu}(\lambda)$, and look at the sign of $\Delta''(0)$. We obtain
\be
    \Delta''(0) = \sec{\theta_\infty}\cos^2\frac{\Delta \phi}{2} \sec^2\frac{\Delta t}{2}\sqrt{\frac{1 +\cos{\Delta t}}{\cos{\Delta t}-\cos{\Delta \phi}}} \left[\sqrt{\frac{1 +\cos{\Delta \phi}}{1+\cos{\Delta t}}}-1\right],
\ee
which is negative whenever $\Delta \phi \geq \Delta t$, i.e.\ whenever the region $(\Delta\phi, \Delta t)$ is spacelike. This is true for any value of $\chi$, so this partial surface (and by symmetry also the opposite one along the spine) poses no threat for cooperation. We now consider one of the four boundary anchored partial surfaces of the ribs. The only difference is the boundary endpoint now located at $\left(\frac{\Delta t}{2}, \theta_\infty, \frac{\Delta \phi}{2}\right)$, otherwise the computation is exactly the same, leading to
\be
    \Delta''(0) = \cos^2{\frac{\Delta \phi}{2}}\sec{\frac{\Delta t}{2}}\left[\cos{\frac{\Delta \phi}{2}}\sec{\frac{\Delta t}{2}}\tan{\theta_\infty}\sqrt{\frac{1 +\cos{\Delta\phi}}{\cos{\Delta t}-\cos{\Delta \phi}}} - \sec{\theta_\infty}\sqrt{\frac{1 +\cos{\Delta t}}{\cos{\Delta t}-\cos{\Delta \phi}}} \ \right]
\ee
which is negative whenever $\Delta \phi \geq \Delta t$, as before. The other three boundary anchored partial ribs have similar expressions.

The only remaining (and also the less trivial) partial surface to check is the middle component that is bounded by the two intersection seams. In particular, we must show that the distance function $\Delta(\lambda_1,\lambda_2)$ between the two timelike geodesics $\gamma^{\mu}_{\perp}(\lambda_1)$ and $\gamma^{\mu}_{\perp}(\lambda_2)$ attains a maximum at $(0,0)$. This can be checked explicitly with Mathematica by performing an analysis of the Hessian $\mathbf{H}(\Delta(\lambda_1,\lambda_2))$ at that point, i.e.\ we check if it is negative-definite. We find that maximality is guaranteed whenever $\chi \leq \chi_{c}$. See figure \ref{fig:AdS3-coop-plots} for the numerics for different configurations.

\subsubsection{Spherically symmetric matter}

\begin{figure}
    \centering
    \includegraphics[width=0.7\linewidth]{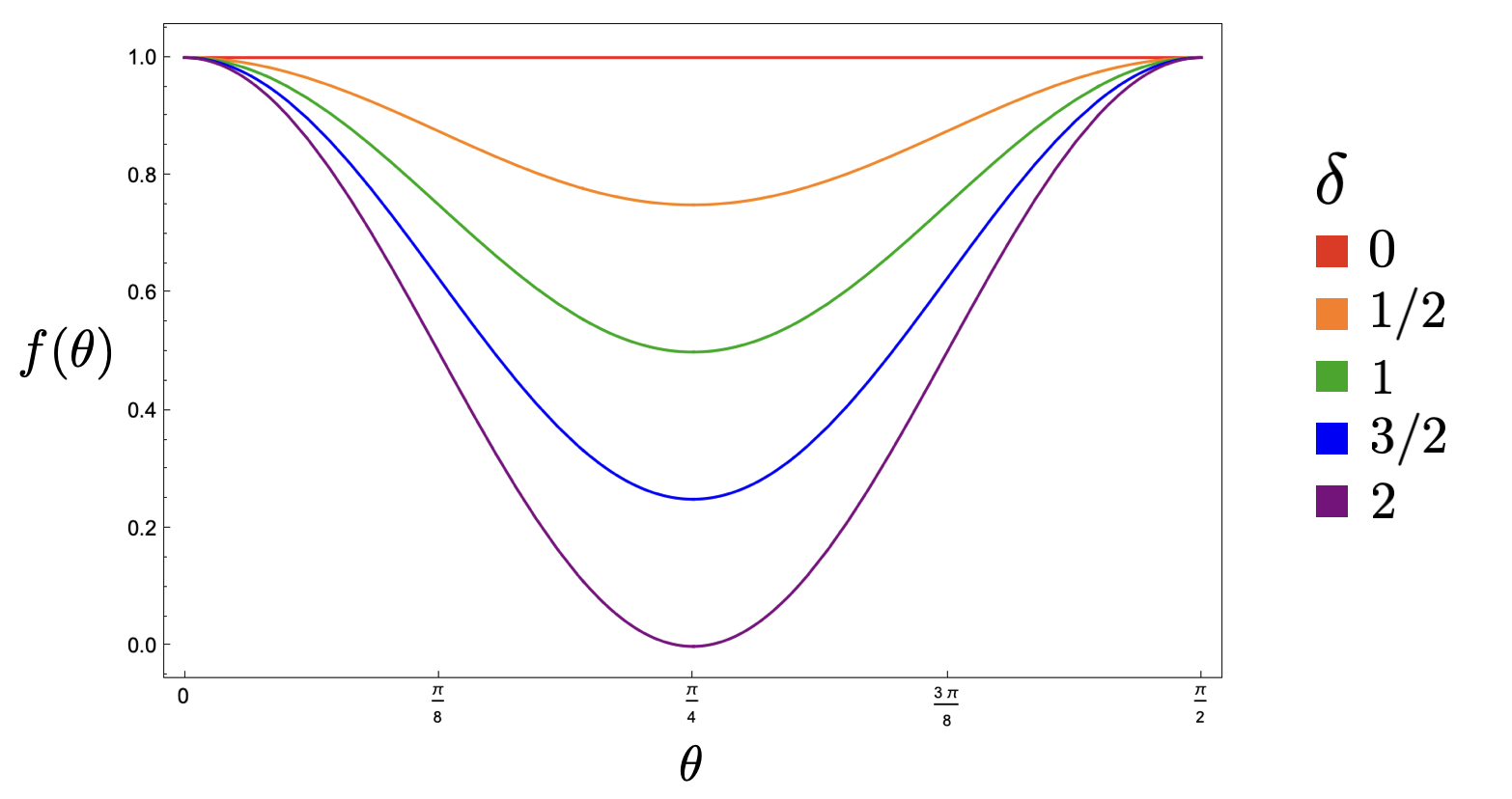}
    \caption{The ``blackening'' factor $f(\theta)$ defining the AAdS geometry \eqref{eq:metric-matter} for various values of the parameter $\delta$. For $\delta = 0$, one recovers pure global AdS$_3$. The matter warps the geometry most strongly at $\theta = \pi/4$, diminishing near the asymptotic boundary as required.}
    \label{fig:blackening}
\end{figure}

Having provided evidence to the conjecture in pure AdS$_3$, we now move away from the vacuum to test how the cooperating property holds up when matter is introduced in the spacetime. Consider the following bulk metric
\be\label{eq:metric-matter}
\dd s^2 = \frac{1}{\cos^2{\theta}}\left(-f(\theta)\, \dd t^2 + \frac{\dd \theta^2}{f(\theta)} + \sin^2{\theta}\, \dd \phi^2\right), \quad f(\theta) = 1 - \frac{\delta}{2}\sin^{2}(2\theta).
\ee
The matter supporting this geometry is described by a spherically symmetric and time independent energy-momentum tensor $T_{ab}$ obeying the null energy condition whenever $0\leq \delta < 2$ (in particular, for a null vector $k^a$, the above matter satisfies $T_{ab}k^a k^b = 0$). In figure \ref{fig:blackening} we show a plot of the factor $f(\theta)$ for different values of the parameter $\delta$. We will again consider boundary regions whose minimax surfaces generate a boosted H configuration in the spacetime. Due to the added complexity, the analysis will not be as exhaustive as that in pure AdS$_3$ and we will have to settle for a numerical exploration of a few configurations instead.

We want to find the spine, ribs, and seams in this geometry, just like we have done for the empty AdS$_3$ case. This is harder, but there are a few considerations that will simplify our problem. First, note that the spine gets unmodified by the presence of matter. This conveniently simplifies the computation for the tangent and orthogonal vectors.  
Second, note that we are doing a local analysis, so we only need to solve the geodesic equations for the timelike seams up to quadratic order in the affine parameter. Once again, we refer the reader to appendix \ref{app:geodesic_seams} for some of the details. We found it easier to work with the angular momentum $L$ and energy $E$ of the spacelike geodesics (the ribs), where $L$ controls the intersection point of the rib along the spine (at $\theta_0 = \arctan{L})$ and $E$ its boost ($E = 0$ being the static case). Thus, a configuration of two nearby ribs is defined by their angular momentum $L_1 = L$, $L_2 = L + \epsilon$ and by their equal and opposite boost $E_1 = -E_2 = E$.  In these variables, the parameterization for the timelike geodesic seam shot orthogonally from the intersection point $(0, \theta_0, 0)$ of the spine with a rib with angular momentum $L$ and energy $E$ is
\be
t(\lambda) = \frac{L \cos^2{\theta_0}}{f(\theta_0)}\lambda, \quad\theta(\lambda) = \theta_0 -\frac{L^3}{2(1+L^2)^2}\frac{g(\theta_0)}{f(\theta_0)}\lambda^2, \quad \phi(\lambda) = \frac{E\cos^2{\theta_0}}{f(\theta_0)}\lambda,
\ee
where $g(\theta) = 1 - \frac{\delta}{2}\cos^{4}{\theta}$.

\begin{figure}
    \centering
    \includegraphics[width=0.8\linewidth]{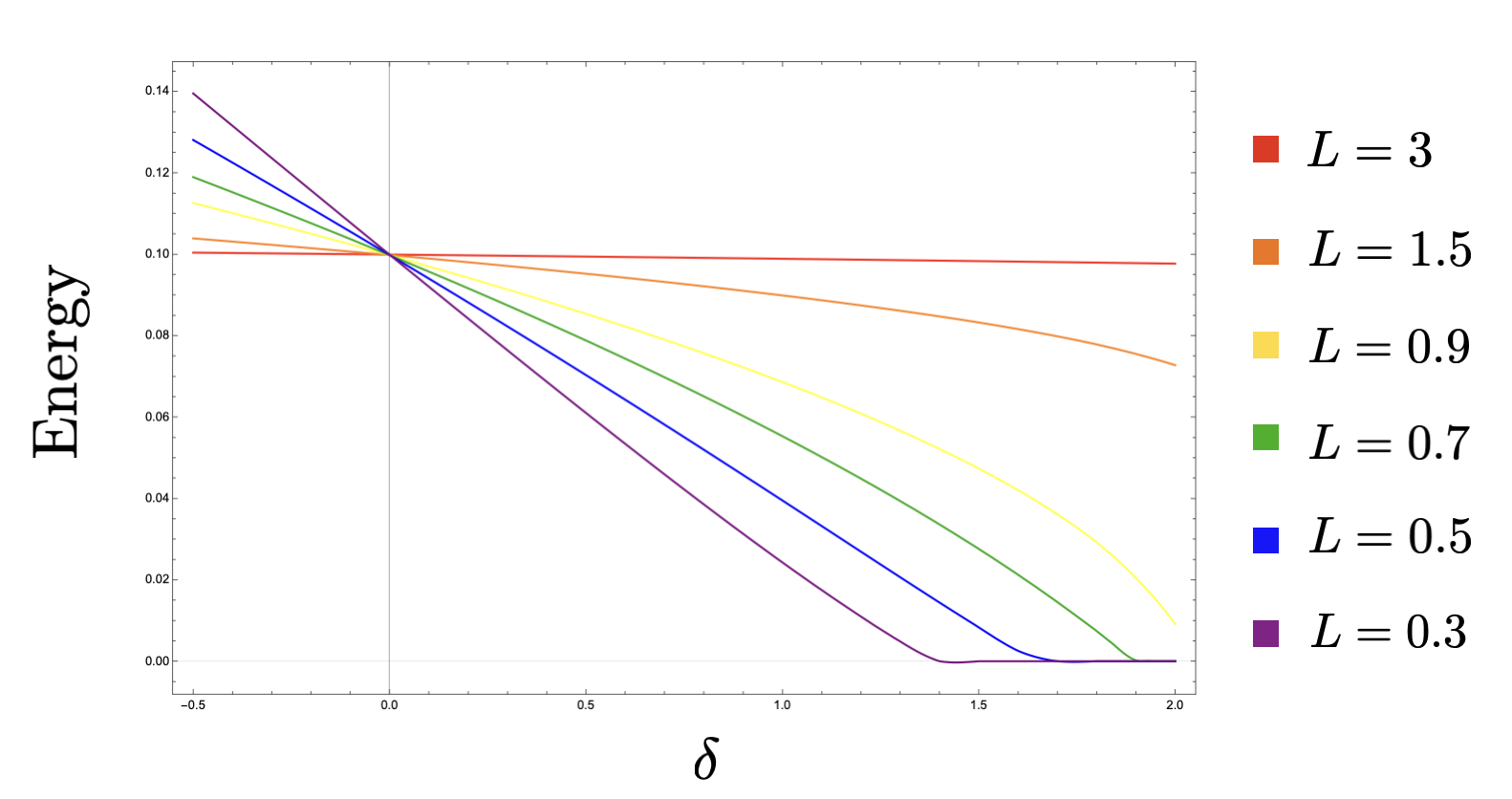}
    \caption{The maximum boost/energy $E_c$ allowed from boundary causality as a function of $\delta$ for various values of $L$ (the position of the first rib) and fixed $\epsilon = 0.2$. Larger values of $L$ correspond to the ribs being near the AdS boundary; this translates to $E_c$ not differing much from its vacuum value and having a flatter curve. As we decrease $L$ we bring the pair of ribs deeper into the bulk and the bound on causality gets stricter and stricter: matter helps the spacelike geodesics converge. For smaller values of $L$, the ribs can cross before reaching the boundary (shown by the purple, blue and green curves hitting the axis before $\delta = 2$). For negative values of $\delta$, where matter does not obey the NEC, we see that $E_c$ increases as we move away from the vacuum: this indicates that NEC may play a key role in cooperation.
    }
    \label{fig:causality-matter}
\end{figure}

\begin{figure}
    \centering
    \includegraphics[width=0.85\linewidth]{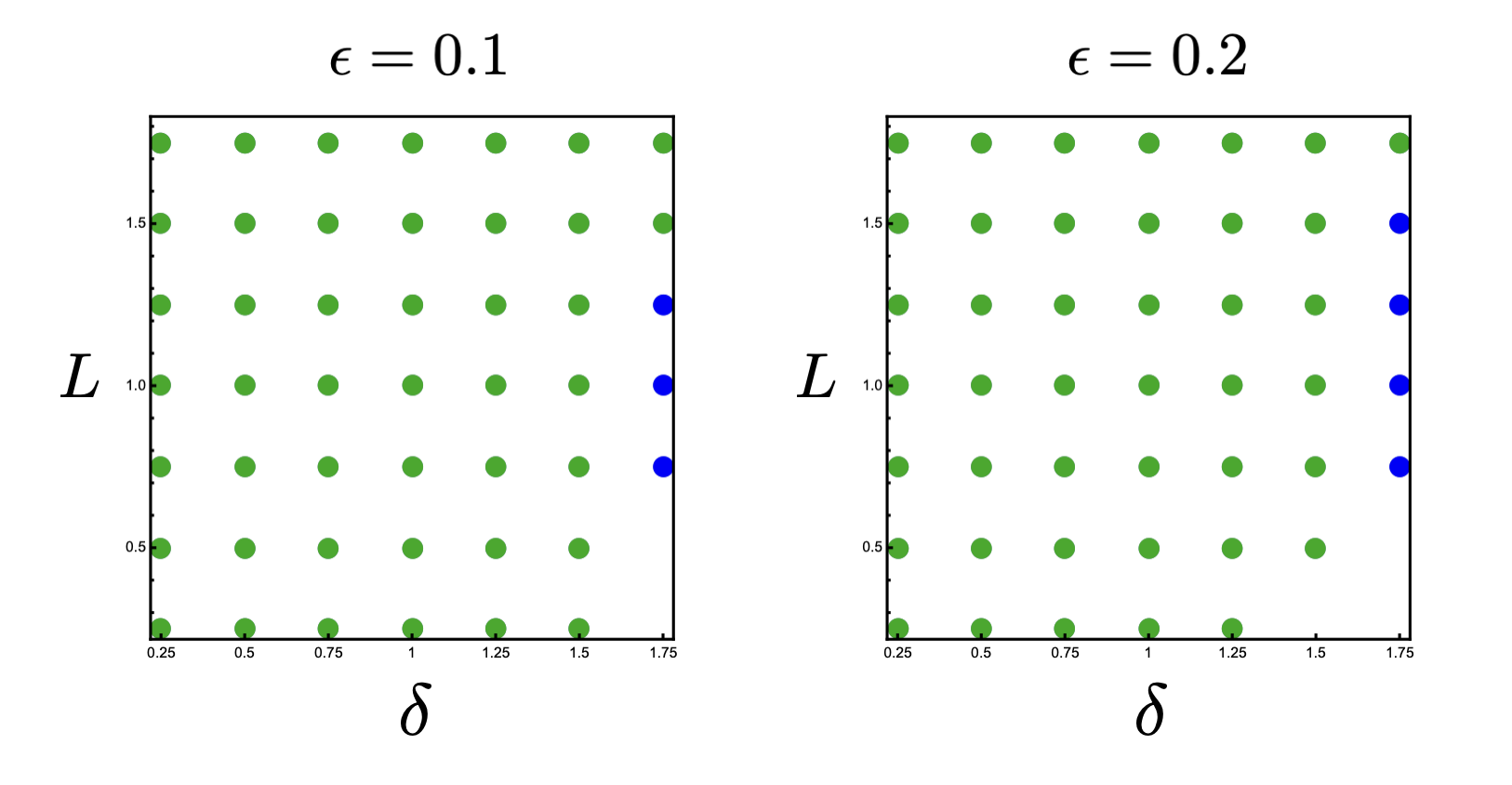}
    \caption{Test of local maximality in the $(L,\delta)$ parameter space for $\epsilon = 0.1$ {\bf (left)} and $\epsilon = 0.2$ {\bf (right)} in the regime where the NEC is obeyed. A green point indicates that the central segment is locally maximal for the worst possible boost dictated by boundary causality, and thus, $E_m$ for that configuration must be at least as large thereby confirming \eqref{eq:matter-bound}. A blue point means that the central segment is still locally maximal, but for a different choice of ansatz for the timelike seams. In particular, one seam was kept geodesic while the second one had a reduced acceleration in the $\theta$ direction (we reduced it by 1/4). A lack of a point means that the triplet $(L,\delta,\epsilon)$ does not produce a boosted H configuration in the bulk, and is therefore discared from the analysis. We chose to explore the range $0.25\leq L \leq 1.75$ because for $L = 1$ the rib intersects the spine at $\theta_0 = \pi/4$, which is the radius where the effects of matter are strongest. Since $f(\theta) \to 1$ both towards the boundary and towards the AdS center, a threat to the conjecture is more likely to be found near $\theta = \pi/4$.}
    \label{fig:coop-matter-1}
\end{figure}

As before, the energy/boost is constrained by boundary causality. We can find this bound numerically by integrating the $L$ and $E$ conservation equations all the way up to the boundary along the two ribs, which allows us to find the angle separation $\delta\phi$ and time separation $\delta t$ of the nearby endpoints of the two geodesics at the boundary, i.e.
\be
t_{\infty} = \int_0^{\infty} \dd\lambda \, \frac{E \cos^2{\theta}}{f(\theta)}, \quad \phi_{\infty} = \int_0^{\infty} \dd\lambda \, \frac{L}{\tan^2{\theta}},
\ee
then $\delta\phi = \phi_\infty(L,E) - \phi_\infty(L+\epsilon,-E)$ and $\delta t = t_\infty(L,E) - t_\infty(L+\epsilon,-E)$. Enforcing $\delta \phi \geq \delta t$ leads to a causality contraint $E \leq E_{c}$. For fixed $(L, \epsilon)$ not all values of $0<\delta<2$ are allowed to obtain a boosted H configuration; there exists a critical value of $\delta$ for which the two ribs meet at the same point on the boundary, and for higher values the two ribs cross before reaching the boundary thereby introducing new intersections between time-sheets, which is beyond the scope of this analysis. So when computing $E_c$, we filter out these cases. The overall trend we observe is for $E_c$ to be decreasing monotonically as a function of $\delta$, see figure \ref{fig:causality-matter}.

\begin{figure}
    \centering
    \includegraphics[width=0.85\linewidth]{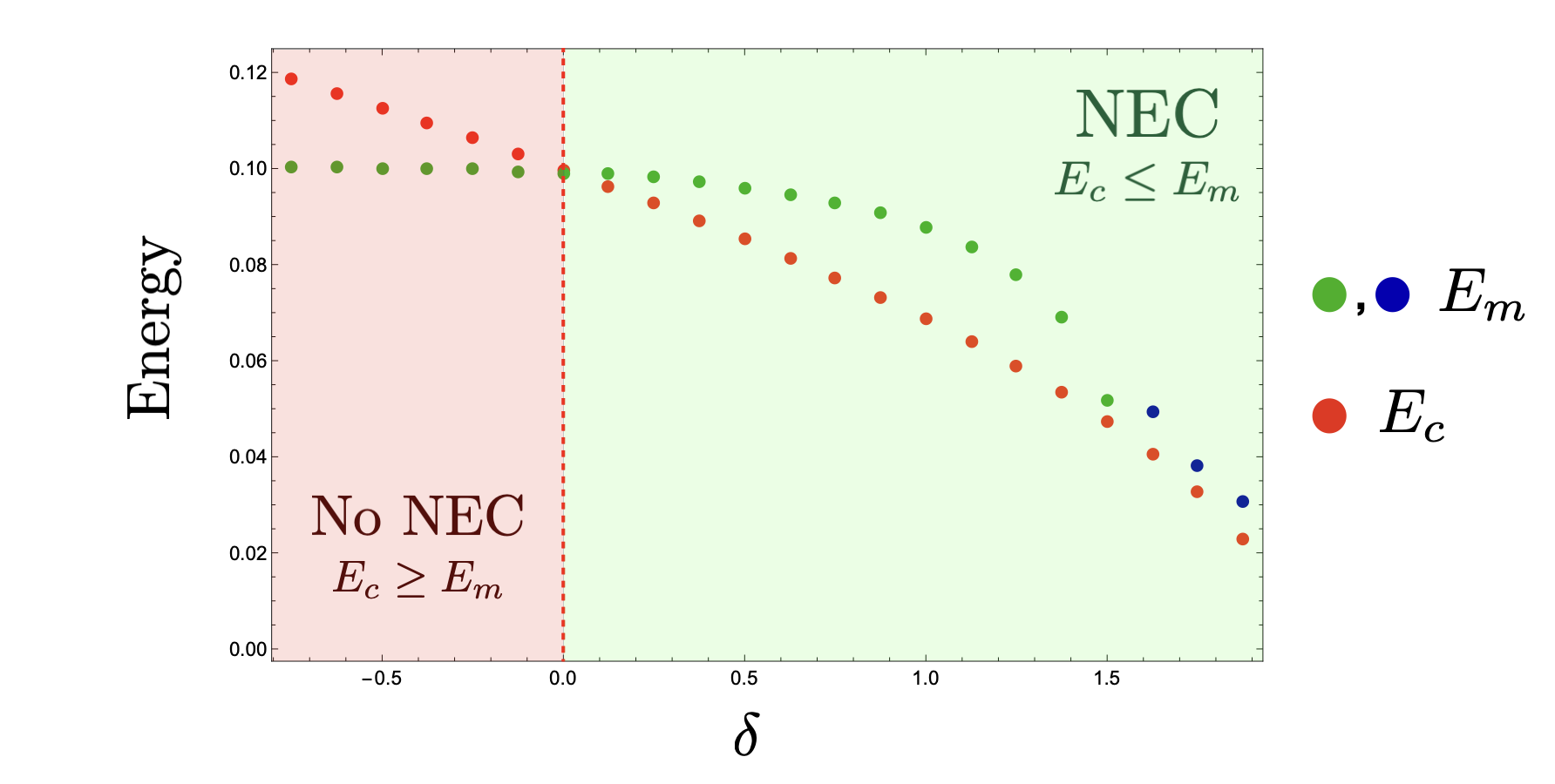}
    \caption{Plot of $E_m$ (green and blue dots) and $E_c$ (red dots) vs $\delta$ for a configuration with $L = 0.9$ and $\epsilon = 0.2$. We chose these values of $L$ and $\epsilon$ as they pose the most threat: the ribs are located near $\theta = \pi/4$ where matter acts the strongest, with one rib towards the outside of the shell and the other towards the inside. In the range $0\leq \delta<2$, NEC is obeyed and we find $E_c \leq E_m$: cooperation is guaranteed. For high values of $\delta \gtrsim 1.5$ we found our simple ansatsz of timelike geodesics not to be enough and we had to slightly deform one of the timelike intersection seams by reducing its accelleration in the $\theta$ direction by a factor of (i) 0.95 when $\delta = 1.625$, (ii) 0.83 when $\delta = 1.750$ and $0.59$ when $\delta = 1.875$ (we indicate these three cases with blue dots). For $\delta = 0$ we recover the pure AdS$_3$ result, $E_c = E_m$, and for negative values of $\delta$ the matter configuration does not obey the NEC and we observe a sudden exchange in dominance between $E_c$ and $E_m$. }
    \label{fig:coop-matter-2}
\end{figure}

We can then proceed with a numerical investigation of local maximality of the central partial segment\footnote{Like in the pure AdS$_3$ example, the other six segments that reach the asymptotic boundary have a much weaker constraint for maximality than the central segment and therefore do not pose a threat.}. In order to do that, we need to solve for the geodesic distance between two points along the two orthogonal timelike geodesics $\gamma_{\perp}(\lambda_1)$ and  $\gamma_{\perp}(\lambda_2)$; this can be done by first solving the geodesic equations with Dirichlet boundary conditions on $\gamma_{\perp}(\lambda_1)$ and  $\gamma_{\perp}(\lambda_2)$ and then numerically integrating the corresponding line element. Then, computing the eigenvalues of the Hessian (which in turn is found by numerical differentiation) at the point $(\lambda_1, \lambda_2) = (0,0)$ for a set of parameters $\{L, E, \epsilon, \delta\}$ will test for the local maximality of the partial surface. More precisely, we are interested in the \emph{critical} value $E_{m}$ past which the partial surface ceases to be maximal (i.e.\ the turning point of the definiteness of the Hessian). For the cooperating property to hold, it must be that
\be\label{eq:matter-bound}
E_c \leq E_m,
\ee
i.e.\ the bound dictated by causality must be stricter than the one for maximality. For the example in pure AdS$_3$, this bound was found to be tight. For the spacetime in question here, we find that whenever NEC is obeyed, \eqref{eq:matter-bound} holds. To test this, we perform two different numerical explorations:
\begin{enumerate}
    \item We sample configurations in the $(L, \epsilon, \delta)$ parameter space and check whether the Hessian is negative definite at $(0,0)$ for the \emph{worst possible boost}, i.e.\ for $E = E_c$. This test contains less information (it does not specify what $E_m$, only that it must be greater), but it has the advantage of being fast. 
    \item We perform an active search for the value $E_m$ for fixed $(L, \epsilon, \delta)$ with a Newton-like method: we start with a guess $(E_{min},\ E_{max})$ for the lower and upper bound on $E_m$ and iteratively either (i) reduce $E_{max}$ if the Hessian test fails or (ii) increase $E_{min}$ if the Hessian test passes. The test converges to the desired $E_m$. We then plot $E_m$ and $E_c$ as functions of $\delta$ and compare.
\end{enumerate}
We show our results for (1) and (2) in figures \ref{fig:coop-matter-1} and \ref{fig:coop-matter-2} respectively. 

\subsection{Challenges for cooperation}\label{sec:against}
In the previous section we considered examples  of configurations with more than two time-sheets. However, these configurations were rather simple in that the two boosted HRT surfaces intersected the static spine orthogonally, and further, they did not cross each other. A more  thorough test of the cooperating conjecture for three time-sheets would involve  pairwise crossing boundary regions with pairwise non-intersecting HRT surfaces. While such setup makes for a much harder configuration to perform concrete computations with, it serves to highlight potential issues with the conjecture. In this subsection we consider an extreme example of a configuration in this class, for which we are unable to find a cooperating set of time-sheets (and identify the challenges in constructing one).

This configuration consists of three boosted HRT surfaces that approach each other (but do not intersect) near the center of AdS.
Each of the three HRT surfaces, $\minimax_1, \minimax_2$ and $\minimax_3$, have endpoints separated by $\Delta \phi = \pi - \epsilon$ and $\Delta t = \pi /3 - 2 \epsilon$, and are taken to be centered on $\phi_0 = 0, 2\pi/3$ and $4\pi/3$, respectively. See figure \ref{fig:small-triangle} for a depiction of the configuration. As $\epsilon$ goes to zero, the endpoint of an HRT surface becomes null separated from the neighboring points. For $\epsilon$ nonzero, every other boundary interval between endpoints is null separated, while the others remain spacelike. Further, as $\epsilon$ approaches zero, the centers of the HRT surfaces approach at the origin (forming a ``triangle'' as viewed from above, with spatial size $O(\epsilon)$). 
\begin{figure}
    \centering
    \includegraphics[width = \textwidth]{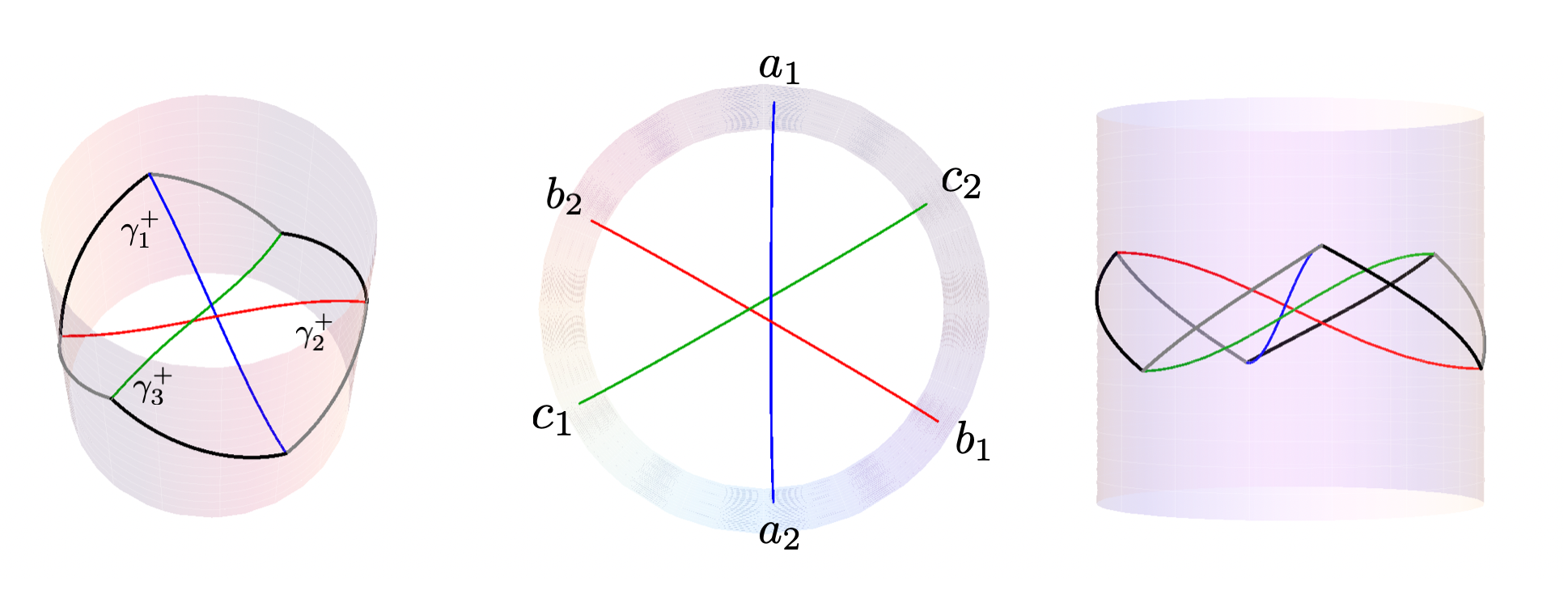}
    \caption{The small triangle configuration, from different points of view. {\bf Left:} Angled view of the configuration, with the three HRT surfaces (shown in blue, red, and green) that come close to intersecting near the center, forming a small triangle. We also show the boundary regions with black and gray lines on the boundary, the black ones are spacelike while the gray ones are null. {\bf Middle:} Top view, showing the endpoints of the regions. Specifically, $a_1$ and $a_2$ have coordinates $\phi_1 = -\pi/2 + \epsilon/2$, $t_1 = \pi/6-\epsilon$ and $\phi_2 =\pi/2 - \epsilon/2$, $t_2 = -\pi/6+\epsilon$ on the boundary, all the others can be obtained by $2\pi/3$ rotations.  {\bf Right:} A side view of the same configuration, emphasizing the alternating nature of the boundary regions and showing how the three HRT surfaces all miss each other near the center.}
    \label{fig:small-triangle}
\end{figure}

For simplicty, we will consider a collection of time-sheets that respects the symmetry of the configuration: rotation by $2\pi/3$ and $\phi\rightarrow-\phi$ plus $t\rightarrow -t$. Such time-sheets intersect each other with seams that pass through $\pm \lambda_0$ on each HRT surface for some chosen $\lambda_0$. This will partition each time-sheet into three segments -- two that are anchored to the boundary, and one central segment. 

To understand why cooperation is difficult to achieve, consider the line segment $\ell$ constructed as follows: on the time-sheet containing $\minimax_1$, connect the point $\minimax_1(\lambda_0)$ to the point $\minimax_2(-\lambda_0)$. If we take $\lambda_0\sim \epsilon$, so that the seams are near the ``triangle,'' this segment has length\footnote{The equality gives the minimal geodesic distance between $\minimax_1(+\lambda_0)$ and $\minimax_2(-\lambda_0)$. Varying the time-sheet will give a strictly larger distance.}
\begin{equation}
    |\ell| \geq \sqrt{\epsilon^2 + 4 \lambda_0^2}. 
\end{equation}
This, however, is larger than the length of the segment of the HRT surface on the central portion of the time-sheet, $2\lambda_0$. Thus, this configuration would not be cooperating. 

We can do slightly better by placing the seam farther out -- while, as mentioned above, the size of the triangle is of order $\epsilon$, any pair of the HRT surfaces will remain timelike separated over a longer distance.  Choosing the $\lambda_0$, which maximizes $2\lambda_0 - |\ell|$, while keeping the seam timelike yields $\lambda_0 \sim \frac{\log{3}}{2}$.  This gives  $|\ell| - 2\lambda_0 \sim \epsilon$, so that the configuration is not automatically non-cooperating. 

However, for the time-sheet configuration to be cooperating, we need to be able to evolve the seams to the future of $\lambda_0$, while both keeping the geodesic distance between equal time points on the seam less than $2\lambda$, and also keeping the HRT segments on the exterior portions of the time-sheet locally maximal. To ensure local maximality, we can form the seam to be orthogonal to the HRT surface at $\lambda_0$. We can follow along this seam and compare the HRT segment length on the middle time-sheet segment ($2\lambda_0$) with the geodesic distance between the points given by the intersection of constant time slices and the seams. However, after $k_{max} =  8/13 \epsilon$, (corresponding to a difference in time between the HRT surface at $k=0$ and $k_{max}$ of $\Delta t = \frac{4  \tanh (\lambda_0)}{\tanh ^2(\lambda_0)+3} \epsilon + O(\epsilon)$), this distance already shrinks to zero. 

Thus, after choosing the seams to evolve for a time of order $\epsilon$ along the normal to the HRT surfaces, keeping the middle time-sheet cooperating would require choosing the seams to go straight up (with constant $\theta,\phi$ after $k_{max}$). However, the difference in area between the HRT surface on the exterior partial time-sheets, and the geodesic distance between a boundary point, and the point at $k_{max}$ goes as $O(\epsilon^2)$. Consequently, the seam can only be evolved straight for a amount of time $\Delta t \sim \epsilon$. At that point, the seam cannot evolve straight or towards the origin (otherwise the outer partial time-sheets would fail to be cooperating), and the seam can not evolve towards the boundary (otherwise the middle partial time-sheet would fail to be cooperating). 

While it remains possible that a different choice of time-sheet configuration (for example, one that does not respect the symmetries of the HRT surfaces) leads to a cooperating configuration, we were unable to find one. 

\subsection{Geometric proof of the inequalities}\label{sec:geometric-proof}
As established in subsection \ref{sec:graph-model}, if the cooperating property holds, it implies the existence of a spacetime graph model which is graph-equivalent to the model for static states. It then follows directly that all holographic entropy inequalities valid for static states will remain valid in the time-dependent setting. However, we find instructive to spell out how these inequalities could be proven geometrically using minimax. In this section we assume the cooperating property and present a strategy which resembles the original exclusion/inclusion argument of \cite{Headrick:2007km, Headrick:2013zda}. 
We find this approach useful because it demonstrates how the minimax prescription utilizes the full spacetime (overcoming issues that maximin faced for higher inequalities) while also clarifying the role of the cooperating property. 

\subsubsection*{RT and maximin strategy}

We first begin by reviewing the static argument. To keep things simple, we will use SSA as a toy example to explain the proof strategy. Consider three disjoint boundary regions $A,B$ and $C$ (plus a purifier), lying on the $t = 0$ slice of the boundary. Recall from \eqref{eq:ssa}, SSA states that
\be
S(AB) + S(BC) \geq S(B) + S(ABC).
\ee
In the time-reflection invariant case, all of the RT surfaces lie on the same slice by construction (the $t=0$ slice) and are minimal on that slice. Let $m(AB)$ and $m(BC)$ be the RT surfaces for the $AB$ and $BC$ regions, and let $r(AB)$ and $r(BC)$ be the two (spatial) homology regions. By taking the intersection and union of these, we define two new regions
\begin{equation}
    r'(B) = r(AB) \cap r(BC), \quad  r'(ABC) = r(AB) \cup r(BC),
\end{equation}
which by construction are homology regions for the boundary regions $B$ and $ABC$ respectively. The bulk part of their boundary, call it $m'(B)$ and $m'(ABC)$ respectively, can be constructed by cutting and gluing portions of the RT surfaces for $AB$ and $BC$. Since these surfaces are homologous to $B$ and $ABC$, their area must be greater than the true entropies $S(B)$ and $S(ABC)$ due to the global minimization step of the RT formula. So we have,
\begin{align}
S(AB) + S(BC) &= |m(AB)| + |m(BC)|\\
&\geq |m'(B)| + |m'(ABC)|\\
&\geq S(B) + S(ABC).
\end{align}

In the dynamical case, this proof strategy seemingly fails, as the relevant HRT surfaces do not generally all lie on the same slice. However, the maximin construction overcomes this with the help of one additional property: there exists a slice on which maximin surfaces for nested regions are all minimal. This allows us to use a nested set of inequalities as indicated in the Introduction. Both SSA and MMI can therefore be proven using this strategy \cite{Wall_2014}. Higher inequalities, however, generally do not possess this critical property: regions on either side of the inequality are generally not nested. Thus, the maximin proof strategy fails for higher party inequalities \cite{Rota:2017ubr}.

\begin{figure}
    \centering
    \includegraphics[width=0.8\textwidth]{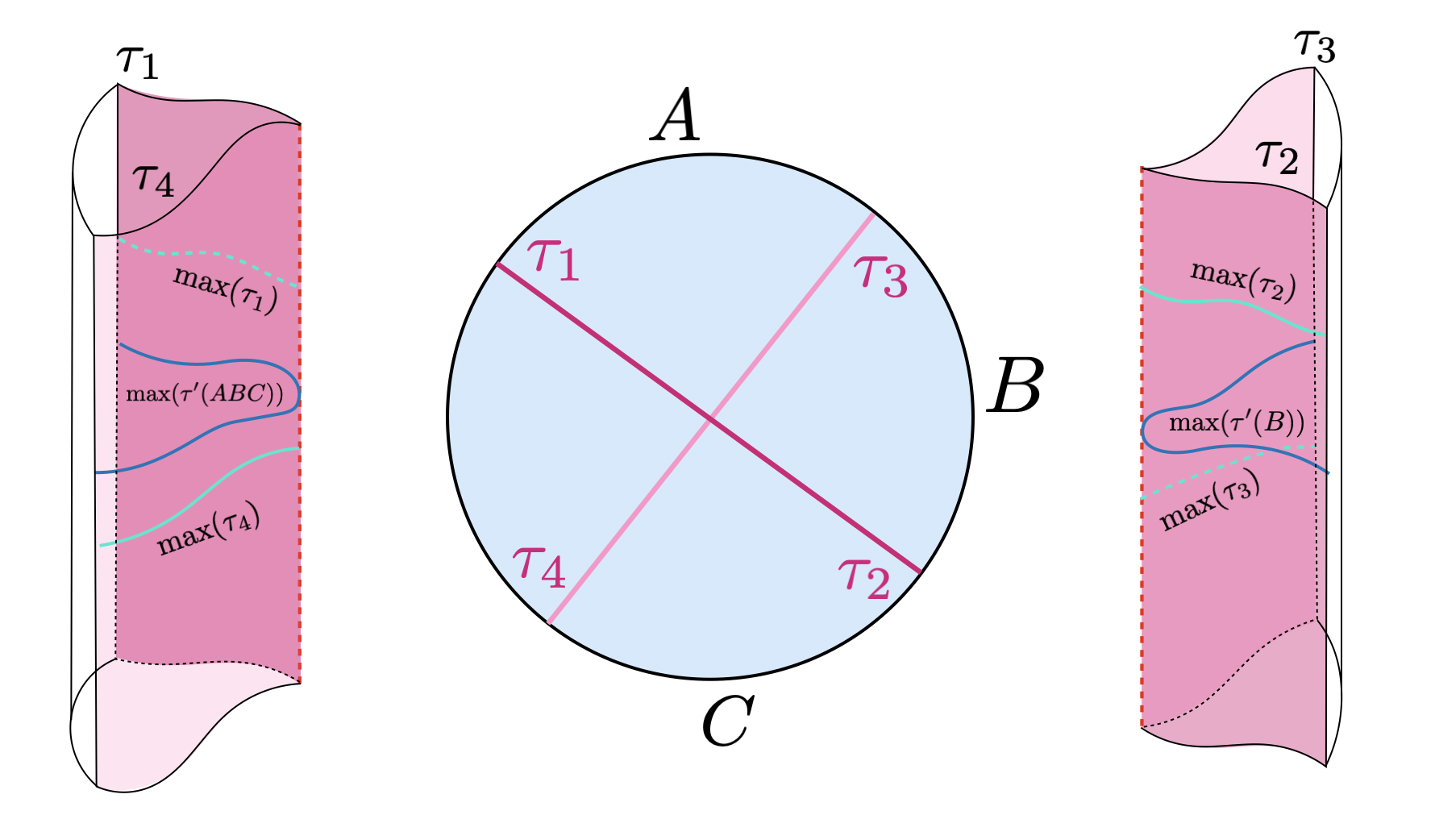}
    \caption{Proof sketch of SSA using minimax. In the {\bf middle}, we show a time slice of the bulk, with the cross sections of time-sheets shown in pink. On the {\bf left} and on the {\bf right} we show the patched partial time-sheets in the full spacetime to form $\ts'(ABC)=\tau_1 \cup \tau_4$
    and $\ts'(B)=\tau_2 \cup \tau_3$ respectively. The maximal achronal surfaces on $\ts'(B)$ and $\ts'(ABC)$ are denoted with $\max(\ts'(B))$ and $\max(\ts'(ABC))$ respectively and are shown in blue. The maximal achronal surfaces on the various $\ts_i$ are denoted by $\max(\ts_i)$ and are shown in cyan. Clearly, one always has 
    $\max(\tilde{\ts}(B)) \leq \max(\ts_2) + \max(\ts_3)$ 
    (and a similar statement on the other time-sheet) since the requirement of being connected poses a stronger constraint.
    The main cooperating conjecture is 
    that $\max(\ts_i) = \minimax_i$, which leads to $\max(\ts'(B)) + \max(\ts'(ABC)) \leq \minimax(AB) + \minimax(BC)$ -- hence proving SSA.}
    \label{fig:ssa-friendly}
\end{figure}

\subsubsection*{Minimax}

Now we can present how minimax overcomes this issue by elevating the static proof to the entire spacetime. In particular, we can use the same exclusion/inclusion strategy as in the RT proof of SSA, but this time the regions which we take intersections and unions of will be codimension-0 homology volumes instead of codimension-1 spacelike regions. Because of this, we will not need to project all of the relevant HRT surfaces onto the same slice, and so the nesting properties will become unimportant. 

Let $\ts(AB)$ and $\ts(BC)$ be 
any (cooperating, minimally intersecting)
minimax time-sheets containing the minimax surfaces $\minimax(AB)$ and $\minimax(BC)$. Define the following spacetime volumes
\be
R'(B) = R(AB) \cap R(BC), \quad R'(ABC) = R(AB) \cup R(BC)
\ee
where $R(AB)$ and $R(BC)$ are the two homology regions bounded by the two minimax time-sheets $\ts(AB)$ and $\ts(BC)$. By construction, $R'(B)$ and $R'(ABC)$ will be homologous to $D(B)$ and $D(ABC)$ respectively and they will be bounded by portions of $\ts(AB)$ and $\ts(BC)$ that have been cut and glued -- call them $\ts'(B)$ and $\ts'(ABC)$. On $\ts'(B)$ there are two relevant sets of surfaces:
\begin{enumerate}
    \item The maximal surface on $\ts'(B)$, which we denote $\max(\ts'(B))$. We know by minimality over time-sheets that this has greater area than the minimax surface, $|\max(\ts'(B))|\geq |\minimax(B)|$. 
    
    \item The second relevant surface will be $\gamma'(B)$ obtained as the union of the portions of $\minimax(BC)$ and $\minimax(AB)$ on $\ts'(B)$ (i.e.\ the partial minimax surfaces which we refered to before in subsections \ref{sec:coop} and \ref{sec:graph-model}). Note that $\gamma'(B)$ will generally be disconnected, as $\minimax(AB)$ and $\minimax(BC)$ need not intersect.
\end{enumerate}

A priori, there is no obvious way to compare the areas of $\max(\ts'(B))$ and $\gamma'(B)$. This is the key step where the cooperating property appears: by showing that each partial minimax surface is maximal on the respective partial time-sheet, the inequality 
\be
|\max(\ts'(B))| \leq |\gamma'(B)| 
\ee
follows since is $|\max(\ts'(B))|$ maximized under the stronger constraint that it must be a connected surface. A similar argument also follows for $\gamma'(ABC)$ and $\max(\ts'(ABC))$. The proof can then be completed in a similar fashion as for the RT case:
\begin{align}
    S(B) + S(ABC) &= |\minimax(B)| + |\minimax(ABC)|\\
    &\leq |\max(\ts'(B))| + |\max(\ts'(ABC))|\\
    &\leq |\gamma'(B))| + |\gamma'(ABC)|\\
    &= |\minimax(AB)| + |\minimax(BC)|\\
    &= S(AB) + S(BC),
\end{align}
where the second line follows by minimality over time-sheets, and the third follows from the cooperating property. If the cooperating conjecture holds, this stretegy is applicable to all higher entropy inequalities provable by the exclusion/inclusion argument (i.e.\ all inequalities provable with a contraction map).

\section{Discussion}\label{sec:discussion}

Having multiple reformulations of the same concept is invaluable in physics, as it deepens our understanding of the concept and introduces new toolkits to address various problems. In the context of holographic entanglement, the maximin formulation is a prime example of this. On one hand, it sacrifices the clarity of the original HRT formulation by rewriting it in a seemingly less covariant way; on the other hand, it provides us with new methods for proving theorems. By introducing a minimization of areas over Cauchy slices, maximin reformulates the extremization problem in a way that brings it into closer contact with the static scenario. 

One of the most useful tools that emerged from maximin is the “trick of representatives,” which has been used repeatedly in many proofs. This trick involves null-projecting an extremal surface onto a Cauchy slice --- often one containing the maximin surface. By focusing, the projected surface will have a smaller area than the original extremal surface while also lying on the same slice as the maximin surface. This is often an advantageous configuration to utilize in proofs.

As is often the case, the greatest strength can secretly be the greatest weakness and for maximin it is no exception. The comfort of projecting surfaces onto the same Cauchy slice is precisely the limiting factor when confronted with the higher entropy inequalities. In this paper, we saw how the minimax approach offers the advantage of moving beyond the familiar setting of slices, lifting the entire problem to the level of spacetime. In this way, the usual spacelike homology condition is replaced by a spacetime homology condition. For the proof of the higher entropy inequalities, it becomes irrelevant that surfaces may not lie on the same Cauchy slice, since the static proof of union and intersection of spacelike regions is turned into a covariant proof using unions and intersections of spacetime volumes bounded by the time-sheets (instead of the RT surfaces). From the minimax perspective, the trick of representatives of maximin is simply a special case of the property that any candidate minimax surface is maximal on its time-sheet, since the entanglement horizon used in projecting the surface is a minimax time-sheet. Hence, minimax often allows for more direct proofs without the need of referencing to any projection, whether null or timelike.

While of course some of the claims presented in this paper are contingent on the validity of the cooperating conjecture, we hope to have shown the potential of the minimax prescription as a powerful new tool for holography. Because of this, there are many interesting future directions and open questions remaining that we believe the minimax prescription may play a central role in solving. We briefly discuss them in the remainder of this section.

\paragraph{What if time-sheets do not cooperate?}

While we hope to have given good evidence in favor of the cooperating property, it still remains a conjecture that needs further investigation. In particular, the comments in subsection \ref{sec:against} suggest that the conjecture may be false in its current form. However, we are hopeful that under some mild modifications a weaker statement may be proven that still allows for the construction of the graph model. If even that turns out to be false, we note that the cooperating property is only a sufficient condition for our results, and there is still overwhelming evidence in favor of the equivalence between the RT and HRT cones, as we argued in the introduction. This would make the search for such an elusive proof even more exciting.

\paragraph{The homology condition and a ``minimin'' formula}

Theorem \ref{thm:relaxed-achronal} established the equivalence between relaxed minimax \eqref{eq:relaxed-minimax} and original minimax \eqref{Splus} (and therefore with HRT as well) by showing that relaxed minimax surfaces must be achronal in the full spacetime. We mentioned how this is an interesting result regarding the nature of the homology condition for holographic entanglement, since it implies that spacelike homology is equivalent to spacetime homology. There is also another, more practical, consequence of this equivalence: it establishes a new ``minimin'' formula originally proposed in footnote 34 of \cite{Headrick:2022nbe}. This approach involves replacing the maximization of the area of an achronal surface with the minimization of the flux of a timelike flow through the Lorentzian max cut-min flow theorem \cite{Headrick:2017ucz}. We encourage interested readers to explore this formulation further. It would be particularly interesting to see if this offers a useful new perspective on the cooperating conjecture.

\paragraph{Tensor networks and time evolution}

Tensor networks can be used to model static holographic states, with entanglement playing a crucial role in their construction \cite{Swingle:2009bg}. One class of such models utilizes tree graphs constructed using non-intersecting RT surfaces, to define a tensor network that approximates a given static holographic state \cite{Bao:2018pvs}. The bond dimension of each edge corresponds to the area of the RT surface it passes through. In this sense, tree graph models of holographic entanglement can be seen as a primitive starting point for building tensor networks. The existence of our spacetime graph model from subsection \ref{sec:graph-model} would allow the construction of a tensor network for any time-dependent classical holographic state by applying the same recipe of \cite{Bao:2018pvs}. In fact, for tree tensor networks formed by non-intersecting time-sheets as in figure \ref{fig:spacetime-graph}, the cooperating conjecture is not needed for the graph construction to work, but it is needed for more general bulk discretizations (though the construction of tensor networks for these are less understood even in the static case).

It is important to stress, however, that while we are calling our graph a ``spacetime'' graph, it is still associated to the entanglement structure of a CFT state on a single (boundary) Cauchy slice. Time evolution of this slice through the boundary Hamiltonian will cause the graph to change both in its edge weights and in its topology. Therefore, while we argue that our minimax construction allows for the definition of tensor networks for time-dependent states, it should not be confused with the problem of time evolving the tensor network itself, which our construction does not address.

\paragraph{Quantum minimax \& bulk inequalities}
Throughout this paper, we have worked in the fully classical regime of $N \to \infty$, $\lambda \to \infty$. Quantum corrections are captured by the quantum extremal surface (QES) formula \cite{Engelhardt:2014gca}
\be\label{eq:QES}
S(A) = \min \underset{\gamma}{\text{ext}}
\left( |\surf| + S(\rho_r)\right),
\ee
where the extremization is over achronal surfaces $\gamma$ spacelike homologous to $A$ and the quantity in brackets being extremized is the generalized entropy, the sum of the area of the surface $\gamma$ and the entropy $S(\rho_r)$ of bulk quantum fields living on the spacelike homology region $r(A)$.

The QES has since been reformulated in its quantum maximin form \cite{Akers:2019lzs}, so it is only natural to propose a quantum minimax formula
\begin{equation}
    S(A) = \min_{\ts} \max_{\gamma} \left( |\surf| + S(\rho_{\text{bulk}})\right).
\end{equation}
One first puzzle is the interpretation of $S(\rho_{\text{bulk}})$, as there appears to be no prefered slice on which to evaluate the bulk quantum fields entropy. This can be fixed by reintroducing Cauchy slices, i.e.
\begin{equation}\label{q-minimax}
    S(A) = \min_{\ts} \max_{\sigma} \left( |\surf| + S(\rho_{\ts \cap \sigma})\right),
\end{equation}
which makes the prescription more explicit: one first maximizes the generalized entropy of spacelike homology regions bounded by $\ts \cap \sigma$ over all Cauchy slices, and then minimizes over the choice of time-sheets. 

As is well known, the holographic entropy inequalities beyond SSA do not hold in general when bulk quantum corrections are turned on. Nevertheless, \cite{Akers:2021lms} studied the interplay between bulk matter entropies and boundary entropies, showing that in the static regime constraining the bulk matter to obey all holographic entropy inequalities implies that the boundary state obeys them as well. It is natural to explore similar questions in the time-dependent setting.

\paragraph{Pythons and bulges}

The python's lunch conjecture (PLC) \cite{Brown:2019rox} gives a geometrical formula for the complexity of reconstructing operators lying in the bulk region between competing extremal surfaces homologous to a given boundary region. Suppose the extremal surface closest to the region $A$, which we will call the constriction $\surf_c$, is \emph{not} the HRT surface. According to the PLC, reconstructing operators between the two surfaces (more precisely, within the entanglement wedge but outside the wedge of $\surf_c$) is exponentially complex, with the log-complexity given by
\be
\ln\mathcal{C}=\frac12\left(|\surf_b|-|\surf_c|\right),
\ee
where $\surf_b$ is the ``bulge'' surface, an extremal surface lying between $\HRT$ and $\surf_b$. The bulge surface can be found by a ``maximinimax'' procedure, which involves the following steps: fix a Cauchy slice containing both $\HRT$ and $\surf_c$; choose a 1-parameter family of surfaces interpolating between the two surfaces (called a ``sweep-out''); find the maximal-area surface within the sweep-out; minimize over the sweep-out; and finally, maximize over the Cauchy slice. In principle, this gives a surface with one spatial negative mode and no temporal positive modes. (The maximization and minimization steps within a fixed Cauchy slice are standard maneuvers in geometric measure theory for proving the existence of an extremal surface with exactly one negative mode on an arbitrary Riemannian manifold. See \cite{Brown:2019rox,Arora:2024edk,Engelhardt:2023bpv} for further details.)

While we will not attempt a careful proof, we claim that, like the HRT surface, the constriction and bulge surfaces can be found using time-sheets rather than Cauchy slices. First, as always, on each time-sheet we find the maximal-area achronal surface; this gives us a function $a_{\rm max}$ on the space $\tsset_A$ of time-sheets. (For the maximal surface $\surf$ on a given time-sheet, there is a large continuous family of other time-sheets on which $\surf$ is also maximal. Hence the function $a_{\rm max}$ has many flat directions.) The HRT surface of course is (the maximal surface on) the global minimum of $a_{\rm max}$. The function $a_{\rm max}$ goes to infinity for time-sheets that approach the AdS boundary, and the constriction is (the maximal surface on) the local minimum closest to the boundary. Between these two local minima, we expect there to exist a saddle point with Morse index 1 (i.e.\ one negative mode), which is the bulge surface. This can be found using sweep-outs in $\tsset_A$ that interpolate between the HRT and constriction time-sheets; we maximize $a_{\rm max}$ on each sweep-out, then minimize over sweep-outs. A possible advantage of this ``minimaximax'' formulation over the original maximinimax one is that it dispenses with the time direction first, leaving a standard exercise in Morse theory.

\paragraph{Entangled universes}

In forthcoming work, Gupta-Headrick-Sasieta \cite{entangleduniverses} will propose a generalization of the HRT formula to spacetimes consisting of asymptotically AdS, Minkowski, and/or de Sitter regions (or universes) connected by wormholes, with the area of an extremal surface lying in the causal shadow inside each wormhole giving the entanglement between the neighboring universes. These (generalized) HRT surfaces can be found using either maximin or minimax (or U- or V-flows \cite{Headrick:2022nbe}). Minimax is perhaps more elegant than maximin because the rules (and especially the spacetime homology condition) are the same for all of these types of boundaries --- AdS, Minkowski, and dS. On the other hand, for maximin, a special rule must be made for dS boundaries, in which Cauchy slices are replaced by timelike hypersurfaces that end on the dS boundary.

\acknowledgments

We would like to thank Raphael Bousso, Sergio Hern\'andez-Cuenca, Mukund Rangamani, and Nico Valdes-Meller for useful conversations. B.G.W., G.G., and M.H. were supported in part by the Department of Energy through awards DE-SC0009986 and QuantISED DE-SC0020360, and in part by the Simons Foundation through the \emph{It from Qubit} Simons Collaboration. B.G.W. was also supported by the AFOSR under FA9550-19-1-0360. V.H. was supported in part by the Department of Energy through awards DE-SC0009999 and QuantISED DE-SC0020360, and by funds from the University of California.
This research was supported in part by grant NSF PHY-2309135 to the Kavli Institute for Theoretical Physics (KITP), where part of this work was completed. This work was also performed in part at the Aspen Center for Physics, which is supported by National Science Foundation grant PHY-2210452. We are also grateful to the Perimeter Institute, the Centro de Ciencias de Benasque Pedro Pascual, and the Yukawa Institute for Theoretical Physics, where part of this work was completed.

\appendix

\section{Spacelike and timelike geodesics}\label{app:geodesic_seams}
\subsection{AdS$_3$}
We work in AdS$_3$ in global coordinates $(t,r,\phi)$ and $(t,\theta, \phi)$, with $\tan\theta = r$ with the following line elements
\be
\dd s^2 = -(1+r^2)\dd t^2 + \frac{\dd r^2}{1+r^2} + r^2 \dd\phi^2
\ee
and
\be
\dd s^2 = \frac{1}{\cos^2{\theta}}\left(-\dd t^2 + \dd \theta^2 + \sin^2{\theta}\,\dd\phi^2\right)
\ee
\subsection*{Spacelike geodesics}
We begin by studying HRT surfaces in AdS$_3$, which are spacelike geodesics. Time-translational and rotational symmetry lead to conservation of energy and angular momentum,
\begin{equation}\label{eq:cons}
    \dv{t}{\lambda} = \frac{E}{r^2 +1}, \quad \dv{\phi}{\lambda} = \frac{L}{r^2}.
\end{equation}
Further, we have the spacelike constraint $g_{ab}\dot{x}^a \dot{x}^b = 1$. Combining the three equations we get
\begin{equation}
    1 = -\frac{E^2}{r^2+1} +\frac{1}{r^2+1}\left(\dv{r}{\lambda}\right)^2 + \frac{L^2}{r^2}.
\end{equation}
The above can be written as
\begin{equation}
    \dot{r}^2 + V_{\text{eff}} = 0,
\end{equation}
where $V_{\text{eff}} = -r^2 - 1 - E^2 + L^2 + \frac{L^2}{r^2}$ is an effective potential. Substituting $u = r^2$ we have
\be
\frac{\dot{u}^2}{4} =  u^2 + a u - L^2 
\ee
with $a = E^2 - L^2 + 1$. The turning point $r_{\star}$ of the spacelike geodesic is the largest positive zero of $V_{\text{eff}}$ which is
\begin{equation}
    r_{\star}^2 = \frac{L^2-E^2-1+\sqrt{4L^2+(E^2-L^2+1)^2}}{2} = \frac{- a + \sqrt{\Delta}}{2}.
\end{equation}
with $\Delta = a^2 + 4L^2$. When $E = 0$, we are in the static/RT case and the turning point reduces to $r_{\star} = L$. The solution of the radial equation under the initial condition $r(0) = r_{\star}$ (i.e.\ we choose it so that intersection point with the spine is at $\lambda = 0$) is
\begin{equation}
    r(\lambda)^2  = \frac{- a  + \sqrt{\Delta}\cosh{2\lambda}}{2}.
\end{equation}
Transforming it back to the $\theta$ global coordinate we have
\begin{equation}
    \theta(\lambda) = \tan^{-1} \sqrt{ \frac{-a + \sqrt{\Delta}\cosh{2\lambda}}{2}  } = \cos^{-1}\sqrt{\frac{2}{2 + a + \sqrt{\Delta}\cosh 2\lambda}},
\end{equation}
for $E = 0$ this reduces to the static case $\theta(\lambda) = \cos^{-1}\left(\frac{\sech{\lambda}}{1+L^2}\right)$. Now that we have $\theta(\lambda)$ we solve for $t(\lambda)$ and $\phi(\lambda)$ (with initial conditions $t(0) = \phi(0) = 0$). We find
\be
t(\lambda) = \tan^{-1}\left[\left(\frac{-2+a +\sqrt{\Delta}}{2E}\right)\tanh\lambda\right],
\ee
and
\be
\phi(\lambda) = \tan^{-1}\left[\left(\frac{a+\sqrt{\Delta}}{2L}\right)\tanh\lambda\right]
\ee
It's useful to rewrite the above formulas not in terms of $(E,L)$ but in terms of the opening angle $\Delta\phi$ and the time separation $\Delta t$ between a region's endpoints. These are easily found by integrating equations \eqref{eq:cons}. One finds
\be
\Delta t = 2\tan^{-1}\left(\frac{-2+a +\sqrt{\Delta}}{2E}\right),\quad \Delta \phi = 2\tan^{-1}\left(\frac{a +\sqrt{\Delta}}{2L}\right).
\ee
Solving for $\Delta\phi$ and $\Delta t$ we have
\be
E = \frac{\sin{\Delta t}}{\cos{\Delta t}-\cos{\Delta \phi}}, \quad L = \frac{\sin{\Delta \phi}}{\cos{\Delta t}-\cos{\Delta \phi}}.
\ee
Obtaining the new parameterization $(\theta(\lambda), t(\lambda), \phi(\lambda))$ 
\be
\theta(\lambda) = \tan^{-1}\left(\sqrt{\frac{\cos{\Delta\phi}+\cosh{2\lambda}}{\cos{\Delta t} - \cos{\Delta\phi}}}\right),
\ee
and
\be
t(\lambda) = \tan^{-1}\left(\tan{\frac{\Delta t}{2}} \tanh{\lambda}\right),
\ee
and
\be
\phi(\lambda) = \tan^{-1}\left(\tan{\frac{\Delta \phi}{2}}\tanh{\lambda}\right).
\ee
We now want to find the perpendicular vector to both the spine and the rib at the point of intersection. The two tangent vectors to the spine and the rib are
\be
s = (0,1,0),\quad r=\left(\frac{-2+a + \sqrt{\Delta}}{2E},0,\frac{a + \sqrt{\Delta}}{2L}\right) = \left(\tan{\frac{\Delta t}{2}},0, \tan{\frac{\Delta\phi}{2}}\right) 
\ee
respectively. The orthogonal vector $\mathfrak{t}$ to both $s$ and $r$ is
\be\label{eq:orthogonality}
    \mathfrak{t} =\left( \frac{1+\cos{\Delta\phi}}{1+\cos{\Delta t}}\tan\frac{\Delta \phi}{2},0, \tan\frac{\Delta t}{2}\right).
\ee
For completeness, the intersection point $\theta_0$ of a rib along the spine in terms of $\Delta\phi$ and $\Delta t$ is given by
\be\label{eq:intersection-rib}
\theta_0 = \tan^{-1}\sqrt{\frac{\cos{\Delta \phi} + 1}{\cos{\Delta t}-\cos{\Delta\phi}}}
\ee
We now have all the ingredients to compute the orthogonal timelike geodesic seams.
\subsection*{Timelike geodesics}
For the timelike geodesics we only need a local analysis. Therefore, we are going to solve the geodesic equations up to quadratic order in the affine parameter $\lambda$, with initial conditions dictated by the orthogonality constraint \eqref{eq:orthogonality}. The parameterization for $t(\lambda)$ and $\phi(\lambda)$ is simple,
\be\label{eq:t-and-phi-seams}
t(\lambda) = \left(\frac{1+\cos{\Delta\phi}}{1+\cos{\Delta t}}\tan\frac{\Delta \phi}{2}\right)\,\lambda, \quad \phi(\lambda) = \tan\frac{\Delta t}{2}\lambda.
\ee
For the $\theta$ coordinate, we solve the geodesic equation
\be
\ddot{\theta} + \Gamma^{\theta}_{\,\,\,\,tt} \dot{t}\dot{t}+ \Gamma^{\theta}_{\,\,\,\,\theta \theta} \dot{\theta}\dot{\theta} + \Gamma^{\theta}_{\,\,\,\,\phi\phi} \dot{\phi} \dot{\phi}  = 0
\ee
by plugging the ansatsz $\theta(\lambda) = \theta_0 + \theta_1 \lambda + \theta_2 \lambda^2$. We find
\be
\theta(\lambda) = \theta_0 - \left(\frac{(t_1^2 - \phi_1^2) \tan{\theta_0}}{2}\right)\,\lambda^2
\ee
where $t_1$ and $\phi_1$ are the slopes in \eqref{eq:t-and-phi-seams}, and $\theta_0$ is given in \eqref{eq:intersection-rib}.

\subsection{Spherically symmetric 
matter}
The bulk metric reads
\be
\dd s^2 = \frac{1}{\cos^2{\theta}}\left(-f(\theta)\dd t^2 + \frac{\dd \theta^2}{f(\theta)} + \sin^2{\theta}\dd \phi^2\right), \quad f(\theta) = 1 - \frac{\delta}{2}\sin^{2}(2\theta).
\ee
Now, we want to find the spine, ribs and seams in this geometry. Note that since we are doing a local analysis, we really only need to solve the geodesic equations up to quadratic order in the affine parameter $\lambda$. We begin by finding the parameterization for a static rib with initial condition $\theta(0) = \theta_0$. The ode to be solved is
\be
L^2 \cot^2{\theta} + \frac{\sec^2{\theta}}{f(\theta)}\left(\frac{\dd \theta}{\dd \lambda}\right)^2 = 1
\ee
putting the ansatsz $\theta(\lambda) = \theta_0 + \theta_1\lambda + \theta_2 \lambda^2$ we find
\be
\theta(\lambda) = \arctan{L} + \frac{f(\theta_0)}{2L}\lambda^2
\ee
notice that for $\delta = 0$ we correctly recover the rib from empty AdS$_3$. Then we have
\be
\frac{\dd t}{\dd \lambda} = \frac{E\, \cos^2{\theta}}{f(\theta)}, \quad\frac{\dd \phi}{\dd \lambda} = \frac{L}{\tan^2{\theta}}
\ee
which have solutions
\be
t(\lambda) = \frac{E\cos^2{\theta_0}}{f(\theta_0)}\lambda, \quad \phi(\lambda) = L\cot^2{\theta_0}\lambda.
\ee
From these, following similar steps to the empty AdS case, we can get the parameterization for the geodesic seam
\be
t(\lambda) = \frac{L\cos^2{\theta_0}}{f(\theta_0)}\lambda, \quad\theta(\lambda) = \theta_0 -\frac{L^3}{2(1+L^2)^2}\frac{1-\frac{\delta}{2}\cos^4{\theta_0}}{f(\theta_0)}\lambda^2, \quad \phi(\lambda) = \frac{E\cos^2{\theta_0}}{f(\theta_0)}\lambda.
\ee
\bibliographystyle{jhep}
\bibliography{references}
\end{document}